\documentclass[english]{article}
\usepackage[T1]{fontenc}
\usepackage[latin9]{inputenc}
\usepackage{geometry}
\usepackage{verbatim}
\usepackage{mathtools}
\usepackage{bm}
\usepackage{amsmath}
\usepackage{amssymb}
\usepackage{graphicx}
\usepackage{esint}

\makeatletter

\providecommand{\tabularnewline}{\\}

\usepackage{amsthm}\usepackage{dsfont}\usepackage{array}\usepackage{mathrsfs}\usepackage{cite}\usepackage{comment}\usepackage{mathrsfs}\usepackage{hyperref,enumerate}

\usepackage[shortlabels]{enumitem}
\makeatother

\usepackage{babel}
\usepackage{color}
\usepackage{centernot}

\usepackage{babel}
\usepackage{sansmath}

\newcommand{\mincut}{{\mathsf{mincut}}}

\definecolor{yxc}{RGB}{0,0,200}

\date{April 2015; ~Revised May 2016}

\makeatother

\usepackage{babel}
\begin{document}

\title{Information Recovery from Pairwise Measurements}

\author{Yuxin Chen \and Changho Suh \and Andrea J. Goldsmith \thanks{Y. Chen is with the Department of Statistics, Stanford University,
Stanford, CA 94305, USA (email: yxchen@stanford.edu). C. Suh is with
the Department of Electrical Engineering, Korea Advanced Institute
of Science and Technology, Daejeon 305-701, Korea (e-mail: chsuh@kaist.ac.kr).
A. J. Goldsmith is with the Department of Electrical Engineering,
Stanford University, Stanford, CA 94305, USA (email: andrea@wsl.stanford.edu).
This paper has been presented in part at \cite{Chen2014Pairwise,chen2015information}.}}

\maketitle
\theoremstyle{plain}\newtheorem{lem}{\textbf{Lemma}}\newtheorem{theorem}{\textbf{Theorem}}\newtheorem{corollary}{\textbf{Corollary}}\newtheorem{prop}{\textbf{Proposition}}\newtheorem{fact}{\textbf{Fact}}

\theoremstyle{definition}\newtheorem{definition}{\textbf{Definition}}\newtheorem{example}{\textbf{Example}}

\theoremstyle{remark}\newtheorem{remark}{\textbf{Remark}}
\begin{abstract}
This paper is concerned with jointly recovering $n$ node-variables
$\left\{ x_{i}\right\} _{1\leq i\leq n}$ from a collection of pairwise
difference measurements. Imagine we acquire a few observations taking
the form of $x_{i}-x_{j}$; the observation pattern is represented
by a measurement graph $\mathcal{G}$ with an edge set $\mathcal{E}$
such that $x_{i}-x_{j}$ is observed if and only if $(i,j)\in\mathcal{E}$.
To account for noisy measurements in a general manner, we model the
data acquisition process by a set of channels with given input/output
transition measures. Employing information-theoretic tools applied
to channel decoding problems, we develop a \emph{unified} framework
to characterize the fundamental recovery criterion, which accommodates
general graph structures, alphabet sizes, and channel transition measures.
In particular, our results isolate a family of \emph{minimum} \emph{channel
divergence measures} to characterize the degree of measurement corruption,
which together with the size of the minimum cut of $\mathcal{G}$
dictates the feasibility of exact information recovery. For various
homogeneous graphs, the recovery condition depends almost only on
the edge sparsity of the measurement graph irrespective of other graphical
metrics; alternatively, the minimum sample complexity required for
these graphs scales like 
\[
\text{minimum sample complexity }\asymp\frac{n\log n}{\mathsf{Hel}_{1/2}^{\min}}
\]
for certain information metric $\mathsf{Hel}_{1/2}^{\min}$ defined
in the main text, as long as the alphabet size is not super-polynomial
in $n$. We apply our general theory to three concrete applications,
including the stochastic block model, the outlier model, and the haplotype
assembly problem. Our theory leads to order-wise tight recovery conditions
for all these scenarios. 
\end{abstract}
\textbf{$\quad$Index Terms}: pairwise difference, information divergence,
random graphs, geometric graphs, homogeneous graphs

\section{Introduction\label{sec:Introduction}}

In various data processing scenarios, one wishes to acquire information
about a large collection of objects, but it is infeasible or difficult
to directly measure each individual object in isolation. Instead,
only certain pairwise relations over a few object pairs can be measured.
Partial examples of pairwise relations include cluster agreements,
relative rotation and translation, pairwise matches, and paired sequencing
reads, as will be discussed in details later. Taken collectively,
these pairwise observations often carry a substantial amount of information
across all objects of interest. As a consequence, reliable joint information
recovery becomes feasible as soon as a sufficiently large number of
pairwise measurements are obtained. 

\begin{figure}
\centering{}\includegraphics[scale=0.35]{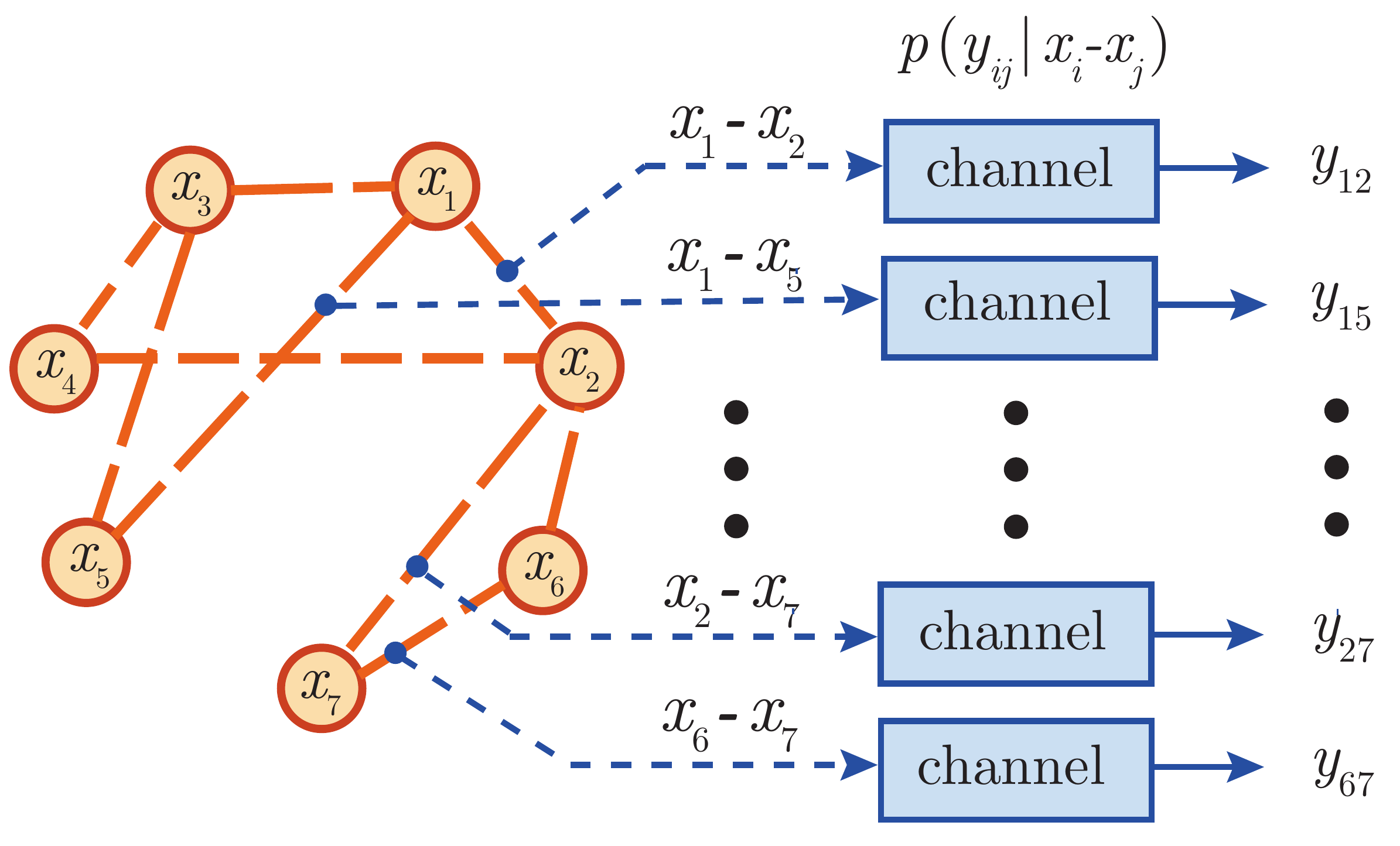}\caption{Measurement graph and equivalent channel model. For each edge $(i,j)$
in the measurement graph, $x_{i}-x_{j}$ is independently passed through
a channel with output $y_{ij}$ and transition probability $p\left(y_{ij}\mid x_{i}-x_{j}\right)$.
\label{fig:Equivalent-channel-model}}
\end{figure}

This paper explores a large family of pairwise measurements, which
we term \emph{pairwise difference measurements}. Consider $n$ variables
$x_{1},\cdots,x_{n}$, and imagine we obtain independent measurements
of the differences\footnote{Here, ``$-$'' represents some algebraic subtraction operation (broadly
defined), as we detail in Section \ref{sec:Problem-Formulation}.} $x_{i}-x_{j}$ over a few pairs ($i,j$). This pairwise difference
functional is represented by a measurement graph $\mathcal{G}$ with
an edge set $\mathcal{E}$ such that $x_{i}-x_{j}$ is observed if
and only if $(i,j)\in\mathcal{E}$. To accommodate the noisy nature
of data acquisition in a general manner, we model the observations
$\left\{ y_{ij}\right\} $ as the output of the following channel:
\begin{equation}
x_{i}-x_{j}\quad\xrightarrow{p\left(y_{ij}\mid x_{i}-x_{j}\right)}\quad y_{ij},\quad\forall(i,j)\in\mathcal{E},\label{eq:model}
\end{equation}
as illustrated in Fig.~\ref{fig:Equivalent-channel-model}. Here,
the output distribution is specified solely by the associated channel
input $x_{i}-x_{j}$, with $p\left(\cdot\mid\cdot\right)$ representing
the channel transition probability. The goal is to recover $\boldsymbol{x}=\left\{ x_{1},\cdots,x_{n}\right\} $
based on these channel outputs $\left\{ y_{ij}\right\} $. Note that
for any connected graph $\mathcal{G}$, the ground truth $\boldsymbol{x}$
is uniquely determined by the pairwise difference functional $\left\{ x_{i}-x_{j}\mid(i,j)\in\mathcal{E}\right\} $,
up to some global offset. Therefore, the problem can alternatively
be posed as decoding the input of the channel (\ref{eq:model}) based
on $\left\{ y_{ij}\right\} $. 

Problems of this kind have received considerable attention across
various fields like social networks, computer science, and computational
biology. A small sample of them are listed as follows.
\begin{itemize}
\item \emph{Community detection and graph partitioning}. Various real-world
networks exhibit community structures \cite{fortunato2010community},
and the nodes are grouped into a few clusters based on shared features.
The aim is to uncover the hidden community structure by observing
the similarities between members. For instance, in the simplest two-community
model, the vertex-variables represent the community assignment, and
the edge variables encode whether two vertices belong to the same
community. This two-community problem, sometimes referred to as graph
partitioning (e.g.~\cite{jalali2011clustering,chen2014improved}),
is a special instance of the pairwise difference model. 
\item \emph{Alignment, registration and synchronization}. Consider $n$
views of a single scene from different angles and positions\footnote{In a variety of applications including structure from motion and cryo-EM,
these views (e.g.~photos of some architectures, or projected images
of 3D molecules) are given to us without revealing their absolute
camera poses / angles with respect to the 3D structure of interest. }. One is allowed to estimate the\emph{ relative} translation~/~rotation
across several pairs of views. The problem aims at simultaneously
aligning all views based on these noisy pairwise estimates. This arises
in many applications including structure from motion in computer vision
\cite{crandall2011discrete,zach2010disambiguating}, spectroscopy
imaging and structural biology \cite{cucuringu2012eigenvector,wang2012exact},
and multi-reference alignment \cite{bandeira2014multireference}. 
\item \emph{Joint matching}. Given $n$ images~/~shapes representing the
same physical object, one wishes to identify common features across
them. The input to a cutting-edge joint matching paradigm is typically
a set of noisy pairwise matches computed between several pairs of
images in isolation \cite{zach2010disambiguating,chen2013matching,pachauri2013solving,huang2013consistent,yan2014graduated},
which falls under the category of pairwise difference measurements.
The goal is to recover globally consistent maps across the features
of all images, by refining these noisy pairwise inputs. This problem
arises in numerous applications in computer vision and graphics, solving
jigsaw puzzles, etc. 
\item \emph{Genome assembly}. The genomes of two unrelated people mostly
differ at specific nucleotide positions called single nucleotide polymorphisms
(SNPs). A haplotype is a collection of associated SNPs on a chromatid,
which is important in understanding genetic causes of various diseases
and developing personalized medicine. Among various sequencing methods,
haplotype assembly is particularly effective from paired sequencing
reads \cite{browning2011haplotype,he2010optimal,donmez2011hapsembler},
which amounts to reconstructing the haplotype based on disagreement
between pairs of single reads \cite{chen2016community,si2014haplotype,kamath2015optimal}
--- a special instance of pairwise difference measurement with binary
alphabet.
\end{itemize}
Many of these practical applications have witnessed a flurry of recent
activity in algorithm development, which are primarily motivated by
computational considerations. For instance, inspired by recent success
in spectral methods \cite{keshavan2010mc,keshavan2010matrix} and
convex relaxation \cite{ExactMC09,CanLiMaWri09,chandrasekaran2011rank}
(particularly those developed for low-rank matrix recovery problems),
many provably efficient algorithms have been proposed for graph clustering
\cite{jalali2011clustering}, joint matching \cite{huang2013consistent,chen2013matching},
synchronization \cite{wang2012exact}, and so on. While these algorithms
have been shown to enjoy intriguing recovery guarantees under simple
randomized models, the choices of performance metrics have mainly
been studied in a model-specific manner. On the fundamental-limit
side, there have been several results in place for a few applications,
e.g.~stochastic block models \cite{abbe2014exact,mossel2014consistency},
synchronization \cite{Chen2014Pairwise}, and haplotype assembly \cite{kamath2015optimal,si2014haplotype,chen2016community}.
Despite their intrinsic connections, these results were developed
primarily on a case-by-case basis instead of accounting for the most
general observation models. 

In the present paper, we emphasize the similarities and connections
among all these motivating applications, by viewing them as a graph-based
functional fed into a collection of general channels. We wish to explore
the following questions from an information-theoretic perspective: 
\begin{enumerate}
\item Are there any distance metrics of the channel transition measures
and graphical properties that dictate the success of exact information
recovery from pairwise difference measurements? 
\item If so, can we characterize the interplay between these channel separation
metrics and graphical constraints and provide insights into the feasibility
of simultaneous recovery? 
\end{enumerate}
All in all, the aim of this work is to gain a \emph{unified} understanding
about the performance limits that underlie various applications falling
in the realm of pairwise-measurement based recovery. In turn, these
fundamental criteria will provide a general benchmark for algorithm
evaluation and comparison.

\subsection{Main Contributions}

The main contribution of this paper is towards a unified characterization
of the fundamental information recovery criterion, using both information-theoretic
and graph-theoretic tools. In particular, we single out and emphasize
a family of minimum channel separation measures (i.e.~the minimum
Kullback\textendash Leibler (KL), Hellinger, and Rényi divergence),
as well as two graphical metrics (i.e.~the minimum cut size and the
cut-homogeneity exponent defined in Section \ref{sub:Key-Graphical-Metrics}),
that play central roles in determining the feasibility of exact recovery.
Equipped with these metrics, we develop a sufficient and a necessary
condition for information recovery, which apply to general graphs,
any type of input alphabets, and general channel transition measures.
Encouragingly, as long as the alphabet size is not super-polynomial
in $n$, these two conditions coincide (modulo some explicit universal
constant) for the broad class of homogeneous graphs, subsuming as
special cases Erd\H{o}s\textendash Rényi models, homogeneous geometric
graphs (e.g.~generalized rings and grids), and many other expander
graphs. 

In a nutshell, the fundamental recovery criterion is specified by
the product of the minimum channel divergence measures and the size
of the minimum cut. Intuitively, this product characterizes the amount
of information one has available to differentiate two minimally separated
input hypotheses. Somewhat surprisingly, for a variety of homogeneous
graphs, the recovery criterion relies only on the edge sparsity of
the measurement graph. Equivalently, the minimum sample complexity
required for exact recovery in these homogeneous graphs scales as
\[
\text{minimum sample complexity }\asymp\frac{n\log n}{\mathsf{Hel}_{1/2}^{\min}}
\]
for some information metric $\mathsf{Hel}_{1/2}^{\min}$ to be specified
later, provided that the alphabet size is polynomial in $n$. This
result holds irrespective of other second-order graphical metrics
like the spectral gap. 

The unified framework we develop is \emph{non-asymptotic}, in the
sense that it accommodates the most general settings without fixing
either the alphabet size or channel transition probabilities. This
allows full characterization of the high-dimensional regime where
all parameters are allowed to scale (possibly with different rates)
--- a setting that has received increasing attention compared to the
classical asymptotics where only $n$ is tending to infinity.

Finally, to illustrate the effectiveness of our general theory, we
develop concrete consequences for three canonical applications that
have been investigated in prior literature, including the stochastic
block model, the outlier model, and the haplotype assembly problem.
In each case, our theory recovers order-wise correct recovery guarantees,
and even strengthens existing results in certain regimes.

\subsection{Related Work\label{sub:Related-Work}}

On the information-theoretic side, most prior works focused on \emph{binary
input and output alphabets}. Among them, Abbe et al.~\cite{abbe2014decoding}
characterized the orderwise information-theoretic limits under the
Erd\H{o}s\textendash Rényi model, uncovering the intriguing observation
that a decoding method based on convex relaxation achieves nearly-optimal
recover guarantees under sparsely connected graphs. In addition, Si
et al. \cite{si2014haplotype} and Kamath et al.~\cite{kamath2015optimal}
determined the information-theoretic limits for a similar setup motivated
from genome sequencing, which correspond to random graphs and (generalized)
ring graphs, respectively. A sufficient recovery condition for general
graphs has also been derived in \cite{abbe2014decoding}, although
it was not guaranteed to be order optimal. Our preliminary work \cite{Chen2014Pairwise}
explored the fundamental recovery limits under general alphabets and
graph structures, but was restricted to the simplistic outlier model
rather than general channel distributions. In contrast, the framework
developed in the current work allows orderwise tight characterization
of the recovery criterion for general alphabets and channel characteristics. 

The pairwise measurement models considered in this paper and the aforementioned
works \cite{abbe2014decoding,si2014haplotype,kamath2015optimal,Chen2014Pairwise}
can all be treated as a special type of \emph{``}graphical channel''
as coined by Abbe and Montanari \cite{abbe2013conditional,Abbe13Allerton},
which refers to a general family of channels whose transition probabilities
factorize over a set of hyper-edges. This previous work on graphical
channels centered on the metric of conditional entropy that quantifies
the residual input uncertainty given the channel output, and uncovered
the stability and concentration of this metric under random sparse
graphs. In comparison, the present paper primarily aims to investigate
how the channel transition measures affect the recovery limits in
the absence of channel coding, which was previously out of reach.
Specifically, the information limit under optimal channel coding is
determined by the mutual information metric; in contrast, the information
limit without channel coding is often dictated by certain minimum
divergence metrics, which could sometimes be much smaller than the
mutual information. This arises because optimal encoding enables us
to code against the channel variation by maximizing the output separation
between distinct input hypotheses, while in the non-coding applications
one has to deal with the minimally separated input hypotheses determined
by the practical applications. In addition, we focus on full recovery
in this work, but in some applications this might be too stringent.
Recent interesting work \cite{globerson2014tight,globersonhard2015}
explored the notion of \emph{partial recovery} under binary alphabets,
which highlighted the two-dimensional grids and supplied a two-step
polynomial-time recovery algorithm. A more general theory regarding
partial or approximate recovery is left for future work.

Finally, the input variables $\left\{ x_{i}\right\} $ can be viewed
as discrete signals on the graph $\mathcal{G}$. Recent years have
seen much activity regarding discrete signal processing on graphs
\cite{shuman2013emerging,sandryhaila2013discrete}. For instance,
it has been studied in \cite{chen2015discrete} how to optimally subsample
band-limited graphs signals, subject to a sampling rate constraint,
while enabling perfect signal recovery. Our model differs from this
line of work in that the samples we take are highly constrained---that
is, we only allow pairwise difference samples taken over the edges---and
hence the resulting sample complexity significantly exceeds the sampling
rate limit.

\subsection{Terminology and Notation\label{sub:Notation}}

\textbf{Graph terminology}. Let $\mathrm{deg}\left(v\right)$ represent
the degree of a vertex $v$. For any two vertex sets $\mathcal{S}_{1}$
and $\mathcal{S}_{2}$, denote by $\mathcal{E}(\mathcal{S}_{1},\mathcal{S}_{2})$
(resp.~$e(\mathcal{S}_{1},\mathcal{S}_{2})$) the set (resp.~the
number) of edges with exactly one endpoint in $\mathcal{S}_{1}$ and
another in $\mathcal{S}_{2}$. A complete graph of $n$ vertices,
denoted by $K_{n}$, is a graph in which every pair of vertices is
connected by an edge. Below we introduce several widely used (random)
graph models; see \cite{durrett2007random,penrose2003random} and
the references therein for in-depth discussion. 
\begin{enumerate}
\item \emph{Erd\H{o}s\textendash Rényi graph}. An Erd\H{o}s\textendash Rényi
graph of $n$ vertices, denoted by $\mathcal{G}_{n,p}$, is constructed
in such a way that each pair of vertices is independently connected
by an edge with probability $p$. 
\item \emph{Random geometric graph}. A random geometric graph, denoted by
$\mathcal{G}_{n,r}$, is generated via a 2-step procedure: (i) place
$n$ vertices uniformly and independently on the surface of a unit
sphere\footnote{We consider $\mathcal{G}_{n,r,}$ on a unit sphere instead of $[0,1]^{2}$
to eliminate edge effects.}; (ii) connect two vertices by an edge if the Euclidean distance between
them is at most $r$.
\item \emph{Expander graph}. A graph $\mathcal{G}$ is said to be an expander
graph with edge expansion $h_{\mathcal{G}}$ if $e\left(\mathcal{S},\mathcal{S}^{\mathrm{c}}\right)\geq h_{\mathcal{G}}\left|\mathcal{S}\right|$
for all vertex set $\mathcal{S}$ satisfying $|\mathcal{S}|\leq n/2$.
\end{enumerate}
\textbf{Divergence measures}. Our results are established upon a family
of divergence measures. Formally, for any two probability measures
$P$ and $Q$, if $P$ is absolutely continuous with respect to $Q$,
then the KL divergence of $Q$ from $P$ is defined as
\begin{equation}
\mathsf{KL}\left(P\hspace{0.2em}\|\hspace{0.2em}Q\right):={\displaystyle \int}\mathrm{d}P\log\left(\frac{\mathrm{d}P}{\mathrm{d}Q}\right),\label{eq:KL-div}
\end{equation}
whereas the Hellinger divergence of order $\alpha\in\left(0,1\right)$
of $Q$ from $P$ is defined to be \cite{liese2006divergences,topsoe2000some}
\begin{equation}
\mathsf{Hel}_{\alpha}\left(P\hspace{0.2em}\|\hspace{0.2em}Q\right):=\frac{1}{1-\alpha}\left[1-{\displaystyle \int}\left(\mathrm{d}P\right)^{\alpha}\left(\mathrm{d}Q\right)^{1-\alpha}\right].\label{eq:Hellinger-alpha}
\end{equation}
When $\alpha=1/2$, this reduces to the so-called \emph{squared Hellinger
distance}\footnote{Several other sources introduce a prefactor of $1/2$ in order to
normalize the squared Hellinger distance, resulting in the definition
$\int\frac{1}{2}(\sqrt{\mathrm{d}P}-\sqrt{\mathrm{d}Q})^{2}$. Here,
we adopt the unnormalized version as given in \cite[Section 2.4]{tsybakov2009introduction}. }
\begin{equation}
\mathsf{Hel}_{\frac{1}{2}}\left(P\hspace{0.2em}\|\hspace{0.2em}Q\right)=2-2{\displaystyle \int}\sqrt{\mathrm{d}P}\sqrt{\mathrm{d}Q}={\displaystyle \int}\big(\sqrt{\mathrm{d}P}-\sqrt{\mathrm{d}Q}\big)^{2}.\label{eq:Hellinger}
\end{equation}
The $\chi^{2}$ divergence is defined as 
\begin{equation}
\chi^{2}\left(P\hspace{0.2em}\|\hspace{0.2em}Q\right)={\displaystyle \int}\left(\frac{\mathrm{d}P}{\mathrm{d}Q}-1\right)^{2}\mathrm{d}Q.\label{eq:ChiSquare}
\end{equation}
In particular, when $ $$P=\mathsf{Bernoulli}\left(p\right)$ and
$Q=\mathsf{Bernoulli}\left(q\right)$, we abuse the notation and let
\begin{equation}
\mathsf{KL}\left(p\hspace{0.2em}\|\hspace{0.2em}q\right)=\mathsf{KL}\left(P\hspace{0.2em}\|\hspace{0.2em}Q\right),\quad\mathsf{Hel}_{\alpha}\left(p\hspace{0.2em}\|\hspace{0.2em}q\right)=\mathsf{Hel}_{\alpha}\left(P\hspace{0.2em}\|\hspace{0.2em}Q\right),\quad\text{and}\quad\chi^{2}\left(p\hspace{0.2em}\|\hspace{0.2em}q\right)=\chi^{2}\left(P\hspace{0.2em}\|\hspace{0.2em}Q\right).\label{eq:Bernoulli}
\end{equation}
More generally, the $f$-divergence of $Q$ from $P$ is defined as
\begin{equation}
D_{f}\left(P\hspace{0.2em}\|\hspace{0.2em}Q\right):={\displaystyle \int}f\left(\frac{\mathrm{d}P}{\mathrm{d}Q}\right)\mathrm{d}Q\label{eq:F-div}
\end{equation}
for any convex function $f\left(\cdot\right)$ such that $f(1)=0$
\cite{liese2006divergences,topsoe2000some}. Note that the Hellinger
divergence of order $\alpha$, the KL divergence, and the $\chi^{2}$
divergence are special cases of $f$-divergence generated by $f(x)=\frac{1}{1-\alpha}(1-x^{\alpha})$,
$f(x)=x\log x$ (or $f(x)=x\log x-x+1$), and $f\left(x\right)=\left(x-1\right)^{2}$,
respectively. These divergence measures can often be efficiently estimated
even under large alphabets; see, e.g., \cite{jiao2014minimax} and
their subsequent work.

Finally, we introduce the Rényi divergence of positive order $\alpha$,
where $\alpha\neq1$, of a distribution $P$ from another distribution
$Q$ as \cite{rrnyi1961measures,van2014renyi}
\begin{eqnarray}
D_{\alpha}\left(P\hspace{0.2em}\|\hspace{0.2em}Q\right): & = & -\frac{1}{1-\alpha}\log\left({\displaystyle \int}\left(\mathrm{d}P\right)^{\alpha}\left(\mathrm{d}Q\right)^{1-\alpha}\right)\label{eq:defn-RenyiDiv}\\
 & = & -\frac{1}{1-\alpha}\log\left(1-\left(1-\alpha\right)\mathsf{Hel}_{\alpha}\right).\label{eq:defn-Renyi-Hel}
\end{eqnarray}
It follows from the elementary inequality $1-x\leq e^{-x}$ that $D_{\alpha}\left(P\hspace{0.2em}\|\hspace{0.2em}Q\right)\geq\mathsf{Hel}_{\alpha}\left(P\hspace{0.2em}\|\hspace{0.2em}Q\right)$.
This together with the monotonicity of $D_{\alpha}$ \cite[Theorem 3]{van2014renyi}
gives 
\begin{equation}
\mathsf{Hel}_{\alpha}\left(P\hspace{0.2em}\|\hspace{0.2em}Q\right)\leq D_{\alpha}\left(P\hspace{0.2em}\|\hspace{0.2em}Q\right)\leq\mathsf{KL}\left(P\hspace{0.2em}\|\hspace{0.2em}Q\right),\quad0<\alpha<1.\label{eq:Hel-D-KL}
\end{equation}

\vspace{0.7em}

\noindent \textbf{Other notation}. Let ${\bf 1}$ and ${\bf 0}$
be the all-one and all-zero vectors, respectively. We denote by $\mathrm{supp}\left(\boldsymbol{x}\right)$
(resp.~$\left\Vert \boldsymbol{x}\right\Vert _{0}$) the support
(resp.~the support size) of $\boldsymbol{x}$. The standard notion
$f(n)=o\left(g(n)\right)$ means $\underset{n\rightarrow\infty}{\lim}f(n)/g(n)=0$;
$f(n)=\omega\left(g(n)\right)$ means $\underset{n\rightarrow\infty}{\lim}g(n)/f(n)=0$;
$f(n)=\Omega\left(g(n)\right)$ or $f(n)\gtrsim g(n)$ mean there
exists a constant $c$ such that $f(n)\geq cg(n)$; $f(n)=\mathcal{O}\left(g(n)\right)$
or $f(n)\lesssim g(n)$ mean there exists a constant $c$ such that
$f(n)\leq cg(n)$; $f(n)=\Theta\left(g(n)\right)$ or $f(n)\asymp g(n)$
mean there exist constants $c_{1}$ and $c_{2}$ such that $c_{1}g(n)\leq f(n)\leq c_{2}g(n)$.
Throughout this paper, $\log\left(\cdot\right)$ represents the natural
logarithm.

\subsection{Organization\label{sub:Organization}}

The remainder of the paper is organized as follows. In Section \ref{sec:Problem-Formulation},
we describe the formal problem setup and introduce the key channel
distance measures. We develop non-asymptotic sufficient and necessary
recovery conditions for the special Erd\H{o}s\textendash Rényi model
in Section \ref{sec:Erd=000151s=002013R=0000E9nyi-Graphs}, along
with some intuitive interpretation of the results. Section \ref{sec:General-Graphs}
presents the recovery conditions in full generality, which accommodate
general alphabets, graph structures, and channel characteristics,
with particular emphasis on the family of homogeneous graphs. To illustrate
the effectiveness of our framework, we apply our general theory to
a few concrete examples in Section \ref{sec:Applications}. Section
\ref{sec:Concluding-Remarks} concludes the paper with a summary of
our findings and a discussion of future directions. The proofs of
the main results and auxiliary lemmas are deferred to the appendices.

\section{Problem Formulation and Key Metrics\label{sec:Problem-Formulation}}

\subsection{Models}

Imagine a collection of $n$ vertices $\mathcal{V}=\left\{ 1,\cdots,n\right\} $,
each represented by a vertex-variable $x_{i}$ over the\emph{ }input
alphabet $\mathcal{X}:=\left\{ 0,1,\cdots,M-1\right\} $, where $M$
represents the alphabet size. 
\begin{itemize}
\item \textbf{Object representation and pairwise difference. }Consider an
additive group formed over $\mathcal{X}$ together with an associative
addition operation ``+'' (broadly defined). For any $x_{i},x_{j}\in\mathcal{X}$,
the pairwise difference operation is defined as
\begin{equation}
x_{i}-x_{j}\text{ }:=\text{ }x_{i}+(-x_{j}),\label{eq:PairwiseDiff}
\end{equation}
where $-x$ stands for the unique additive inverse of $x$. We assume
throughout that ``$+$'' satisfies the following bijective property:
\begin{equation}
\forall x_{i}\in\mathcal{X}:\quad\begin{cases}
x_{i}+x_{j}\neq x_{i}+x_{l},\quad & \forall x_{l}\neq x_{j};\\
x_{i}+x_{j}\neq x_{l}+x_{j}, & \forall x_{l}\neq x_{i}.
\end{cases}\label{eq:CrossCutAssumption-alpha}
\end{equation}
A partial list of examples includes:

\begin{enumerate}
\item \emph{Modular arithmetic}: if we define ``$+$'' to be the modular
addition over integers $\left\{ 0,1,\cdots,M-1\right\} $, then $x_{i}-x_{j}$
($\mathsf{mod}$ $M$) is a valid example of (\ref{eq:PairwiseDiff}).
\item \emph{Relative rotation}: set $x_{i}=\boldsymbol{R}_{i}$ for some
rotation matrix $\boldsymbol{R}_{i}$ and let ``$+$'' denote matrix
multiplication. Then $x_{i}-x_{j}$ stands for $\boldsymbol{R}_{i}\boldsymbol{R}_{j}^{-1}$,
which represents the relative rotation between $i$ and $j$, and
hence is a special case of (\ref{eq:PairwiseDiff}). 
\item \emph{Pairwise map}: if we set $x_{i}$ to be some permutation matrix
$\boldsymbol{\Pi}_{i}$ and let ``$+$'' be matrix multiplication,
then the pairwise map between two isomorphic sets---captured by $\boldsymbol{\Pi}_{i}\boldsymbol{\Pi}_{j}^{\top}$---also
belongs to the pairwise difference model. 
\end{enumerate}
\end{itemize}
\begin{itemize}[listparindent =1em] \item \textbf{Measurement graph
and channel model}. The measurement pattern is represented by a \emph{measurement
graph} $\mathcal{G}$ that comprises an undirected edge set $\mathcal{E}$,
so that $x_{i}-x_{j}$ is measured if and only if $(i,j)\in\mathcal{E}$.
As illustrated in Fig.~\ref{fig:Equivalent-channel-model}, for each
$\left(i,j\right)\in\mathcal{E}$ ($i>j$), the pairwise difference
$x_{i}-x_{j}$ is independently passed through a channel, whose output
$y_{ij}$ follows the conditional distribution
\begin{equation}
p\left(y_{ij}\Big|x_{i}-x_{j}=l\right)=\mathbb{P}_{l}\left(y_{ij}\right),\quad0\leq l<M.\label{eq:channel-model}
\end{equation}
Here, $\mathbb{P}_{l}\left(\cdot\right)$ denotes the transition measure
that maps a given input $l$ to the output alphabet $\mathcal{Y}$;
see Fig.~\ref{fig:defn-Pl} for an illustration. With a slight abuse
of notation, we let $\mathbb{P}_{i}=\mathbb{P}_{i\text{ \ensuremath{\mathsf{mod}\text{ }}}M}$
for any $i\notin\left\{ 0,1,\cdots,M-1\right\} $. We assume throughout
that the observations are symmetric\footnote{We assume the observation model is symmetric because this is the case
in all motivating applications listed in Section \ref{sec:Introduction}.
We note, however, that all results and analyses immediately extend
to the non-asymmetric case, provided that $\mathbb{P}_{l}\left(\cdot\right)$
is defined with respect to $(y_{ij},y_{ji})$, that is, $\mathbb{P}_{l}(y_{ij},y_{ji}):=p(y_{ij},y_{ji}\Big|x_{i}-x_{j}=l)$.} in the sense that there exists a one-to-one mapping between $y_{ij}$
and $y_{ji}$ for any $(i,j)\in\mathcal{E}$; that said, all information
are contained in the upper triangular part $\left\{ y_{ij}\right\} _{1\leq i<j\leq n}$.
The output alphabet $\mathcal{Y}$ can be either continuous or discrete,
finite or infinite, which allows general modeling of distortion, corruption,
etc. As opposed to conventional information theory settings, no coding
is employed across channel uses. \end{itemize}

\begin{figure}
\centering{}\includegraphics[width=0.7\textwidth]{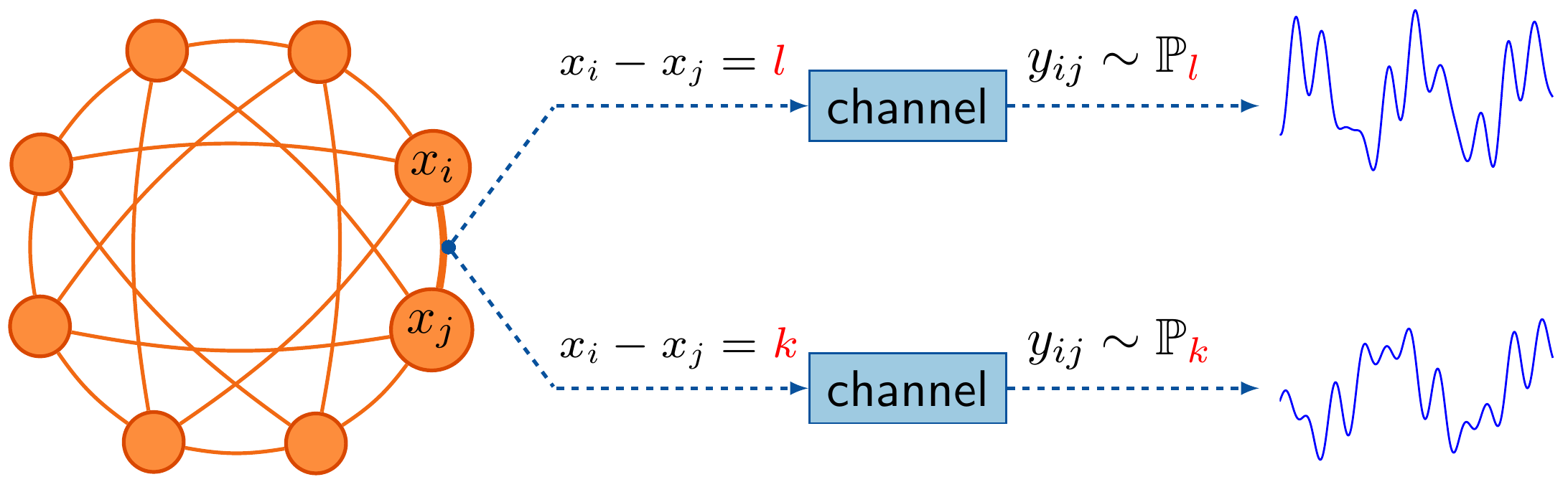}\caption{The probability measure $\mathbb{P}_{l}(\cdot)$ is defined to be
the distribution of $y_{ij}$ given $x_{i}-x_{j}=l$. \label{fig:defn-Pl}}
\end{figure}

This paper centers on exact information recovery, that is, to reconstruct
all input variables $\boldsymbol{x}=\left\{ x_{1},\cdots,x_{n}\right\} $
precisely, except for some \emph{global} \emph{offset}. This is all
one can hope for since there is absolutely no basis to distinguish
$\boldsymbol{x}$ from its shifted version $\boldsymbol{x}+l\cdot{\bf 1}=\left\{ x_{1}+l,\cdots,x_{n}+l\right\} $
given only the output $\boldsymbol{y}:=\{y_{ij}\mid\left(i,j\right)\in\mathcal{E}\}$.
In light of this, we introduce the zero-one distance modulo a global
offset factor as follows
\begin{eqnarray}
\mathrm{dist}\left(\boldsymbol{w},\boldsymbol{x}\right) & := & 1-\max_{0\leq l<M}\mathbb{I}\left\{ \boldsymbol{w}=\boldsymbol{x}+l\cdot{\bf 1}\right\} ,\label{eq:Defn-dist}
\end{eqnarray}
where $\mathbb{I}$ is the indicator function. Apparently, $\mathrm{dist}\left(\boldsymbol{w},\boldsymbol{x}\right)=0$
holds for all $\boldsymbol{w}$ that differ from $\boldsymbol{x}$
only by a global offset. With this metric in place, we define, for
any recovery procedure $\psi:\text{ }\mathcal{Y}^{\left|\mathcal{E}\right|}\mapsto\mathcal{X}^{n}$,
the \emph{probability of error} as
\begin{eqnarray}
P_{\mathsf{e}}\left(\psi\right) & := & \max_{\boldsymbol{x}\in\mathcal{X}^{n}}\mathbb{P}\left\{ \mathrm{dist}\left(\psi\left(\boldsymbol{y}\right),\boldsymbol{x}\right)\neq0\text{ }\Big|\text{ }\boldsymbol{x}\right\} .\label{eq:Prob-Error}
\end{eqnarray}
The aim is to characterize the regime where the minimax probability
of error $\inf_{\psi}P_{\mathsf{e}}\left(\psi\right)$ is vanishing.

\subsection{Key Separation Metrics on Channel Transition Measures\label{sec:Key-Metrics}}

Before proceeding to the main results, we introduce a few channel
separation measures that capture the resolutions of the measurements,
which will be critical in subsequent development of our theory. Specifically,
we isolate the minimum KL, Hellinger, and Rényi divergence with respect
to the channel transition measures as follows\footnote{Here and throughout, we assume that $\mathbb{P}_{l}$ is absolutely
continuous with respect to $\mathbb{P}_{k}$ for any $l$ and $k$. }
\begin{eqnarray}
\mathsf{KL}^{\min} & := & \min_{l\neq k}\mathsf{KL}\left(\mathbb{P}_{l}\hspace{0.2em}\|\hspace{0.2em}\mathbb{P}_{k}\right);\label{eq:def-KL_min}\\
\mathsf{Hel}_{\alpha}^{\min} & := & \min_{l\neq k}\mathsf{Hel}_{\alpha}\left(\mathbb{P}_{l}\hspace{0.2em}\|\hspace{0.2em}\mathbb{P}_{k}\right);\label{eq:DefnHel_min}\\
D_{\alpha}^{\min} & := & \min_{l\neq k}D_{\alpha}\left(\mathbb{P}_{l}\hspace{0.2em}\|\hspace{0.2em}\mathbb{P}_{k}\right)=-\frac{1}{1-\alpha}\log\left(1-\left(1-\alpha\right)\mathsf{Hel}_{\alpha}^{\min}\right).\label{eq:DefnRenyi-min}
\end{eqnarray}
These minimum divergence measures essentially reflect the distinguishability
of channel outputs given minimally separated inputs\footnote{One natural question arises as so how to estimate such divergence
metrics from measured data, which has become an active research topic.
When both the input and output alphabet sizes are small, one can first
estimate the entire probability measures via low-rank matrix recovery
schemes (e.g.~\cite{huang2016recovering}), and then plug them in
to calculate the divergence metrics. When the alphabet size is large
or when the output is continuous-valued, one might resort to more
careful functional estimation algorithms (e.g.~\cite{jiao2014minimax,kraskov2004estimating}).}. As will be seen later, the minimum Hellinger and Rényi divergence
are crucial in developing sufficient recovery conditions, while the
minimum KL divergence plays an important role in deriving minimax
lower bounds. It is well known (see \cite{sason2015bounds,tsybakov2009introduction,liese2006divergences,topsoe2000some}
for various inequalities connecting them) that these measures are
almost equivalent (modulo some small constant) when any two probability
measures under study are close to each other---a regime where two
measures are the hardest to differentiate. In particular, we underscore
one fact that links the KL divergence and the squared Hellinger distance,
which we shall use several times in the rest of the paper; see \cite[Proposition 2]{dragomir2002upper}
for an alternative version. 

\begin{fact}\label{Fact:KL-Hellinger}Suppose that $P$ and $Q$
are two probability measures such that 
\[
\frac{\mathrm{d}P}{\mathrm{d}Q}\leq R\quad\text{and}\quad\frac{\mathrm{d}Q}{\mathrm{d}P}\leq R
\]
hold uniformly over the probability space. Then, one has
\begin{equation}
\max\left\{ 2-0.5\log R,\text{ }1\right\} \cdot\mathsf{Hel}_{\frac{1}{2}}\left(P\hspace{0.2em}\|\hspace{0.2em}Q\right)\leq\mathsf{KL}\left(P\hspace{0.2em}\|\hspace{0.2em}Q\right)\leq\left(2+\log R\right)\cdot\mathsf{Hel}_{\frac{1}{2}}\left(P\hspace{0.2em}\|\hspace{0.2em}Q\right).\label{eq:KL-Hellinger}
\end{equation}
Furthermore, if $R\leq4.5$, then one has 
\begin{equation}
\left(2-0.4\log R\right)\cdot\mathsf{Hel}_{\frac{1}{2}}\left(P\hspace{0.2em}\|\hspace{0.2em}Q\right)\leq\mathsf{KL}\left(P\hspace{0.2em}\|\hspace{0.2em}Q\right)\leq\left(2+0.4\log R\right)\cdot\mathsf{Hel}_{\frac{1}{2}}\left(P\hspace{0.2em}\|\hspace{0.2em}Q\right).\label{eq:KL-Hellinger-small-R}
\end{equation}
\end{fact}

\begin{proof}See Appendix \ref{sec:Proof-of-Fact-KL-Hellinger}.\end{proof}

We conclude this part with another quantity that will often prove
useful in tightening our results. Specifically, for any $\zeta>0$,
we define
\begin{eqnarray}
m^{\mathsf{kl}}\left(\zeta\right) & := & \max_{l}\left|\left\{ i\hspace{0.2em}\hspace{0.2em}\big|\hspace{0.2em}\hspace{0.2em}i\neq l,\text{ }\mathsf{KL}\left(\mathbb{P}_{i}\hspace{0.2em}\|\hspace{0.2em}\mathbb{P}_{l}\right)\leq\left(1+\zeta\right)\mathsf{KL}^{\min}\right\} \right|.\label{eq:m-KL}
\end{eqnarray}
It is self-evident that $1\leq m^{\mathsf{kl}}\left(\zeta\right)<M$
holds regardless of $\zeta$. This quantity determines the number
of distinct input pairs under study that result in nearly-minimal
output separation.

\section{Main Results: Erd\H{o}s\textendash Rényi Graphs\label{sec:Erd=000151s=002013R=0000E9nyi-Graphs}}

At an intuitive level, faithful decoding is feasible only when (i)
the measurement graph $\mathcal{G}$ is sufficiently connected so
that we have enough measurements involving each vertex variable, and
(ii) the channel output distributions given any two distinct inputs
are sufficiently separated and hence distinguishable. To develop a
more quantitative understanding about these two factors, we start
with the Erd\H{o}s\textendash Rényi model, a tractable yet the most
widely adopted random graph model for numerous applications. Specifically,
we suppose that the measurement graph $\mathcal{G}$ is drawn from
$\mathcal{G}_{n,p_{\mathrm{obs}}}$ for some edge probability $p_{\mathrm{obs}}\gtrsim\log n/n$.
As will be shown in Section \ref{sec:General-Graphs}, many properties
and intuitions that we develop for this specific graph model hold
in greater generality.

\subsection{Maximum Likelihood Decoding\label{sub:ML-ERG}}

To begin with, we analyze the performance guarantees of the maximum
likelihood (ML) decoder
\begin{equation}
\psi_{\mathrm{ml}}\left(\boldsymbol{y}\right):=\arg\max_{\boldsymbol{x}\in\mathcal{X}^{n}}\mathbb{P}\left\{ \boldsymbol{y}\mid\boldsymbol{x}\right\} .\label{eq:MLtest}
\end{equation}
It is well-known that the ML rule minimizes the Bayesian probability
of error under uniform input priors. We develop a sufficient recovery
condition in terms of the edge probability and the minimum information
divergence, which characterizes the tradeoff between the degree of
graph connectivity and the resolution of channel outputs.

\begin{theorem}\label{thm:ML-Hellinger-ERG}Fix $\delta>0$, and
suppose that $\mathcal{G}\sim\mathcal{G}_{n,p_{\mathrm{obs}}}$. Then
there exist some universal constants $C$, $c_{1}>0$ such that if
\begin{eqnarray}
\sup_{0<\alpha<1}\left\{ \left(1-\alpha\right)\mathsf{Hel}_{\alpha}^{\min}\right\} \cdot\left(p_{\mathrm{obs}}n\right) & \geq & \left(1+\delta\right)\log\left(2n\right)+2\log\left(M-1\right),\label{eq:Achievability-Hel}
\end{eqnarray}
then the ML decoder $\psi_{\mathrm{ml}}$ obeys
\begin{eqnarray*}
P_{\mathrm{e}}\left(\psi_{\mathrm{ml}}\right) & \leq & \frac{1}{\left(2n\right)^{\max\left\{ \frac{3}{4}\delta-\frac{1}{4}\delta^{2},\text{ }\frac{\delta-1}{2}\right\} }-1}+\frac{3}{n^{10}-1}+Cn^{-c_{1}\delta n}.
\end{eqnarray*}
\end{theorem}

\begin{proof}See Appendix \ref{sec:Proof-of-Theorem-ML-Hellinger-ERG}.\end{proof}

Theorem \ref{thm:ML-Hellinger-ERG} essentially suggests that the
ML rule is guaranteed to work with high probability as long as 
\[
\sup_{\alpha}\left\{ \left(1-\alpha\right)\mathsf{Hel}_{\alpha}^{\min}\right\} \geq\left(1+o\left(1\right)\right)\frac{\log n+2\log M}{p_{\mathrm{obs}}n}.
\]
Our result is non-asymptotic in the sense that it holds for all parameters
$(n,M,\mathsf{Hel}_{\alpha}^{\min})$ instead of limiting to the asymptotic
regime with $n$ tending to infinity. Recognizing that $p_{\mathrm{obs}}n$
is exactly the average vertex degree $d_{\mathrm{avg}}$, our recovery
condition reads
\begin{equation}
\sup_{\alpha}\left\{ \left(1-\alpha\right)\mathsf{Hel}_{\alpha}^{\min}\right\} \cdot d_{\mathrm{avg}}\text{ }\gtrsim\text{ }\log n,\label{eq:RecoveryCond-simple-ER}
\end{equation}
provided that $M\lesssim\mathcal{O}\left(\mathrm{poly}\left(n\right)\right)$
and $\alpha\in(0,1)$ is some fixed constant independent of $n$. 

We pause to develop some intuitive understanding about the condition
(\ref{eq:RecoveryCond-simple-ER}). In contrast to classical information
theory settings, the channel decoding model considered herein concerns
``uncoded'' channel input. Consequently, the recovery bottleneck
for the ML rule is presented by the minimum output distance given
two distinct hypotheses, rather than the mutual information that plays
a crucial rule in coded transmission. To be more precise, two hypotheses
$\boldsymbol{x}$ and $\tilde{\boldsymbol{x}}$ are the least separated
when they differ only by one component, say, $v$. As a concrete example,
one can take $\boldsymbol{x}=[0,\cdots,0]$ and $\tilde{\boldsymbol{x}}=[1,0,\cdots,0]$
as illustrated in Fig.~\ref{fig:Intuition}. The resulting pairwise
outputs $\left\{ y_{ij}\mid(i,j)\in\mathcal{E}\right\} $ thus contain
about $\mathrm{deg}\left(v\right)$ pieces of information for distinguishing
$\boldsymbol{x}$ and $\tilde{\boldsymbol{x}}$; see, e.g., the orange
shaded region highlighted in Fig.~\ref{fig:Intuition}. Since the
information contained in each measurement can be quantified by certain
divergence metrics, namely, $\mathsf{Hel}_{\alpha}^{\min}$ (or $\mathsf{KL}^{\min}$
as adopted in Section \ref{sub:Minimax-Lower-Bound-ER}), the total
amount of information one has available to distinguish two minimally
separated hypotheses is captured by
\begin{equation}
\mathsf{Hel}_{\alpha}^{\min}\cdot d_{\mathrm{avg}}\quad\text{or}\quad\mathsf{KL}^{\min}\cdot d_{\mathrm{avg}}.\label{eq:total-info}
\end{equation}
Furthermore, there are at least $n$ distinct hypotheses that are
all minimally apart from the ground truth $\boldsymbol{x}$ (e.g.~$\tilde{\boldsymbol{x}}=[1,0,\cdots,0]$,
$\tilde{\boldsymbol{x}}=[0,1,\cdots,0]$, $\cdots$, $\tilde{\boldsymbol{x}}=[0,\cdots,0,1]$).
Representation of these hypotheses calls for at least $\log n$ bits,
and hence the information that one can exploit to distinguish $\boldsymbol{x}$
from them---i.e.~(\ref{eq:total-info})---needs to exceed $\log n$.
This offers an intuitive interpretation of the recovery condition
(\ref{eq:RecoveryCond-simple-ER}).

\begin{figure}
\begin{tabular}{cccc}
\includegraphics[width=0.23\textwidth]{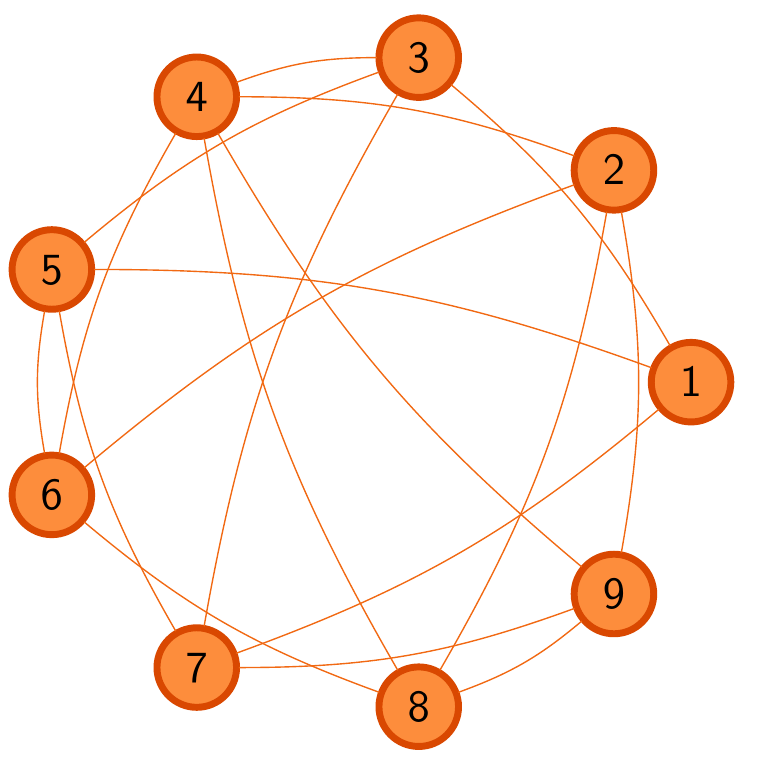} & \includegraphics[width=0.23\textwidth]{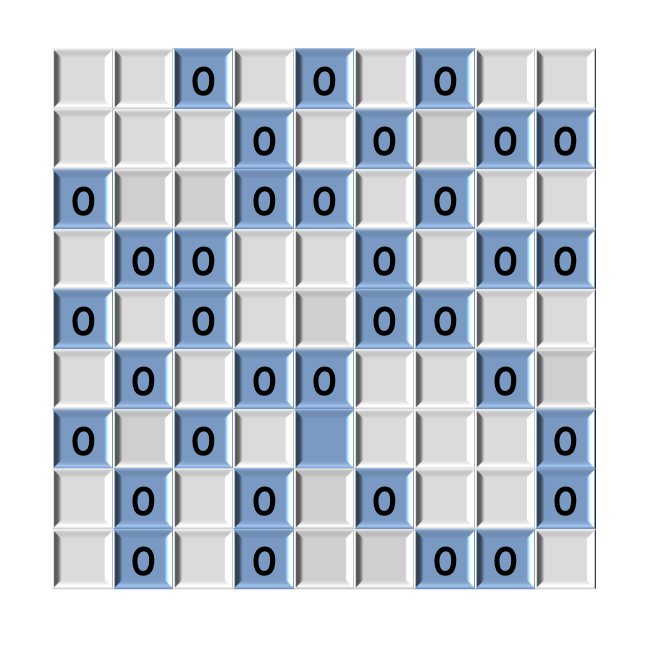} & \includegraphics[width=0.23\textwidth]{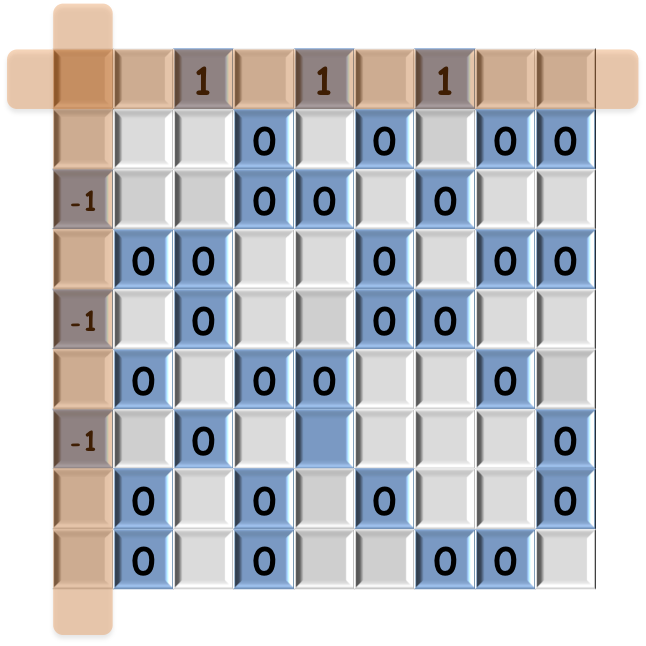} & \includegraphics[width=0.23\textwidth]{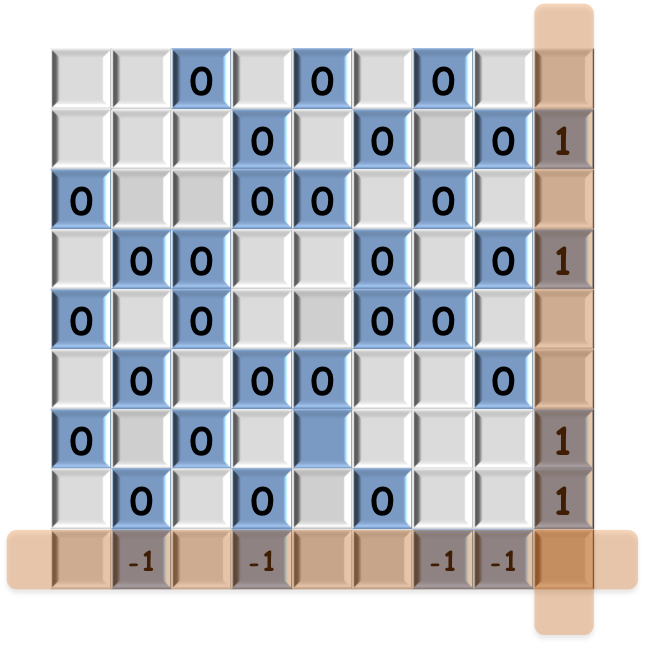}\tabularnewline
(a) & (b) & (c) & (d)\tabularnewline
\end{tabular}\caption{\label{fig:Intuition}The pairwise inputs $[x_{i}-x_{j}]_{1\leq i,j\leq n}$
under a realization of $\mathcal{G}_{n,p_{\mathrm{obs}}}$ with $n=8$
and $p_{\mathrm{obs}}=0.3$ as shown in (a). The input patterns are
shown for (b) the ground truth $\boldsymbol{x}=\boldsymbol{0}$, (c)
the hypothesis $\tilde{\boldsymbol{x}}=[1,0,\cdots,0]$, and (d) the
hypothesis $\tilde{\boldsymbol{x}}=[0,\cdots,0,1]$. The blue parts
represent the entries being measured, and the orange region constitutes
the parts of pairwise inputs that disagree with the ground truth. }
\end{figure}

Careful readers will note that Theorem \ref{thm:ML-Hellinger-ERG}
is presented in terms of the Hellinger divergence rather than the
KL divergence. We remark on these this technical matter as follows.

\begin{remark}In general, we are unable to develop the recovery conditions
in terms of the KL divergence. This arises partly because the KL divergence
cannot be well controlled for all measures, especially when $\left\Vert \frac{\mathrm{d}\mathbb{P}_{l}}{\mathrm{d}\mathbb{P}_{j}}\right\Vert _{\infty}$
($l\neq j$) grows. In contrast, the Hellinger divergence is generally
stable and more convenient to analyze in this case. \end{remark}

We conclude this part with an extension. Examining our analysis reveals
that all arguments continue to hold even if the output distributions
are \emph{location-dependent}. Formally, suppose that the distribution
of $y_{ij}$ is parametrized by 
\begin{equation}
p\left(y_{ij}\Big|x_{i}-x_{j}=l\right)=\mathbb{P}_{l}^{ij}\left(y_{ij}\right),\quad0\leq l<M,\text{ }(i,j)\in\mathcal{E}.\label{eq:channel-model-general}
\end{equation}
This leads to a modified version of the minimum divergence metric
as follows
\begin{eqnarray}
\overline{\mathsf{Hel}}_{\alpha}^{\min} & := & \min\left\{ \left.\mathsf{Hel}_{\alpha}(\mathbb{P}_{l}^{ij}\hspace{0.2em}\|\hspace{0.2em}\mathbb{P}_{k}^{ij})\hspace{0.2em}\right|l\neq k,\hspace{0.2em}0\leq l,k<M,\hspace{0.2em}(i,j)\in\mathcal{E}\right\} ;\label{eq:Helmin_general}
\end{eqnarray}
With these modified metrics in place, the preceding sufficient recovery
condition immediately extends to this generalized model. 

\begin{theorem}\label{theorem:general-Hellinger-ERG}The recovery
condition of Theorem \ref{thm:ML-Hellinger-ERG} continues to hold
under the transition probabilities (\ref{eq:channel-model-general}),
if $\mathsf{Hel}_{\alpha}^{\min}$ is replaced by $\overline{\mathsf{Hel}}_{\alpha}^{\min}$
as defined in (\ref{eq:Helmin_general}).\end{theorem}

\subsection{Minimax Lower Bound\label{sub:Minimax-Lower-Bound-ER}}

In order to assess the tightness of our recovery guarantee for the
ML rule, we develop two necessary conditions that apply to any recovery
procedure. Here and below, $H\left(x\right):=-x\log x-\left(1-x\right)\log\left(1-x\right)$
stands for the binary entropy function. 

\begin{theorem}\label{corollary:ConverseHellinger-ERG}Suppose that
$\mathcal{G}\sim\mathcal{G}_{n,p_{\mathrm{obs}}}$. Fix any $\zeta\geq0$
and $\epsilon>0$, and assume that $p_{\mathrm{obs}}>\frac{c\log n}{n}$
for some sufficiently large constant $c>0$.

(a) If
\begin{eqnarray}
\mathsf{KL}^{\min}\cdot p_{\mathrm{obs}}n & \leq & \frac{\left(1-\epsilon\right)\Big(\log n+\log m^{\mathsf{kl}}\left(\zeta\right)\Big)-H\left(\epsilon\right)}{\left(1+\epsilon\right)\left(1+\zeta\right)},\label{eq:Fano-KL-thm-Gnp}
\end{eqnarray}
then $\inf_{\psi}P_{\mathsf{e}}\left(\psi\right)\geq\epsilon-n^{-10}$.

(b) Suppose that $\alpha\leq\frac{1}{1+\epsilon}$ and $p_{\mathrm{obs}}n>2\epsilon\alpha\log n$.
If
\begin{equation}
\left(1-\alpha\right)\mathsf{Hel}_{\alpha}^{\min}\cdot p_{\mathrm{obs}}n\text{ }<\text{ }\frac{\epsilon\alpha\log n}{1+\zeta}-r_{\epsilon}\label{eq:Converse-Hel-Gnp}
\end{equation}
for some residual \footnote{More precisely, $r_{\epsilon}:=\log2+\frac{2\left[\epsilon\alpha\log n-\log2\right]^{2}}{np_{\mathrm{obs}}}$.}
$r_{\epsilon}$, then $\inf_{\psi}P_{\mathsf{e}}\left(\psi\right)\geq n^{-\epsilon}-n^{-10}.$\end{theorem}\begin{proof}See
Appendices \ref{sec:Proof-of-Theorem-ConverseM-Degree} and \ref{sec:Proof-of-Theorem-Converse-ERG}.\end{proof}

We remark that the two necessary recovery conditions in Theorem \ref{corollary:ConverseHellinger-ERG}
concern two regimes of separate interest. Specifically, Condition
(\ref{eq:Fano-KL-thm-Gnp}) based on the KL divergence is most useful
when investigating first-order convergence, namely, the situation
where we only require the minimax probability of error to be asymptotically
vanishing without specifying convergence rates. In comparison, Condition
(\ref{eq:Converse-Hel-Gnp}) based on the Hellinger distance is more
convenient when we further demand exact recovery to occur with polynomially
high probability (e.g.~$1-\frac{1}{n}$). In various ``big-data''
applications, the term ``with high probability'' might only refer
to the case where the error probability decays at least at a polynomial
rate. 

On the other hand, while Condition (\ref{eq:Fano-KL-thm-Gnp}) is
not directly presented in terms of $M$, we can often capture the
effect of the alphabet size through the surrogate $m^{\mathsf{kl}}$,
provided that $\log m^{\mathsf{kl}}\asymp\log M$. In fact, this arises
in many scenarios of interest. As an example, see the outlier model
to be discussed in Section \ref{sec:Outlier}, where $m^{\mathsf{kl}}=M-1$.

\subsection{Tightness of Theorems \ref{thm:ML-Hellinger-ERG} and \ref{corollary:ConverseHellinger-ERG}
and Minimal Sample Complexity}

Encouragingly, the recovery conditions derived in Theorems \ref{thm:ML-Hellinger-ERG}
and \ref{corollary:ConverseHellinger-ERG} are often tight up to some
small multiplicative constant. In the sequel, we will assume that
$p_{\mathrm{obs}}\gg\log n/n$, and will pay special attention to
two of the most popular divergence metrics: the KL divergence and
the squared Hellinger distance.
\begin{enumerate}
\item Consider the first-order convergence, that is, the regime where $\inf_{\psi}P_{\mathrm{e}}\left(\psi\right)\rightarrow0$
($n\rightarrow\infty$). Combining Theorems \ref{thm:ML-Hellinger-ERG}
and \ref{corollary:ConverseHellinger-ERG}(a) suggests that
\begin{eqnarray*}
\inf_{\psi}P_{\mathsf{e}}\left(\psi\right) & \overset{n\rightarrow\infty}{\longrightarrow} & 0\quad\quad\quad\text{if }\mathsf{Hel}_{\frac{1}{2}}^{\min}\cdot p_{\mathrm{obs}}n>\left(1+o\left(1\right)\right)\left(2\log n+4\log M\right),\\
\inf_{\psi}P_{\mathsf{e}}\left(\psi\right) & \overset{n\rightarrow\infty}{\centernot\longrightarrow} & 0\quad\quad\quad\text{if }\mathsf{KL}^{\min}\cdot p_{\mathrm{obs}}n<\left(1-o\left(1\right)\right)\left(\log n+\log m^{\mathsf{kl}}\right).
\end{eqnarray*}
When applied to the most challenging case where $\frac{\mathrm{d}\mathbb{P}_{l}}{\mathrm{d}\mathbb{P}_{j}}=1+o\left(1\right)$
for all $l\neq j$, these conditions read (with the assistance of
Fact \ref{Fact:KL-Hellinger})
\begin{eqnarray}
\inf_{\psi}P_{\mathsf{e}}\left(\psi\right) & \overset{n\rightarrow\infty}{\longrightarrow} & 0\quad\quad\quad\text{if }\mathsf{Hel}_{\frac{1}{2}}^{\min}\cdot p_{\mathrm{obs}}n>\left(1+o\left(1\right)\right)\left(2\log n+4\log M\right),\label{eq:Hel-ML-ER}\\
\inf_{\psi}P_{\mathsf{e}}\left(\psi\right) & \overset{n\rightarrow\infty}{\centernot\longrightarrow} & 0\quad\quad\quad\text{if }\mathsf{Hel}_{\frac{1}{2}}^{\min}\cdot p_{\mathrm{obs}}n<\left(1-o\left(1\right)\right)\frac{\log n+\log m^{\mathsf{kl}}}{2},\label{eq:Hel-converse-ER}
\end{eqnarray}
which are matching conditions modulo some multiplicative factor not
exceeding 
\begin{equation}
\left(1+o\left(1\right)\right)\frac{4\log n+8\log M}{\log n+\log m^{\mathsf{kl}}}.\label{eq:factor-1}
\end{equation}

\item We now move on to more stringent convergence by considering the regime
where $\lim_{n\rightarrow\infty}\inf_{\psi}P_{\mathrm{e}}\left(\psi\right)\lesssim1/n$.
Putting Theorem \ref{thm:ML-Hellinger-ERG} (with $\delta=3$) and
Theorem \ref{corollary:ConverseHellinger-ERG}(b) (with $\alpha=1/2$
and $\epsilon=1$) together implies that 
\begin{eqnarray}
\inf_{\psi}P_{\mathsf{e}}\left(\psi\right) & < & \frac{1}{n}\quad\quad\quad\text{if }\mathsf{Hel}_{\frac{1}{2}}^{\min}\cdot p_{\mathrm{obs}}n>\left(1+o\left(1\right)\right)\left(8\log n+4\log M\right),\label{eq:Hel-ER-achievability}\\
\inf_{\psi}P_{\mathsf{e}}\left(\psi\right) & > & \frac{1}{n}\quad\quad\quad\text{if }\mathsf{Hel}_{\frac{1}{2}}^{\min}\cdot p_{\mathrm{obs}}n<\left(1-o\left(1\right)\right)\log n,\label{eq:Hel-ER-converse}
\end{eqnarray}
which holds regardless of $\frac{\mathrm{d}\mathbb{P}_{l}}{\mathrm{d}\mathbb{P}_{j}}$.
These conditions do not impose constraints on the alphabet size, and
are tight up to a multiplicative gap of
\begin{equation}
\left(1+o\left(1\right)\right)\left(8+\frac{4\log M}{\log n}\right).\label{eq:factor-2}
\end{equation}

\end{enumerate}
\begin{remark}The multiplicative factor (\ref{eq:factor-1}) is small
when either the alphabet size $M=O\left(\text{poly}\log(n)\right)$
(in which case this factor is $4$) or when $M\asymp m^{\mathrm{kl}}$
(in which case this factor is at most 8). Similarly, the multiplicative
factor (\ref{eq:factor-2}) is the smallest when the alphabet size
is $O\left(\text{poly}\log(n)\right)$. However, both results might
become loose when $\log M$ and $\frac{\log M}{\log m^{\mathrm{kl}}}$
grow. For many practical applications, the alphabet size is typically
much smaller than $n$, in which case our results are tight within
a reasonable constant factor. \end{remark}

In summary, we have characterized the fundamental recovery condition
under the Erd\H{o}s\textendash Rényi model, which reads
\begin{equation}
\mathsf{Hel}_{\frac{1}{2}}^{\min}\cdot d_{\mathrm{avg}}\text{ }\gtrsim\text{ }\log n\label{eq:-2}
\end{equation}
as long as the alphabet size $M$ is not super-polynomial\footnote{When the alphabet size is super-polynomial in $n$, our upper and
lower bounds are within a factor of $O\left(\frac{\log M}{\log n}\right)$
from optimal. We note, however, that in all our motivating applications,
the alphabet size $M$ is typically much smaller than $\exp\left(\Theta\left(n\right)\right)$
and hence the regime with super-polynomial alphabet size is of little
practical relevance. } in $n$. Put another way, in order to allow exact recovery, the \emph{sample
complexity}---i.e.~the total number of edges of $\mathcal{G}$ (which
is around $nd_{\mathrm{avg}}/2$)---necessarily obeys
\begin{equation}
\text{minimum sample complexity }\asymp\text{ }\frac{n\log n}{\mathsf{Hel}_{\frac{1}{2}}^{\min}}.\label{eq:min-sample-size}
\end{equation}
Interestingly, these simple characterizations as well as the underlying
intuitions carry over to many more homogeneous graphs, as will be
seen in the next section.

Before concluding this section, we remark on the possibility of improving
the preconstant. There are a few cases where the tight preconstants
have been settled, including the stochastic block model \cite{abbe2014exact,hajek2014achieving}
and the censor block model \cite{hajek2014achieving,chen2016community},
provided that the alphabet size (or the number of communities) is
$M=2$. Asymptotically, the necessary and sufficient recovery condition
reads
\[
\sup_{0<\alpha<1}\left\{ (1-\alpha)\mathsf{Hel}_{\alpha}^{\min}\right\} >\frac{\log n}{p_{\mathrm{obs}}n}
\]
\[
\text{or}\quad\text{minimum sample complexity}>\frac{n\log n}{2\sup_{0<\alpha<1}\left\{ (1-\alpha)\mathsf{Hel}_{\alpha}^{\min}\right\} },
\]
thus justifying the tightness of the sufficient recovery condition
we derive. In fact, when the alphabet size is a constant, the fundamental
divergence measure that dictates the information limits is often some
variant of the minimum Chernoff information\footnote{Note that the Chernoff information is defined to be $-\log\left\{ 1-\sup_{0<\alpha<1}(1-\alpha)\mathsf{Hel}_{\alpha}\right\} $. }
or the Hellinger divergence (when optimized over $\alpha$). Nevertheless,
it is not clear whether such findings extend to the large alphabet
settings. Part of the reason is that the minimum Chernoff information
or Hellinger divergence do not necessarily capture the precise error
exponent when testing many hypotheses. We leave to future work the
investigation of tight preconstants in the large alphabet scenarios.

\section{Main Results: General Graphs\label{sec:General-Graphs}}

We now broaden our scope by exploring general measurement graphs beyond
the simple Erd\H{o}s\textendash Rényi model, with emphasis on the
family of homogeneous graphs.

\subsection{Preliminaries: Key Graphical Metrics\label{sub:Key-Graphical-Metrics}}

Our theory relies on several widely encountered graphical metrics
including the minimum vertex degree, the average vertex degree, the
maximum vertex degree, and the size of the minimum cut, which we denote
by $d_{\min}$, $d_{\mathrm{avg}}$, $d_{\max}$, and $\mincut$,
respectively. This subsection introduces a few other not-so-common
graphical quantities that prove crucial in presenting our results. 

For any integer $m$, define
\begin{eqnarray}
\mathcal{N}\left(m\right): & = & \left\{ \mathcal{S}\subset\mathcal{V}\mid\text{ }e\left(\mathcal{S},\mathcal{S}^{\mathrm{c}}\right)\leq m\right\} ,\label{eq:DefnNk-D}
\end{eqnarray}
which comprises all cuts of size at most $m$. We are particularly
interested in the peak growth rate of the cardinality of $\mathcal{N}$
as defined below
\begin{equation}
\tau_{k}^{\mathrm{cut}}:=\frac{1}{k}\log\Big|\mathcal{N}\left(k\cdot\mincut\right)\Big|\quad\quad\text{and}\quad\quad\tau^{\mathrm{cut}}:=\underset{k>0}{\max}\text{ }\tau_{k}^{\mathrm{cut}}.\label{eq:DefnExponent}
\end{equation}
In the sequel, we will term $\tau^{\mathrm{cut}}$ the \textbf{\emph{cut-homogeneity
exponent}}. In fact, if we rewrite
\begin{equation}
\tau_{k}^{\mathrm{cut}}:=\mincut\cdot\left\{ \frac{1}{k\cdot\mincut}\log\Big|\mathcal{N}\left(k\cdot\mincut\right)\Big|\right\} ,\label{eq:tau_cut_alternative}
\end{equation}
then we see that $\tau^{\mathrm{cut}}$ relies on two factors: (i)
the cut-set distribution exponents $\left\{ \frac{1}{k}\log\left|\mathcal{N}(k)\right|\right\} _{k>0}$
and (2) the size of the minimum cut, both of which are important in
capturing\emph{ }the degree of homogeneity of the cut-set distribution.
This metric is best illustrated through the following two extreme
examples:
\begin{itemize}
\item \emph{Complete graph $K_{n}$ on $n$ vertices}. This homogeneous
graph obeys $e\left(\mathcal{S},\mathcal{S}^{\mathrm{c}}\right)=\left|\mathcal{S}\right|\left(n-\left|\mathcal{S}\right|\right)$
and $\mincut=n-1$. A simple combinatorial argument suggests that
$|\mathcal{N}\left(m\right)|\asymp{n \choose m/n}\asymp n^{m/n}$,
revealing that
\[
\tau^{\mathrm{cut}}=\max_{k}\frac{1}{k}\log\big|\mathcal{N}\left(k\cdot\mincut\right)\big|\text{ }\asymp\text{ }\max_{k}\frac{1}{k}\frac{k\cdot\mincut}{n}\log n\text{ }\asymp\text{ }\log n.
\]

\item \emph{Two complete subgraphs $K_{n/2}$ connected by a single bridge}.
In this graph, the min-cut size is $\mincut=1$ due to the existence
of a bridge, but we still have $|\mathcal{N}\left(m\right)|\lesssim n^{m/n}$
when $m\geq n$. A little algebra gives
\[
\tau^{\mathrm{cut}}=\max_{k}\frac{1}{k}\log\big|\mathcal{N}\left(k\cdot\mincut\right)\big|\text{ }\lesssim\text{ }\max_{k}\frac{1}{k}\frac{k\cdot\mincut}{n}\log n\text{ }\asymp\text{ }\frac{\log n}{n}.
\]

\end{itemize}
Interestingly, for various homogeneous graphs of interest, $\tau^{\mathrm{cut}}$
can be bounded above in a tight and simple manner, namely, $\tau^{\mathrm{cut}}\lesssim\log n$.
This is formally stated in the following lemma, which accounts for
homogeneous geometric graphs and expander graphs. In words, a graph
is said to be a homogeneous geometric graph if it satisfies two properties:
(i) each connected pair of vertices shares sufficiently many neighbors;
(ii) when two vertices are geometrically close, they share a large
fraction of neighbors. Here and throughout, we shall use $\mathcal{V}\left(u\right)$
to denote the set of neighbors of a vertex $u$. 

\begin{lem}\label{lemma:alpha_random_graphs}\textbf{(1) Homogeneous
geometric graphs}. Suppose that $\mathcal{G}$ is connected and is
embedded in some Euclidean space. Assume that there exist two numerical
constants $\rho>0$ and $0<\kappa<\frac{1}{2}$ such that \begin{enumerate}[(a)]{\setlength\itemindent{15pt}\item for
each $\left(u,v\right)\in\mathcal{E}$, \vspace{-1em}

\begin{equation}
\left|\mathcal{V}\left(u\right)\cap\mathcal{V}\left(v\right)\right|\geq\rho\cdot\mincut;\label{eq:GeometricOverlap}
\end{equation}
}{\setlength\itemindent{15pt}\item for each $\left(u,v\right)\in\mathcal{E}$,
denoting by $w^{(i)}$ the $i^{\text{th}}$ closest vertex to $v$
among the vertices in $\mathcal{V}\left(u\right)\cap\mathcal{V}\left(v\right)$
, one has
\begin{equation}
\left|\mathcal{V}\left(v\right)\text{ }\backslash\text{ }\mathcal{V}(w^{(i)})\right|\leq\frac{1}{2}\rho\cdot\mincut,\quad1\leq i\leq\kappa\rho\cdot\mincut.\label{eq:Closest}
\end{equation}
}\end{enumerate}

Under the above two conditions, one has 
\begin{equation}
\tau^{\mathrm{cut}}\leq\frac{8}{\kappa\rho}\log\left(2n\right).\label{eq:tau-Geometric}
\end{equation}
\textbf{(2) Expander graphs}. If $\mathcal{G}$ is an expander graph
with edge expansion $h_{\mathcal{G}}$, then 
\begin{equation}
\tau^{\mathrm{cut}}\leq\frac{\mincut}{h_{\mathcal{G}}}\log n+\log2.\label{eq:tau-expander}
\end{equation}

\end{lem}\begin{proof}See Appendix \ref{sec:Proof-of-Lemma:alpha_random_graphs}.\end{proof}

We highlight a few concrete examples covered by this lemma. 
\begin{itemize}
\item The following instances of homogeneous geometric graphs are worth
mentioning. The first is a random geometric graph $\mathcal{G}_{n,r}$,
provided that $r^{2}>c\log n$ for some sufficiently large $c>0$.
The second is a generalized ring in which two vertices are connected
as long as they are at most a few vertices apart. For both cases,
$\kappa$ and $\rho$ are constants bounded away from zero, indicating
that
\[
\tau^{\mathrm{cut}}\lesssim\log n.
\]

\item Another situation concerns those expander graphs with good expansion
properties, including but not limited to Erd\H{o}s\textendash Rényi
graphs, random regular graphs, and small world graphs. Since the expansion
properties of these graphs obey $h_{\mathcal{G}}/\mincut=\Theta\left(1\right)$,
we conclude from Lemma \ref{lemma:alpha_random_graphs} that
\[
\tau^{\mathrm{cut}}\lesssim\log n.
\]

\end{itemize}
As a final remark, we are not aware of a graph for which $\tau^{\mathrm{cut}}$
exceeds the order of $\log n$. In all aforementioned examples, one
always has $\tau^{\mathrm{cut}}\lesssim\log n$. In-depth study about
the upper limit on $\tau^{\mathrm{cut}}$ might lead to further simplification
of our results, which we leave for future work.

\subsection{ML Decoding and Minimax Lower Bounds}

This section presents recovery conditions based on the minimum information
divergence and certain graphical metrics, which accommodate general
graph structures, channel characteristics, and input alphabets. We
defer detailed discussion of our results to Section \ref{sub:Interpretation-and-Discussion}. 

To begin with, the following theorem---whose proof can be found in
Appendix \ref{sec:Proof-of-Theorem-AchieveM-Deg}---characterizes
a regime where the ML decoder is guaranteed to work.

\begin{theorem}\label{thm:AchieveM-Deg}Consider any connected graph
$\mathcal{G}$. For any $\delta>0$ and any $0<\alpha<1$, the ML
rule $\psi_{\mathrm{ml}}$ achieves
\[
P_{\mathrm{e}}\left(\psi_{\mathrm{ml}}\right)\leq\frac{1}{\left(2n\right)^{\delta}-1},
\]
provided that 
\begin{equation}
\sup_{0<\alpha<1}\left\{ \left(1-\alpha\right)\mathsf{Hel}_{\alpha}^{\min}\right\} \cdot\mincut\text{ }\geq\text{ }8\tau^{\mathrm{cut}}+\left(\delta+8\right)\log\left(2n\right)+4\log M.\label{eq:Hellinger-ML-General}
\end{equation}
\end{theorem}\begin{remark}Theorem \ref{thm:AchieveM-Deg} continues
to hold if $\mathsf{Hel}_{\alpha}^{\min}$ is replaced by $D_{\alpha}^{\min}$.\end{remark}

The sufficient recovery condition given in Theorem \ref{thm:AchieveM-Deg}
is universal and holds for all graphs, and depends only on the min-cut
size and the cut-homogeneity exponent irrespective of other graphical
metrics. Similarly, the above sufficient condition extends to the
scenario with location-dependent output distributions, as stated below.

\begin{theorem}\label{theorem:general-Hellinger}The recovery condition
of Theorem \ref{thm:AchieveM-Deg} continues to hold under the transition
probabilities (\ref{eq:channel-model-general}), if $\mathsf{Hel}_{\alpha}^{\min}$
and $D_{\alpha}^{\min}$ are replaced by $\overline{\mathsf{Hel}}_{\alpha}^{\min}$
and $\overline{D}_{\alpha}^{\min}$, respectively, where $\overline{D}_{\alpha}^{\min}:=-\frac{1}{1-\alpha}\log\big(1-\left(1-\alpha\right)\overline{\mathsf{Hel}}_{\alpha}^{\min}\big)$.\end{theorem}

Next, we present a fundamental lower limit on $\mathsf{KL}^{\min}$
that admits perfect information recovery, based on the same graphical
metrics in addition to the maximum vertex degree. 

\begin{theorem}[\textbf{KL Version}]\label{thm:ConverseM-Degree-KL}Fix
$\zeta\geq0$ and $0<\epsilon\leq1/2$. For any graph $\mathcal{G}$,
if the KL divergence satisfies
\begin{eqnarray}
\mathsf{KL}^{\min}\cdot\mincut & \leq & \max\left\{ \left(1-\epsilon\right)\tau^{\mathrm{cut}}-H(\epsilon),\text{ }\frac{\left(1-\epsilon\right)\log m^{\mathsf{kl}}\left(\zeta\right)-H(\epsilon)}{1+\zeta}\right\} \label{eq:Fano-KL-thm-homogeneous}
\end{eqnarray}
or
\begin{eqnarray}
\mathsf{KL}^{\min}\cdot d_{\max} & \leq & \text{ }\frac{\left(1-\epsilon\right)\left(\log n+\log m^{\mathsf{kl}}\left(\zeta\right)\right)-H(\epsilon)}{1+\zeta},\label{eq:Fano-KL-thm-homogeneous-1}
\end{eqnarray}
then the minimax probability of error exceeds $\inf_{\psi}P_{\mathsf{e}}\left(\psi\right)\geq\epsilon$.\end{theorem}

\begin{proof}See Appendix \ref{sec:Proof-of-Theorem-ConverseM-Degree}.\end{proof}

Notably, the conditions (\ref{eq:Fano-KL-thm-homogeneous}) and (\ref{eq:Fano-KL-thm-homogeneous-1})
do not imply each other. The first condition (\ref{eq:Fano-KL-thm-homogeneous})---which
characterizes the effects of cut-set distributions and alphabet size---is
dominant for inhomogeneous graphs where $\mincut\ll d_{\max}$ (e.g.~the
graph formed by connecting two $K_{n/2}$ with a single bridge as
described in Section \ref{sub:Key-Graphical-Metrics}). In comparison,
the other condition (\ref{eq:Fano-KL-thm-homogeneous-1}) becomes
tighter as $\frac{\mathrm{\mincut}}{d_{\max}}$ grows, which is particularly
useful when accounting for the family of homogeneous graphs where
$d_{\max}\asymp\mincut$. 

Finally, we complement the above KL version by another lower bound
developed directly based on the Hellinger divergence, although it
becomes loose for those inhomogeneous graphs obeying $\mincut\ll d_{\max}$.
This is particularly useful when investigating the scenario that demands
high-probability recovery (e.g. with success probability at least
$1-n^{-1}$). The proof can be found in Appendix \ref{sec:Proof-of-Theorem-Converse-ERG}.

\begin{theorem}[\textbf{Hellinger Version}]\label{thm:ConverseHellinger}Consider
any graph $\mathcal{G}$, any $\epsilon>0$, and $\alpha\leq\frac{1}{1+\epsilon}$.
Suppose that $d_{\max}\geq2\epsilon\alpha\log n$. If 
\begin{equation}
\left(1-\alpha\right)\mathsf{Hel}_{\alpha}^{\min}\cdot d_{\max}\leq\epsilon\alpha\log n-r_{\epsilon}\label{eq:Converse-Hel}
\end{equation}
for some residual\footnote{More precisely, $r_{\epsilon}:=\log2+\frac{2\left[\epsilon\alpha\log n-\log2\right]^{2}}{d_{\max}}$.}
$r_{\epsilon}$, then $\inf_{\psi}P_{\mathsf{e}}\left(\psi\right)\geq n^{-\epsilon}$.\end{theorem}

\begin{remark}For any fixed $\epsilon>0$, one can optimize (\ref{eq:Converse-Hel})
over all $0<\alpha\leq\frac{1}{1+\epsilon}$ to derive a tighter condition.
\end{remark}

\subsection{Interpretation and Discussion\label{sub:Interpretation-and-Discussion}}

We now discuss the messages conveyed by the aforementioned results,
for which we emphasize a broad family of homogeneous graphs before
turning to the most general graphs. In what follows, our discussion
assumes $\frac{\mathrm{d}\mathbb{P}_{l}}{\mathrm{d}\mathbb{P}_{j}}=\mathcal{O}(1)$
for all $0\leq l,j<M$, in which case one has (by invoking Fact \ref{Fact:KL-Hellinger})
\begin{equation}
\mathsf{KL}^{\min}\text{ }\asymp\text{ }\mathsf{Hel}_{\frac{1}{2}}^{\min}.
\end{equation}
Operating upon such assumptions enables us to significantly simplify
the presentation, while still capturing the regime that is statistically
the most challenging (compared to its complement regime where $\frac{\mathrm{d}\mathbb{P}_{l}}{\mathrm{d}\mathbb{P}_{j}}\gg1$).

\subsubsection{Homogeneous Graphs}

Our recovery conditions are most useful when applied to homogeneous
graphs. Formally speaking, we term $\mathcal{G}$ a \emph{homogeneous
graph} if it satisfies
\begin{equation}
\mincut\text{ }\asymp\text{ }d_{\mathrm{avg}}\text{ }\asymp\text{ }d_{\max},\label{eq:homogeneous-graphs}
\end{equation}
which subsumes as special cases the widely adopted Erd\H{o}s\textendash Rényi
graphs, random geometric graphs, small world graphs, rings, grids,
and many other expander graphs. A few implications are in order. 
\begin{enumerate}
\item For all homogeneous graphs, one has 
\begin{eqnarray*}
\inf_{\psi}P_{\mathsf{e}}\left(\psi\right) & \overset{n\rightarrow\infty}{\longrightarrow} & 0\quad\quad\quad\text{if }\mathsf{Hel}_{\frac{1}{2}}^{\min}\cdot d_{\mathrm{avg}}\gtrsim\tau^{\mathrm{cut}}+\log n+\log M,\\
\inf_{\psi}P_{\mathsf{e}}\left(\psi\right) & \overset{n\rightarrow\infty}{\centernot\longrightarrow} & 0\quad\quad\quad\text{if }\mathsf{Hel}_{\frac{1}{2}}^{\min}\cdot d_{\mathrm{avg}}\lesssim\tau^{\mathrm{cut}}+\log n+\log m^{\mathsf{kl}}.
\end{eqnarray*}
In general, these results are within a multiplicative $\mathsf{gap}$
\[
\mathsf{gap}\lesssim\frac{\tau^{\mathrm{cut}}+\log n+\log M}{\tau^{\mathrm{cut}}+\log n+\log m^{\mathsf{kl}}}
\]
from optimal, which are orderwise tight when either $M\lesssim\mathrm{poly}\left(n\right)$
or $\log m^{\mathsf{kl}}\asymp\log M$. In particular, as long as
the alphabet size is not super-polynomial in $n$, we arrive at the
fundamental recovery condition for this class of graphs: 
\begin{equation}
\mathsf{Hel}_{\frac{1}{2}}^{\min}\cdot d_{\mathrm{avg}}\text{ }\gtrsim\text{ }\log n+\tau^{\mathrm{cut}}.\label{eq:condition-homogeneous}
\end{equation}

\item In comparison to the recovery guarantee developed for the Erd\H{o}s\textendash Rényi
model, the condition (\ref{eq:condition-homogeneous}) includes one
extra correction term $\tau^{\mathrm{cut}}$ concerning the cut-set
distribution. To provide some intuition about $\tau^{\mathrm{cut}}$,
suppose that the ground truth is $\boldsymbol{x}=\boldsymbol{0}$
and consider an alternative hypothesis $\tilde{\boldsymbol{x}}$ whose
non-zero entries are all identical. If we denote by $\mathcal{S}$
the vertex set corresponding to the support of $\tilde{\boldsymbol{x}}$,
then it is straightforward to see that all measurements that can help
distinguish $\boldsymbol{x}$ and $\tilde{\boldsymbol{x}}$ reside
in the cut set $\mathcal{E}(\mathcal{S},\mathcal{S}^{\mathrm{c}})$.
By definition, $\tau_{k}^{\mathrm{cut}}$ determines the total number
of distinct cuts whose size is within some fixed range. Since $\tau_{k}^{\mathrm{cut}}$
is defined in a logarithmic and normalized manner, this in turn specifies
how many bits are needed to represent all these cuts and, hence, all
hypotheses associated with them. As a consequence, $\tau^{\mathrm{cut}}$
presents another information-theoretic requirement. 
\item While our results fall short of a general upper bound on $\tau^{\mathrm{cut}}$,
we note that $\tau^{\mathrm{cut}}\lesssim\log n$ holds for a broad
class of interesting models studied in the literature (and in fact
all models that we are aware of), including but not limited to various
homogeneous geometric graphs and expander graphs (cf.~Lemma \ref{lemma:alpha_random_graphs}).
As a consequence, the recovery condition (\ref{eq:condition-homogeneous})
for these graphs further simplifies to 
\begin{equation}
\mathsf{Hel}_{\frac{1}{2}}^{\min}\cdot d_{\mathrm{avg}}\text{ }\gtrsim\text{ }\log n,\label{eq:Condition-various-homogeneous}
\end{equation}
which coincides with the one under the special Erd\H{o}s\textendash Rényi
model. Following the intuition given in Section \ref{sub:ML-ERG},
one must rely on around $d_{\mathrm{avg}}$ measurements to distinguish
two minimally separated hypotheses---i.e.~those that differ by a
single component---and hence the information bottleneck constitutes
around $\mathsf{Hel}_{\frac{1}{2}}^{\min}\cdot d_{\mathrm{avg}}$
bits, which needs to be at least $\log n$ bits in order to encode
$n$ minimally apart hypotheses. 
\item The condition (\ref{eq:Condition-various-homogeneous}) in turn leads
to an interesting observation: for a variety of homogeneous graphs,
the information-theoretic limits for graph-based decoding are determined
solely by the \emph{edge sparsity}, as opposed to the performance
guarantees for many tractable algorithms (e.g.~spectral methods or
semidefinite programming) whose success typically rely on strong \emph{second-order}
expansion properties. 
\end{enumerate}
Finally, by combining Theorem \ref{thm:AchieveM-Deg} and Theorem
\ref{thm:ConverseHellinger} (with $\epsilon=1$), we arrive at the
following criterion concerning ``high-probability'' recovery: for
various homogeneous graphs that obey $\tau^{\mathrm{cut}}\lesssim\log n$,
the probability of error $P_{\mathrm{e}}\left(\psi\right)\leq n^{-1}$
is possible if and only if 
\begin{equation}
\mathsf{Hel}_{\frac{1}{2}}^{\min}\cdot d_{\mathrm{avg}}\text{ }\gtrsim\text{ }\log n.
\end{equation}
In contrast to the preceding discussion, this statement holds regardless
of how $\frac{\mathrm{d}\mathbb{P}_{l}}{\mathrm{d}\mathbb{P}_{j}}$
scales. 

\subsubsection{General Graphs}

We now move on to discussing the results in their full generality.
One distinguishing feature from the family of homogeneous graphs is
that the recovery boundary is dictated by the size of the minimum
cut rather than the graph edge sparsity. For the convenience of the
reader, we summarize all key results in Table \ref{tab:Summary-of-results}.

\begin{table}
\centering%
\begin{tabular}{c|c}
\hline 
Measurement graphs & Fundamental recovery conditions\tabularnewline
\hline 
Erd\H{o}s\textendash Rényi graphs & $\mathsf{Hel}_{1/2}^{\min}\cdot d_{\mathrm{avg}}\gtrsim\log n$\tabularnewline
\hline 
homogeneous geometric graphs, expander graphs & $\mathsf{Hel}_{1/2}^{\min}\cdot d_{\mathrm{avg}}\gtrsim\log n$\tabularnewline
\hline 
homogeneous graphs & $\mathsf{Hel}_{1/2}^{\min}\cdot d_{\mathrm{avg}}\gtrsim\log n+\tau^{\mathrm{cut}}$\tabularnewline
\hline 
general graphs & $\mathsf{Hel}_{1/2}^{\min}\cdot\mincut\gtrsim g\left(n\right)+\tau^{\mathrm{cut}}$$\text{ }$$\text{ }$
$(1\lesssim g\left(n\right)\lesssim\log n)$\tabularnewline
\hline 
\end{tabular}\caption{\label{tab:Summary-of-results}Summary of key results for all graph
models ($M\lesssim\mathrm{poly}\left(n\right)$).}
\end{table}

\begin{enumerate}
\item \textbf{Tightness under general graphs}. The recovery conditions presented
in Theorems \ref{thm:AchieveM-Deg} and \ref{thm:ConverseM-Degree-KL}
can be summarized as follows 
\begin{eqnarray*}
\inf_{\psi}P_{\mathsf{e}}\left(\psi\right) & \overset{n\rightarrow\infty}{\longrightarrow} & 0\quad\quad\quad\text{if }\mathsf{Hel}_{\frac{1}{2}}^{\min}\cdot\mincut\gtrsim\tau^{\mathrm{cut}}+\log M+\log n,\\
\inf_{\psi}P_{\mathsf{e}}\left(\psi\right) & \overset{n\rightarrow\infty}{\centernot\longrightarrow} & 0\quad\quad\quad\text{if }\mathsf{Hel}_{\frac{1}{2}}^{\min}\cdot\mincut\lesssim\tau^{\mathrm{cut}}+\log m^{\mathsf{kl}}+\frac{\mincut}{d_{\max}}\log n.
\end{eqnarray*}
These are within a multiplicative $\mathsf{gap}$ from optimal, satisfying
that
\[
\mathsf{gap}\lesssim\frac{\tau^{\mathrm{cut}}+\log M+\log n}{\tau^{\mathrm{cut}}+\log m^{\mathsf{kl}}+\frac{\mincut}{d_{\max}}\log n}.
\]
Recognizing that $\tau^{\mathrm{cut}}\gtrsim1$, we see that the derived
bounds are orderwise optimal when $\log m^{\mathsf{kl}}\asymp\log M\asymp\log n$
(e.g.~in the outlier model presented in Section \ref{sec:Outlier}
under large alphabet). Even for the loosest case, the gap is at most
logarithmic (i.e.~$\mathcal{O}\left(\log n+\log M\right)$). 
\item \textbf{Information bottleneck}. In contrast to (\ref{eq:condition-homogeneous})
and (\ref{eq:Condition-various-homogeneous}), the amount of information
one has available to differentiate two minimally separated hypotheses
is approximately given by $\mathsf{Hel}_{\frac{1}{2}}^{\min}\cdot\mincut$
instead of $\mathsf{Hel}_{\frac{1}{2}}^{\min}\cdot d_{\mathrm{avg}}$.
This makes sense since the two hypotheses that are most difficult
to differentiate are no longer those that differ by one component.
Instead, the most challenging task lies in linking the variables across
the minimum cut, which can convey at most $\mathsf{Hel}_{\frac{1}{2}}^{\min}\cdot\mincut$
bits of information, forming the most fragile component for simultaneous
recovery.
\item \textbf{A unified non-asymptotic framework.} Our framework can accommodate
a variety of practical scenarios that respect the high-dimensional
regime: the alphabet size might be growing with $n$ while the channel
divergence metrics might be decaying. Furthermore, our problem falls
under the category of multi-hypothesis testing in the presence of
exponentially many hypotheses, where each hypothesis is \textit{not}
necessarily formed by i.i.d. sequences. Under such a setting, the
conventional Sanov bound \cite{csiszar1984sanov} based on the Chernoff
information measure \cite{dembo1998large} becomes unwieldy. In contrast,
our results build upon alternative probability divergence measures
(particularly the Hellinger / Rényi divergence). This results in a
simple unified framework that enables non-asymptotic characterization
of the minimax limits (modulo some constant factor) simultaneously
for most settings.
\end{enumerate}
In general, the current approach is unable to close the worst-case
gap $\mathcal{O}\left(\log n+\log M\right)$, which could be large
when either $n$ or $M$ are exceedingly large. In order to improve
the recovery conditions, one alternative is to derive a tighter lower
bound on the graphical metric $\tau^{\mathrm{cut}}$. For instance,
our bounds become orderwise tight whenever $\tau^{\mathrm{cut}}\gtrsim\log n$,
which arise in various graphs beyond the family of homogeneous graphs.
We leave this for future investigation. In addition, our general lower
bounds are developed based on Fano's inequality, since Fano's inequality
allows us to accommodate a set of input hypotheses that have significant
overlaps. Unfortunately, Fano's inequality typically relies on the
KL divergence between input hypotheses, which is in general not capable
of capturing the right error exponent for hypothesis testing. It would
be interesting to develop a variant of the Fano-type inequality based
directly on the Chernoff information measures.

\section{Consequences for Specific Applications \label{sec:Applications}}

In this section, we apply our general theory to a few concrete examples
that have been studied in prior literature. As will be seen, our general
theorems lead to order-wise tight characterization for all these canonical
examples.

\subsection{Stochastic Block Model\label{sub:SBM}}

We start by analyzing the stochastic block model (SBM), which is a
generative way to model community structure. In the standard SBM,
nodes are partitioned into two disjoint clusters (so one can assign
labels $x_{i}\in\left\{ 0,1\right\} $ for each node). Each pair of
nodes is connected with probability $\frac{\alpha\log n}{n}$ or $\frac{\beta\log n}{n}$
depending on whether they fall within the same cluster or not. The
goal is to infer the underlying clusters that produce the network.
Of particular interest is exact recovery of the entire clusters, which
has received considerable attention; see \cite{mossel2012stochastic,abbe2014exact,jalali2011clustering,mossel2014consistency,chen2014statistical,hajek2014achieving,montanari2015semidefinite,bandeira2015random,lei2014generic,yun2014accurate,chin2015stochastic,gao2015achieving}
for a highly incomplete list of references.

We focus on the regime where $\alpha,\beta=o\left(n/\log n\right)$
and $\alpha>\beta$, which subsumes all but the densest community
structures. Treating the SBM as a graphical channel over a complete
measurement graph (i.e.~$p_{\mathrm{obs}}=1$) with outputs being
either 0 or 1---which encodes whether two nodes belong to the same
cluster or not, we see that (cf.~Definition \ref{eq:channel-model})
\[
\mathbb{P}_{0}=\mathsf{Bern}\left(\frac{\alpha\log n}{n}\right),\quad\text{and}\quad\mathbb{P}_{1}=\mathsf{Bern}\left(\frac{\beta\log n}{n}\right).
\]
This allows us to compute
\begin{eqnarray*}
\mathsf{Hel}_{\frac{1}{2}}^{\min} & = & \left(\sqrt{\frac{\alpha\log n}{n}}-\sqrt{\frac{\beta\log n}{n}}\right)^{2}+\left(\sqrt{1-\frac{\alpha\log n}{n}}-\sqrt{1-\frac{\beta\log n}{n}}\right)^{2}\\
 & = & \left(1+o\left(1\right)\right)(\sqrt{\alpha}-\sqrt{\beta})^{2}\frac{\log n}{n}.
\end{eqnarray*}
In addition, the relation between KL divergence and $\chi^{2}$ divergence
(e.g.~\cite[Equation (7)]{van2014renyi}) suggests that
\begin{eqnarray}
\mathsf{KL}^{\min} & \leq & \mathsf{KL}\left(\frac{\beta\log n}{n}\hspace{0.2em}\|\hspace{0.2em}\frac{\alpha\log n}{n}\right)\leq\chi^{2}\left(\frac{\beta\log n}{n}\hspace{0.2em}\|\hspace{0.2em}\frac{\alpha\log n}{n}\right)\nonumber \\
 & \overset{(\text{a})}{=} & \frac{\left(\frac{\beta\log n}{n}-\frac{\alpha\log n}{n}\right)^{2}}{\frac{\alpha\log n}{n}\left(1-\frac{\alpha\log n}{n}\right)}=\left(1+o(1)\right)\frac{\left(\alpha-\beta\right)^{2}}{\alpha}\frac{\log n}{n},\label{eq:KL-Bern}
\end{eqnarray}
where (a) follows from the identity $\chi^{2}\left(p\hspace{0.2em}\|\hspace{0.2em}q\right)=\frac{\left(p-q\right)^{2}}{q}+\frac{\left(p-q\right)^{2}}{1-q}=\frac{\left(p-q\right)^{2}}{q\left(1-q\right)}.$
With these two estimates in place, Theorem \ref{thm:ML-Hellinger-ERG}
and Corollary \ref{corollary:ConverseHellinger-ERG} immediately give
\begin{eqnarray}
\inf_{\psi}P_{\mathrm{e}}\left(\psi\right) & \overset{n\rightarrow\infty}{\longrightarrow} & 0\quad\text{if }(\sqrt{\alpha}-\sqrt{\beta})^{2}\geq2\left(1+o\left(1\right)\right),\\
\inf_{\psi}P_{\mathrm{e}}\left(\psi\right) & \overset{n\rightarrow\infty}{\centernot\longrightarrow} & 0\quad\text{if }\left(\alpha-\beta\right)^{2}\leq\left(1-o\left(1\right)\right)\alpha.\label{eq:KL_SBM}
\end{eqnarray}

In fact, precise phase transition for exact cluster recovery has only
been determined last year \cite{mossel2014consistency,abbe2014exact}.
There results assert that
\begin{eqnarray}
\inf_{\psi}P_{\mathrm{e}}\left(\psi\right) & \overset{n\rightarrow\infty}{\longrightarrow} & 0\quad\text{if }(\sqrt{\alpha}-\sqrt{\beta})^{2}>2,\\
\inf_{\psi}P_{\mathrm{e}}\left(\psi\right) & \overset{n\rightarrow\infty}{\centernot\longrightarrow} & 0\quad\text{if }(\sqrt{\alpha}-\sqrt{\beta})^{2}<2,\label{eq:necessary-SBM}
\end{eqnarray}
justifying that the sufficient condition we develop is precise. When
it comes to the necessary condition, one can verify that the condition
(\ref{eq:necessary-SBM}) is more stringent than\footnote{To see this, observe that $\left(\sqrt{\alpha}-\sqrt{\beta}\right)^{2}<2$
is identical to $\left(\alpha-\beta\right)^{2}<2\left(\sqrt{\alpha}+\sqrt{\beta}\right)^{2}$,
which is more stringent than $\left(\alpha-\beta\right)^{2}<4\left(\alpha+\beta\right)$
due to the elementary inequality $(a+b)^{2}\leq2\left(a^{2}+b^{2}\right)$.} $\left(\alpha-\beta\right)^{2}<4\left(\alpha+\beta\right)$. In comparison,
the boundary of our condition (\ref{eq:KL_SBM}) is sandwiched between
the curves $\left(\alpha-\beta\right)^{2}\leq\frac{1}{2}\left(\alpha+\beta\right)$
and $\left(\alpha-\beta\right)^{2}\leq\alpha+\beta$. These taken
collectively indicate that our theory is tight up to a small constant
factor. 

Several remarks are in order. To begin with, our results accommodate
all values of $\alpha,\beta$ up to $o\left(n/\log n\right)$, which
is broader than \cite{abbe2014exact} that concentrates on the sparsest
possible regime (i.e.~$\alpha,\beta\asymp1$). Leaving out this technical
matter, a more interesting observation is that the achievability bound
we develop for the ML rule matches the fundamental recovery limit
in a precise manner, which seems to imply that the squared Hellinger
distance is the right metric that dictates the recovery limits for
the SBMs.

When finishing up this paper, we became aware of a very recent work
\cite{abbe2015community} that characterizes the fundamental limits
for the generalized SBM, that is, the model where $n$ nodes are partitioned
into multiple clusters. Extending our framework so as to accommodate
the SBM in its full generality is a topic of future work. 

\subsection{Outlier Model \label{sec:Outlier}}

We now turn to another model called the \emph{outlier model}, which
subsumes as special cases several applications including alignment,
synchronization, and joint matching (e.g. \cite{wang2012exact,huang2013consistent,chen2013matching}). 

Suppose that the measurements $y_{ij}$'s are independently corrupted
following a distribution
\begin{equation}
y_{ij}=\begin{cases}
x_{i}-x_{j},\quad & \text{with probability }p_{\mathrm{true}},\\
\mathsf{Unif}{}_{M},\quad & \text{else},
\end{cases}\label{eq:CorruptionModel}
\end{equation}
where $\mathsf{Unif}{}_{M}$ is the uniform distribution over $\left\{ 0,\cdots,M-1\right\} $,
$p_{\mathrm{true}}$ stands for the non-corruption rate, and ``$-$''
is some general subtraction operation defined in Section \ref{sec:Problem-Formulation}.
In words, a fraction $1-p_{\mathrm{true}}$ of measurements act as
\emph{random outliers} and contain no useful information. Note that
under this outlier model, one has
\[
m^{\mathsf{kl}}\left(\epsilon\right)\equiv M-1,\quad\forall\epsilon\geq0.
\]

The following corollary---an immediate consequence of Theorem \ref{thm:ML-Hellinger-ERG}
and Corollary \ref{corollary:ConverseHellinger-ERG}---presents concrete
recovery limits for the outlier model. For ease of presentation, we
restrict our discussion to the Erd\H{o}s\textendash Rényi model, but
remark that all results extend to homogeneous geometric graphs and
other expander graphs (up to some constant factors) if one replaces
$p_{\mathrm{obs}}n$ with the average vertex degree.

\begin{corollary}\label{corollary-Outlier}Fix $\epsilon>0$. Consider
the outlier model (\ref{eq:CorruptionModel}), and assume $\mathcal{G}\sim\mathcal{G}_{n,p_{\mathrm{obs}}}$with
$p_{\mathrm{obs}}>\frac{c_{1}\log n}{n}$ for some sufficiently large
$c_{1}>0$. Then, one has 
\begin{eqnarray}
\inf_{\psi}P_{\mathrm{e}}\left(\psi\right) & \overset{n\rightarrow\infty}{\longrightarrow} & 0\quad\text{if }\frac{1}{M}\left(\sqrt{1-p_{\mathrm{true}}+Mp_{\mathrm{true}}}-\sqrt{1-p_{\mathrm{true}}}\right)^{2}\geq\left(1+\epsilon\right)\frac{\log n+2\log M}{p_{\mathrm{obs}}n},\label{eq:Hel-outlier-achievable-InfoLimited}\\
\inf_{\psi}P_{\mathrm{e}}\left(\psi\right) & \overset{n\rightarrow\infty}{\centernot\longrightarrow} & 0\quad\text{if }p_{\mathrm{true}}\leq\max\left\{ \frac{\left(1-\epsilon\right)\left(\log n+\log M\right)}{p_{\mathrm{obs}}n\log\left(1+\frac{p_{\mathrm{true}}M}{1-p_{\mathrm{true}}}\right)},\text{ }\frac{M}{M-1}\left(\frac{\log n}{p_{\mathrm{obs}}n}-\frac{1}{M}\right)\right\} .\label{eq:KL-outlier-converse-InfoLimited}
\end{eqnarray}
\end{corollary}

To establish this corollary, we start by considering the graph $\mathcal{G}_{\mathrm{true}}$
that comprises all edges where $y_{ij}=x_{i}-x_{j}$. It is self-evident
that $\mathcal{G}_{\mathrm{true}}\sim\mathcal{G}_{n,\left(p_{\mathrm{true}}+\frac{1-p_{\mathrm{true}}}{M}\right)p_{\mathrm{obs}}}$,
and thus $\left(\left(1-\frac{1}{M}\right)p_{\mathrm{true}}+\frac{1}{M}\right)p_{\mathrm{obs}}>\frac{\log n}{n}$
is necessary to ensure connectivity (otherwise there will be no basis
to link the node variables across disconnected components). Apart
from this, everything boils down to calculating $\mathsf{KL}^{\min}$
and $\mathsf{Hel}^{\min}$, which we gather in the following lemma.

\begin{lem}\label{lemma-Outlier}Consider the outlier model (\ref{eq:CorruptionModel}).
For any $0\leq p_{\mathrm{true}}<1$, one has
\begin{equation}
\mathsf{KL}^{\min}=p_{\mathrm{true}}\log\left(1+\frac{p_{\mathrm{true}}M}{1-p_{\mathrm{true}}}\right)\quad\text{and}\quad\mathsf{\mathsf{Hel}}_{\frac{1}{2}}^{\min}=\frac{2}{M}\left(\sqrt{1-p_{\mathrm{true}}+Mp_{\mathrm{true}}}-\sqrt{1-p_{\mathrm{true}}}\right)^{2}.\label{eq:KL-smallM}
\end{equation}
More simply, these metrics can be bounded as
\begin{equation}
\mathsf{KL}^{\min}\leq\frac{p_{\mathrm{true}}^{2}M}{1-p_{\mathrm{true}}}\quad\text{and}\quad\mathsf{\mathsf{Hel}}_{\frac{1}{2}}^{\min}\geq\frac{p_{\mathrm{true}}^{2}M}{2\left(1-p_{\mathrm{true}}+Mp_{\mathrm{true}}\right)}.\label{eq:KL-smallM-1}
\end{equation}
\end{lem} \begin{proof}See Appendix \ref{sec:Proof-of-Lemma-Outlier}.
\end{proof}

To illustrate these guarantees numerically, we depict in Fig. \ref{fig:outlier}
an example of the preceding recovery conditions. In the sequel, we
will discuss the tightness and implications of the above result for
specific regimes, ranging from small alphabet to large alphabet. For
convenience of theoretical comparison, we supply an alternative form
obtained by applying the general theory but using the bounds (\ref{eq:KL-smallM-1}):
\begin{eqnarray}
\inf_{\psi}P_{\mathrm{e}}\left(\psi\right) & \overset{n\rightarrow\infty}{\longrightarrow} & 0\quad\text{if }p_{\mathrm{true}}\geq2\left(1+\epsilon\right)\sqrt{\frac{(1-p_{\mathrm{true}}+Mp_{\mathrm{true}})(\log n+2\log M)}{p_{\mathrm{obs}}nM}},\label{eq:Hel-outlier-achievable-InfoLimited-2}\\
\inf_{\psi}P_{\mathrm{e}}\left(\psi\right) & \overset{n\rightarrow\infty}{\centernot\longrightarrow} & 0\quad\text{if }p_{\mathrm{true}}\leq\max\left\{ \left(1-\epsilon\right)\sqrt{\frac{(1-p_{\mathrm{true}})(\log n+\log M)}{p_{\mathrm{obs}}nM}},\frac{\log n}{p_{\mathrm{obs}}n}\right\} .\label{eq:KL-outlier-converse-InfoLimited-2}
\end{eqnarray}

\begin{figure}
\centering

\begin{tabular}{cc}
\includegraphics[width=0.42\textwidth]{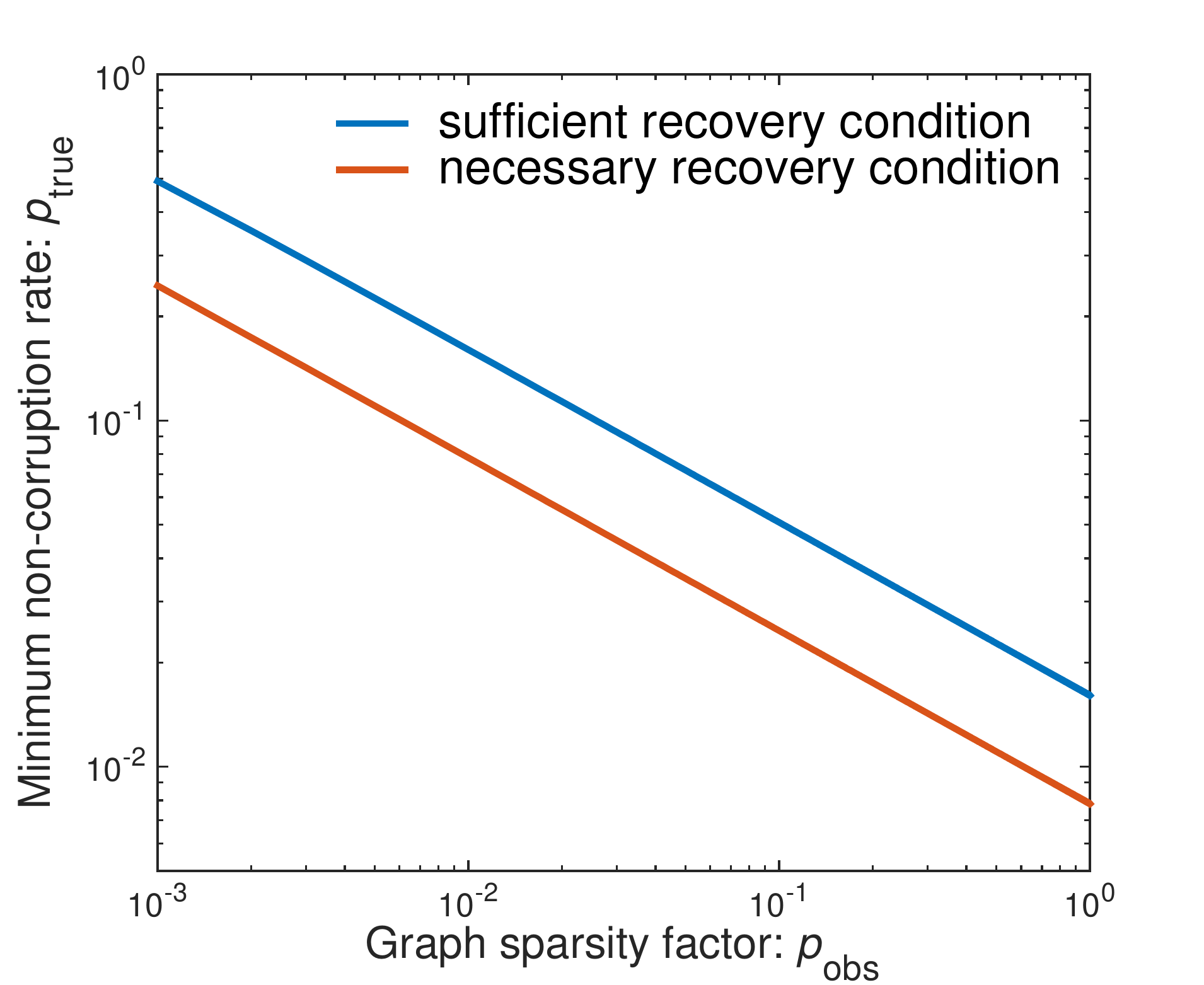} & \includegraphics[width=0.42\textwidth]{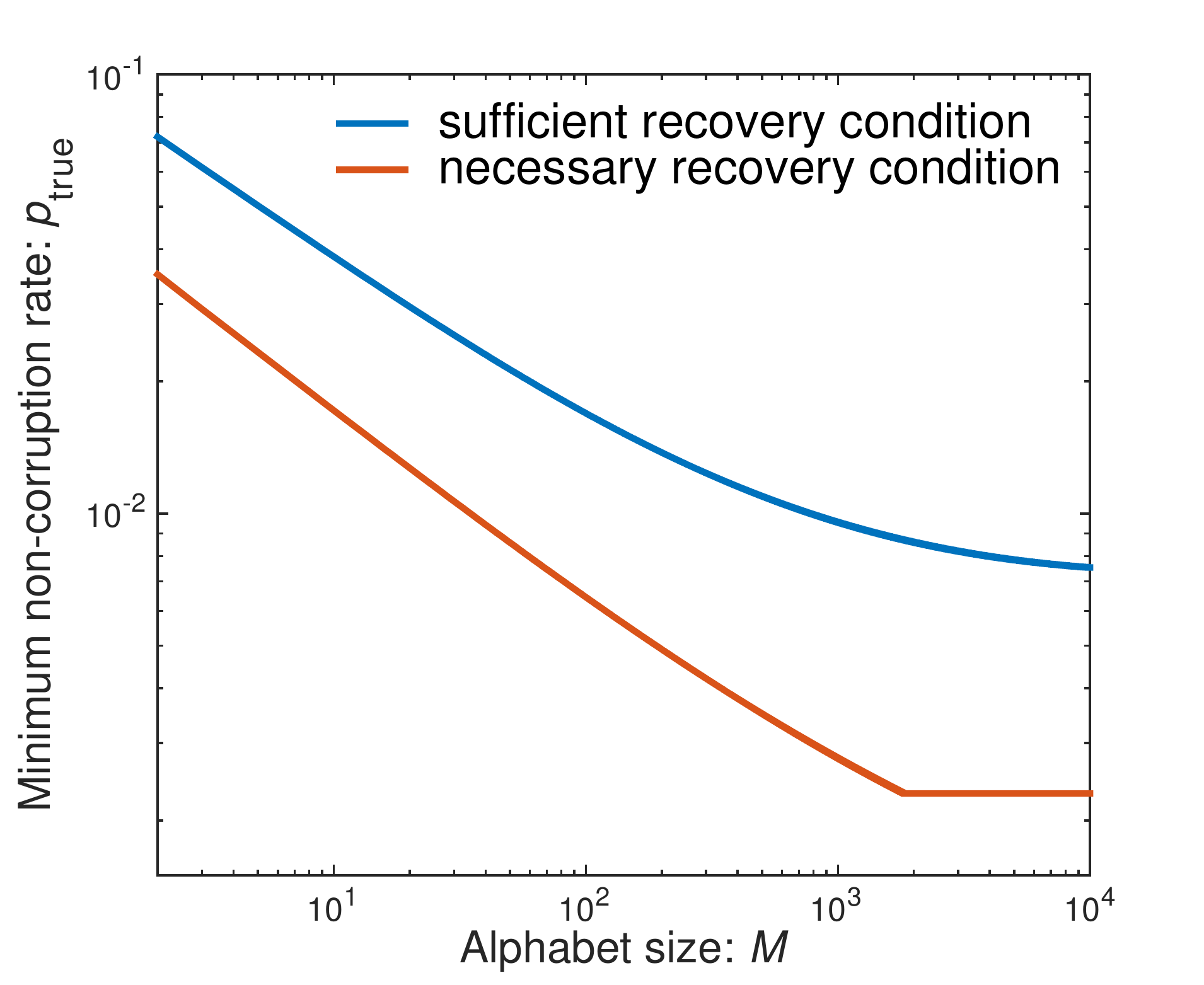}\tabularnewline
(a) & (b)\tabularnewline
\end{tabular}

\caption{The sufficient and the necessary conditions given in Corollary \ref{corollary-Outlier}
when $n=10^{5}$. The results are shown for: (a) $M=2$; (b) $p_{\mathrm{obs}}=0.05$.}
\label{fig:outlier}
\end{figure}

\subsubsection{Tightness under Binary Alphabet\label{sub:CBM}}

We start with the case where $M=2$, which was also studied by \cite{abbe2014decoding}.
When $p_{\mathrm{obs}}\gg\frac{\log n}{n}$, our results (\ref{eq:Hel-outlier-achievable-InfoLimited-2})
and (\ref{eq:KL-outlier-converse-InfoLimited-2}) assert that
\begin{eqnarray}
\inf_{\psi}P_{\mathrm{e}}\left(\psi\right) & \overset{n\rightarrow\infty}{\longrightarrow} & 0\quad\text{if }p_{\mathrm{true}}\geq\left(1+o\left(1\right)\right)\sqrt{\frac{2\log n}{p_{\mathrm{obs}}n}},\label{eq:Hel-outlier-achievable-InfoLimited-1}\\
\inf_{\psi}P_{\mathrm{e}}\left(\psi\right) & \overset{n\rightarrow\infty}{\centernot\longrightarrow} & 0\quad\text{if }p_{\mathrm{true}}\leq\left(1-o\left(1\right)\right)\sqrt{\frac{\log n}{2p_{\mathrm{obs}}n}}.\label{eq:KL-outlier-converse-InfoLimited-1}
\end{eqnarray}
As a result, our bounds are within a factor $2+o(1)$ from optimal,
which holds for all possible values of $(p_{{\rm obs}},p_{{\rm true}})$.
This constant gap is illustrated in Fig. \ref{fig:outlier}(a) as
well.

In contrast\textbf{, }the bounds presented in \cite{abbe2014decoding}
fall short of a uniform constant factor gap accommodating different
parameter configurations. Adopting our notation, \cite[Theorems 4.1 - 4.2]{abbe2014decoding}
reduce to\footnote{Note that $p_{\mathrm{true}}=1-2\varepsilon$ and $d=np_{\mathrm{obs}}$
for the notation $\varepsilon$ and $d$ defined in \cite{abbe2014decoding},
respectively. }: 
\begin{eqnarray*}
\inf_{\psi}P_{\mathsf{e}}\left(\psi\right) & \overset{n\rightarrow\infty}{\longrightarrow} & 0\quad\quad\quad\text{if }p_{{\rm true}}>\sqrt{\frac{2\log n}{p_{{\rm obs}}n}},\\
\inf_{\psi}P_{\mathsf{e}}\left(\psi\right) & \overset{n\rightarrow\infty}{\centernot\longrightarrow} & 0\quad\quad\quad\text{if }p_{{\rm true}}<\sqrt{\frac{2\left(1-3\tau/2\right)\log n}{p_{{\rm obs}}n}},
\end{eqnarray*}
where $0<\tau<\frac{2}{3}$ is some numerical value so that $p_{{\rm obs}}\leq2n^{\tau-1}$.
Hence, their bounds are tight up to a factor
\[
g\left(\tau\right)=\frac{1+o(1)}{\sqrt{1-3\tau/2}},
\]
which approaches 1 in the sparse graph regime as $\tau\rightarrow0$
(e.g. $p_{\mathrm{obs}}\asymp\frac{\log n}{n}$). On the other hand,
it does not deliver meaningful conditions for the case where $\tau\geq\frac{2}{3}$
(i.e. $2n^{-\frac{1}{3}}\leq p_{\mathrm{obs}}\leq1$). In comparison,
our bounds are looser for sparse graphs ($\tau<\frac{1}{2}$ or $p_{\mathrm{obs}}<\frac{2}{\sqrt{n}}$)
where $g\left(\tau\right)\leq2$, but tighter for dense graphs ($\tau\geq\frac{1}{2}$
or $p_{\mathrm{obs}}\geq\frac{2}{\sqrt{n}}$) where $g(\tau)\geq2$.

Notably, when $p_{\mathrm{obs}}\asymp\frac{\log n}{n}$, the fundamental
limit approaches $\sqrt{\frac{2\log n}{p_{{\rm obs}}n}}$ in an accurate
manner \cite{abbe2014decoding}. This again corroborates the tightness
of our achievability bound, implying that the squared Hellinger distance
is the right quantity to control in the sparsest possible regime.

\subsubsection{From Small Alphabet to Large Alphabet}

The recovery conditions given in Corollary \ref{corollary-Outlier}
can be further divided into and simplified for two respective regimes,
depending on whether $Mp_{\mathrm{true}}\lesssim1$ or $Mp_{\mathrm{true}}\gtrsim1$.
By substituting each of these two hypotheses into (\ref{eq:Hel-outlier-achievable-InfoLimited-2}),
deriving the corresponding minimum $p_{\mathrm{true}}$ for the respective
case, and then checking the compatibility of $p_{\mathrm{true}}M$
with the hypotheses, one immediately deduces:
\begin{enumerate}
\item When $M=o\left(\frac{p_{\mathrm{obs}}n}{\log n}\right)$, one has
\begin{eqnarray}
\inf_{\psi}P_{\mathrm{e}}\left(\psi\right) & \overset{n\rightarrow\infty}{\longrightarrow} & 0\quad\text{if }p_{\mathrm{true}}\geq2\left(1+o\left(1\right)\right)\sqrt{\frac{\log n+2\log M}{p_{\mathrm{obs}}nM}},\label{eq:Hel-outlier-achievable-smallM}\\
\inf_{\psi}P_{\mathrm{e}}\left(\psi\right) & \overset{n\rightarrow\infty}{\centernot\longrightarrow} & 0\quad\text{if }p_{\mathrm{true}}\leq\left(1-o\left(1\right)\right)\sqrt{\frac{\log n+\log M}{p_{\mathrm{obs}}nM}};\label{eq:KL-outlier-achievable-smallM}
\end{eqnarray}

\item When $M=\omega\left(\frac{p_{\mathrm{obs}}n}{\log n}\right)$, one
has
\begin{eqnarray}
\inf_{\psi}P_{\mathrm{e}}\left(\psi\right) & \overset{n\rightarrow\infty}{\longrightarrow} & 0\quad\text{if }p_{\mathrm{true}}\geq4\left(1+o\left(1\right)\right)\frac{\log n+2\log M}{p_{\mathrm{obs}}n},\label{eq:Hel-outlier-achievable-largeM}\\
\inf_{\psi}P_{\mathrm{e}}\left(\psi\right) & \overset{n\rightarrow\infty}{\centernot\longrightarrow} & 0\quad\text{if }p_{\mathrm{true}}\leq\frac{\log n}{p_{\mathrm{obs}}n}.\label{eq:KL-outlier-achievable-largeM}
\end{eqnarray}

\end{enumerate}
That being said, the recovery boundary presented in terms of $p_{\mathrm{true}}$
exhibits contrasting features in two separate regimes, as illustrated
in Fig. \ref{fig:outlier}(b). Some interpretations are in order.
\begin{enumerate}
\item \textbf{Information-limited regime} ($M=o\left(\frac{d_{\min}}{\log n}\right)$).
The amount of information that can be conveyed through each pairwise
measurement is captured by the divergence measure. In this small-alphabet
regime, a little algebra gives $\mathsf{KL}^{\min}\approx p_{\mathrm{true}}^{2}M$
(see Lemma \ref{lemma-Outlier}), which is increasing in $M$. As
a result, the alphabet size limits the amount of information that
we can harvest, and the fundamental recovery boundary improves with
$M$. For Erd\H{o}s\textendash Rényi graphs, the recovery conditions
are tight up to a factor of $2$ in the presence of a constant alphabet
size, and up to a factor of $2\sqrt{\frac{3}{2}}$ for all $M\ll d_{\min}/\log n$.
\item \textbf{Connectivity-limited regime} ($M=\omega\left(\frac{d_{\min}}{\log n}\right)$).
When $M$ further increases and enters this regime, the information
carried by each measurement saturates and no longer scales as $p_{\mathrm{true}}^{2}M$.
In this regime, the measurement graph $\mathcal{G}$ presents a fundamental
connectivity bottleneck. In fact, if $p_{\mathrm{true}}=o\left(\frac{\log n}{d_{\min}}\right)$,
then there will be at least one vertex that is not connected with
a single useful measurement, and hence there will be absolutely no
basis to infer the value of this isolated vertex. Our bounds in this
regime are order-wise optimal as long as the alphabet size is not
super-polynomial in $n$.
\end{enumerate}

\subsection{Haplotype Assembly\label{sub:Haplotype-Assembly}}

The pairwise measurement model can also be applied to analyze the
haplotype assembly problem discussed in Section \ref{sec:Introduction}.
As formulated in \cite{kamath2015optimal,si2014haplotype}, consider
$n$ SNPs on a chromosome, represented by a sequence $\{x_{1},\cdots,x_{n}\}\in\{0,1\}^{n}$
such that a major (resp. minor) allele is denoted by 0 (resp. 1).
Employing certain sequencing technologies, one obtains a collection
of \emph{independent} paired reads such that for any $(i,j)\in\mathcal{E}$,
\begin{equation}
y_{ij}^{(k)}=\left\{ \begin{array}{ll}
x_{i}\oplus x_{j}, & \text{w.p. }1-\theta,\\
x_{i}\oplus x_{j}\oplus1, & \text{w.p. }\theta.
\end{array}\right.\label{eq:Haplotypemodel}
\end{equation}
Here, $y_{ij}^{(k)}$ stands for the $k^{\text{th}}$ noisy read of
the parity between the $i^{\text{th}}$ and the $j^{\text{th}}$ SNPs,
and $0<\theta<1/2$ denotes the read error rate. We assume that the
reads taken on each edge are independent.

A realistic measurement graph that respects current sequencing technologies
is the one in which measurements are obtained only when the $i^{\text{th}}$
and the $j^{\text{th}}$ SNPs are geometrically close, i.e., $|i-j|\leq w$
for some constant\footnote{As discussed in \cite{kamath2015optimal}, the separation between
two DNA reads (called the \emph{insert size}) is typically bounded
within a small range, with the median insert size not exceeding a
few times the separation between adjacent SNPs. } $w>0$. This is captured by a \emph{generalized ring graph}, denoted
by ${\cal G}_{\mathrm{ring}}=\left(\mathcal{V},{\cal E}_{\mathrm{ring}}\right)$,
such that
\begin{equation}
\left(i,j\right)\in{\cal E}_{\mathrm{ring}}\quad\text{iff}\quad\left|i-j\right|\leq w.\label{eq:-1}
\end{equation}
The number $L_{i,j}$ of reads taken between $i$ and $j$ is assumed
to be dependent on their separation, i.e.\footnote{Careful readers will note that this assumption is different from the
model adopted in \cite{kamath2015optimal,si2014haplotype}, where
the total number of reads is fixed with the reads independently generated.
Nevertheless, the model considered here (which significantly simplifies
presentation) is sufficient to capture the right scaling of the performance
limits, since these two models are orderwise equivalent due to measure
concentration. }
\begin{equation}
L_{i,j}=Lp_{|i-j|}\label{eq:L_prob}
\end{equation}
for some parameters $L$ and $\left\{ p_{l}\mid1\leq l\leq w\right\} $. 

Additionally, a random and geometry-free measurement model has been
investigated in \cite{si2014haplotype} as well. The fundamental limit
under this model is orderwise\emph{ }equivalent to that under an Erd\H{o}s\textendash Rényi
graph with $L_{i,j}\equiv L$ for all $(i,j)\in\mathcal{E}$. For
the sake of completeness, we derive consequences for both models as
follows. 

\begin{corollary}\label{corollary-Haplotype}Consider the model (\ref{eq:Haplotypemodel}),
and assume that $\theta$ and $p_{l}$ are bounded away from 0.

(1) Suppose that $\mathcal{G}\sim{\cal G}_{\mathrm{ring}}$. There
exist some universal constants $c_{1}>c_{2}>0$ such that
\begin{eqnarray}
\inf_{\psi}P_{\mathsf{e}}\left(\psi\right) & \overset{n\rightarrow\infty}{\longrightarrow} & 0\quad\quad\quad\text{if }\left(1-2\theta\right)^{2}>c_{1}\frac{\log n}{L},\label{eq:Haplotype-ring-Achievability}\\
\inf_{\psi}P_{\mathsf{e}}\left(\psi\right) & \overset{n\rightarrow\infty}{\centernot\longrightarrow} & 0\quad\quad\quad\text{if }\left(1-2\theta\right)^{2}<c_{2}\frac{\log n}{L}.\label{eq:Haplotype-ring-converse}
\end{eqnarray}

(2) Suppose that $\mathcal{G}\sim\mathcal{G}_{n,p_{{\rm obs}}}$ and
$p_{{\rm obs}}>\frac{c_{3}\log n}{n}$ for some sufficiently large
constant $c_{3}>0$. Then there exist some universal constants $c_{4},c_{5}>0$
such that
\begin{eqnarray}
\inf_{\psi}P_{\mathsf{e}}\left(\psi\right) & \overset{n\rightarrow\infty}{\longrightarrow} & 0\quad\quad\quad\text{if }\left(1-2\theta\right)^{2}>\frac{c_{4}\log n}{Lnp_{\mathrm{obs}}},\label{eq:Haplotype-ER-Achievability}\\
\inf_{\psi}P_{\mathsf{e}}\left(\psi\right) & \overset{n\rightarrow\infty}{\centernot\longrightarrow} & 0\quad\quad\quad\text{if }\left(1-2\theta\right)^{2}<\frac{c_{5}\log n}{Lnp_{\mathrm{obs}}}.\label{eq:Haplotype-ER-converse}
\end{eqnarray}
 \end{corollary}

\begin{proof}For the sufficient condition, we only need to calculate
the Rényi divergence. For each $(i,j)\in\mathcal{E}$, letting $D_{1/2}^{i,j}$
be the Rényi divergence of order $1/2$ between the distributions
of $\{y_{ij}^{(k)}\}_{1\leq k\leq L_{i,j}}$ given two distinct inputs
(i.e.~0 and 1), one obtains
\begin{eqnarray}
-D_{1/2}^{i,j} & \overset{(\text{i})}{=} & L_{i,j}\left\{ -2\log\left(1-\frac{1}{2}\mathsf{Hel}\left(\theta\hspace{0.2em}\|\hspace{0.2em}1-\theta\right)\right)\right\} \overset{(\text{ii})}{\geq}L_{i,j}\mathsf{Hel}\left(\theta\hspace{0.2em}\|\hspace{0.2em}1-\theta\right),\label{eq:lower-bound1}
\end{eqnarray}
where (i) follows from additivity of Rényi divergence \cite[Theorem 2.8]{van2014renyi},
and (ii) follows since $1-x\leq e^{-x}$. Furthermore, 
\begin{equation}
\frac{1}{2}\mathsf{Hel}\left(\theta\hspace{0.2em}\|\hspace{0.2em}1-\theta\right)=\left(\sqrt{\theta}-\sqrt{1-\theta}\right)^{2}=\frac{\left(1-2\theta\right)^{2}}{(\sqrt{\theta}+\sqrt{1-\theta})^{2}}\asymp\left(1-2\theta\right)^{2}.\label{eq:lower-bound2}
\end{equation}
Recall that $L_{i,j}=Lp_{|i-j|}\asymp L$ when $\mathcal{G}\sim\mathcal{G}_{\mathrm{ring}}$,
whereas $L_{i,j}=L$ when $\mathcal{G}\sim\mathcal{G}_{n,p_{\mathrm{obs}}}$.
These taken together with Theorem \ref{theorem:general-Hellinger}
and Lemma \ref{lemma:alpha_random_graphs} (resp. Theorem \ref{theorem:general-Hellinger-ERG})
establish the sufficient condition for $\mathcal{G}_{\mathrm{ring}}$
(resp. $\mathcal{G}_{n,p_{\mathrm{obs}}}$).

For the necessary condition, by replacing all $p_{l}$ with $\max_{l}p_{l}$
in (\ref{eq:L_prob}), we obtain a new model such that any sufficient
recovery condition for the original model holds for this new model
as well. We then move on to compute the KL divergence for the new
model:
\begin{equation}
\mathsf{KL}^{\min}\leq L\mathsf{KL}\left(\theta\hspace{0.2em}\|\hspace{0.2em}1-\theta\right)\overset{(\text{a})}{\leq}L\mathsf{\chi}^{2}\left(\theta\hspace{0.2em}\|\hspace{0.2em}1-\theta\right)=L\frac{\left(1-2\theta\right)^{2}}{\theta\left(1-\theta\right)}\asymp L\left(1-2\theta\right)^{2},\label{eq:KL-bound-HA}
\end{equation}
where (a) follows from \cite[Equation (7)]{van2014renyi}. Substitution
into Theorem \ref{thm:ConverseM-Degree-KL} finishes the proof.\end{proof}

We now compare our results with prior results. The fundamental limits
given in \cite{kamath2015optimal,si2014haplotype} were based on coverage
(or sample complexity) as a metric, that is, the total number of reads
required for perfect haplotype assembly. Recognizing that $nLw$ (resp.~${n \choose 2}p_{\mathrm{obs}}L$)
captures the order of the total number of paired reads for $\mathcal{G}_{\mathrm{ring}}$
(resp.~$\mathcal{G}_{n,p_{\mathrm{obs}}}$), we see that the minimal
sample complexity obeys 
\begin{eqnarray}
nLw\text{ } & \asymp & \text{ }\frac{n\log n}{\left(1-2\theta\right)^{2}},\quad\text{when }\mathcal{G}\sim\mathcal{G}_{\mathrm{ring}};\\
{n \choose 2}Lp_{\mathrm{obs}} & \asymp & \frac{n\log n}{\left(1-2\theta\right)^{2}},\quad\text{when }\mathcal{G}\sim\mathcal{G}_{n,p_{\mathrm{obs}}}.
\end{eqnarray}
Consequently, for the generalized ring graph, our results match the
sample complexity limits characterized in \cite{kamath2015optimal}
in an orderwise sense, which is proportional to 
\[
\frac{n\log n}{1-e^{-\mathsf{KL}(0.5||\theta)}}\asymp\frac{n\log n}{\mathsf{KL}(0.5\hspace{0.2em}\|\hspace{0.2em}\theta)}=\frac{n\log n}{\left(\frac{1}{2}+o(1)\right)\left(1-2\theta\right)^{2}}.
\]
On the other hand, for the Erd\H{o}s\textendash Rényi graphs, the
minimum sample complexity scales as $\frac{n\log n}{(1-2\theta)^{2}}$,
which coincides with the orderwise limits $\Theta(n\log n)$ derived
in \cite{si2014haplotype}.

Notably, our results are not restricted to the classical large-sample
asymptotics where $\theta$ is fixed while $n$ grows to infinity.
This strengthens \cite{kamath2015optimal,si2014haplotype} by accommodating
the regime where $\theta-1/2=o(1)$, which characterizes the non-asymptotic
tradeoff between $n$ and the read quality. As a final remark, while
our results are tight in capturing the right scaling w.r.t.~the read
error rate as well as the number of SNPs, our derivation is not tight
in characterizing the behavior w.r.t.~$p_{l}$ (or the notation $W$
given in \cite{kamath2015optimal}).

\section{Concluding Remarks\label{sec:Concluding-Remarks}}

This paper investigates simultaneous recovery of multiple node variables
based on noisy graph-based measurements, under the pairwise difference
model. The problem formulation spans numerous applications including
image registration, graph matching, community detection, and computational
biology. We develop a unified framework in understanding all problems
of this kind based on representing the available pairwise measurements
as a graph, and then representing the noise on the measurements using
a general channel with a given input/output transition measure. This
framework accommodates large alphabets, general channel transition
probabilities, and general graph structures in a non-asymptotic manner.
Our results underscore the interplay between the minimum channel divergence
measures and the minimum cut size of the measurement graph. Moreover,
for various homogeneous graphs, the recovery criterion relies almost
only on the first-order graphical metrics independent of other second-order
metrics like the spectral gap. We expect that such fundamental recovery
criterion will provide a general benchmark for evaluating the performance
of practical algorithms over many applications.

For concreteness, we restrict our attention to the pairwise difference
model in this paper, but we remark that the analysis framework is
somewhat generic and applies to a broader family of pairwise measurements.
For instance, consider a more general invertible pairwise relation,
denoted by $x_{i}\ominus x_{j}$, that satisfies
\begin{equation}
\begin{cases}
x_{1}\ominus x_{2}\neq x_{1}\ominus x_{3},\quad & \forall x_{2}\neq x_{3};\\
x_{1}\ominus x_{2}\neq x_{4}\ominus x_{2}, & \forall x_{1}\neq x_{4}.
\end{cases}\label{eq:CrossCutAssumption-alpha-1}
\end{equation}
As an example, the addition operator defined as $x_{i}\ominus x_{j}:=ax_{i}+bx_{j}\text{ }(\mathsf{mod}\text{ }M)$
falls within this class as long as both $(a,M)$ and $(b,M)$ are
coprime. Interestingly, most of the analyses carry over to such models
and reappear suitably generalized. Details concerning full generalization
of our results are left for future work. 

While our paper centers on the minimax recovery involving all possible
input configurations, there exists another family of applications
where the inputs fall within a more restricted class (e.g. the class
of inputs whose components are spread out over the entire alphabet).
In addition, it would be interesting to establish how the fundamental
limits can be improved under the partial recovery setting, namely,
the situation where one only demands reconstruction of a (large) fraction
of input variables. Even in the exact recovery situation, it remains
to be seen whether the universal pre-constants can be further tightened.
Moving away from the statistical guarantees, another important issue
is the computational feasibility of information recovery. It has recently
been demonstrated by \cite{chen2016community} that the information
and computation limits meet for many homogeneous graphs (e.g.~rings,
lines, small-world graphs, grids). It would be of great interest to
see whether there exists any computational gap away from the statistical
limits for a more general family of graphs.

\section*{Acknowledgments}

Y. Chen would like to thank Emmanuel Abbe, Afonso Bandeira, Amit Singer
for inspiring discussion on synchronization and graphical channels,
Jiaming Xu for stimulating discussion on hypothesis testing, Amir
Dembo for his valuable instruction on large and moderate deviation
theory. We thank Jonathan Scarlett and I-Hsiang Wang for helpful discussions
and suggestions. Y. Chen and A. J. Goldsmith are supported in part
by the NSF Center for Science of Information, and the AFOSR under
MURI Grant FA9550-12-1-0215. C. Suh is partly supported by the ICT
R\&D program of MSIP/IITP {[}B0101-15-1272, Development of Device
Collaborative Giga-Level Smart Cloudlet Technology{]}. 

\appendix

\section{Proof of Theorem \ref{thm:ML-Hellinger-ERG}\label{sec:Proof-of-Theorem-ML-Hellinger-ERG}}

Suppose that both the ground truth and the null hypothesis are $\boldsymbol{x}=\boldsymbol{x}^{*}$.
Consider the class of alternative hypotheses parametrized by $k$
($1\leq k\leq n$) as follows
\begin{equation}
\mathcal{H}_{k}:=\left\{ \boldsymbol{x}\mid\left\Vert \boldsymbol{x}-\boldsymbol{x}^{*}\right\Vert _{0}=n-k\right\} ,
\end{equation}
which comprises at most ${n \choose k}\left(M-1\right)^{n-k}$ distinct
hypotheses. For notational convenience, denote by $\mathbb{P}_{\boldsymbol{w}}\left(\cdot\right)$
(resp.~$\mathbb{P}_{\boldsymbol{0}}\left(\cdot\right)$) the probability
measure of $\boldsymbol{y}$ conditional on the alternative hypothesis
$\boldsymbol{x}=\boldsymbol{w}$ (resp.~the null hypothesis $\boldsymbol{x}=\boldsymbol{x}^{*}$).
We let $P_{\mathrm{e},\mathcal{H}_{k}}$ represent the probability
of error when restricted to the class $\mathcal{H}_{k}$ of alternative
hypotheses. For simplicity of presentation, we will assume $\boldsymbol{x}^{*}=\boldsymbol{0}$
in what follows, but all steps apply to other choices of $\boldsymbol{x}^{*}$.

For any $\boldsymbol{w}\in\mathcal{H}_{k}$, denote by $\mathcal{S}_{i}$
($0\leq i<M$) the set of vertices $v$ obeying $w_{v}=i$, and let
$n_{i}=\left|\mathcal{S}_{i}\right|$. Apparently, there are $\frac{1}{2}\sum_{i=1}^{M-1}e(\mathcal{S}_{i},\mathcal{S}_{i}^{\mathrm{c}})$
distinct locations $(l,j)$ satisfying $l>j$ and $w_{l}-w_{j}\neq0$,
where $e(\mathcal{S},\mathcal{S}^{\mathrm{c}})$ denotes the number
of cut edges as defined in Section \ref{sub:Notation}. With this
in mind, it follows from the Chernoff bound that 
\begin{eqnarray}
 &  & \mathbb{P}_{\boldsymbol{0}}\left\{ \left.\log\frac{\mathrm{d}\mathbb{P}_{\boldsymbol{w}}\left(\boldsymbol{y}\right)}{\mathrm{d}\mathbb{P}_{\boldsymbol{0}}\left(\boldsymbol{y}\right)}>0\right|\mathcal{E}\right\} =\mathbb{P}_{\boldsymbol{0}}\left\{ \left.\sum_{(l,j)\in\mathcal{E},\text{ }l>j}\alpha\log\frac{\mathrm{d}\mathbb{P}_{\boldsymbol{w}}\left(y_{lj}\right)}{\mathrm{d}\mathbb{P}_{\boldsymbol{0}}\left(y_{lj}\right)}>0\right|\mathcal{E}\right\} \nonumber \\
 &  & \quad\leq\prod_{(l,j)\in\mathcal{E},\text{ }l>j}\mathbb{E}_{{\bf 0}}\left[e^{\alpha\log\frac{\mathrm{d}\mathbb{P}_{\boldsymbol{w}}\left(y_{lj}\right)}{\mathrm{d}\mathbb{P}_{\boldsymbol{0}}\left(y_{lj}\right)}}\right]\\
 &  & \quad=\prod_{(l,j)\in\mathcal{E},\text{ }l>j}\Big[1-\left(1-\alpha\right)\mathsf{Hel}_{\alpha}\left(\mathbb{P}_{\boldsymbol{w}}\left(y_{lj}\right)\hspace{0.2em}\|\hspace{0.2em}\mathbb{P}_{\boldsymbol{0}}\left(y_{lj}\right)\right)\Big]\label{eq:Hellinger-Exp}\\
 &  & \quad=\exp\left(-\left(1-\alpha\right)\sum_{(l,j)\in\mathcal{E},\text{ }l>j}D_{\alpha}\left(\mathbb{P}_{\boldsymbol{w}}\left(y_{lj}\right)\hspace{0.2em}\|\hspace{0.2em}\mathbb{P}_{\boldsymbol{0}}\left(y_{lj}\right)\right)\right)\label{eq:Hel-Exp-Log}\\
 &  & \quad\leq\exp\left(-\left(1-\alpha\right)\frac{\sum_{i=0}^{M-1}e\left(\mathcal{S}_{i},\mathcal{S}_{i}^{\text{c}}\right)}{2}D_{\alpha}^{\min}\right),\label{eq:Bound-Hel}
\end{eqnarray}
where (\ref{eq:Hellinger-Exp}) follows from the definition of the
Hellinger divergence, (\ref{eq:Hel-Exp-Log}) comes from the definition
(\ref{eq:defn-Renyi-Hel}), and (\ref{eq:Bound-Hel}) arises since
$D_{\alpha}\left(\mathbb{P}_{\boldsymbol{w}}\left(y_{lj}\right)\hspace{0.2em}\|\hspace{0.2em}\mathbb{P}_{\boldsymbol{0}}\left(y_{lj}\right)\right)\neq0$
if and only if $w_{l}-w_{j}\neq0$. 

Additionally, define the quantity 
\[
N_{\bm{w}}:=\frac{1}{2}\left|\cup_{i}\left\{ (l,j)\in\left(\mathcal{S}_{i},\mathcal{S}_{i}^{\text{c}}\right)\right\} \right|=\frac{1}{2}\sum_{i=0}^{M-1}\left|\mathcal{S}_{i}\right|\left(n-\left|\mathcal{S}_{i}\right|\right),
\]
then it follows from the definition of $\mathcal{G}_{n,p_{\mathrm{obs}}}$
that
\[
\frac{\sum_{i=0}^{M-1}e\left(\mathcal{S}_{i},\mathcal{S}_{i}^{\mathrm{c}}\right)}{2}\sim\mathsf{Binomial}\left(N_{\bm{w}},p_{\mathrm{obs}}\right).
\]
Unconditioning on $\mathcal{E}$ in the inequality (\ref{eq:Bound-Hel})
gives 
\begin{eqnarray}
\mathbb{P}_{\boldsymbol{0}}\left\{ \log\frac{\mathrm{d}\mathbb{P}_{\boldsymbol{w}}\left(\boldsymbol{y}\right)}{\mathrm{d}\mathbb{P}_{\boldsymbol{0}}\left(\boldsymbol{y}\right)}>0\right\}  & \leq & \mathbb{E}\left[\exp\left(-\left(1-\alpha\right)\frac{\sum_{i=0}^{M-1}e\left(\mathcal{S}_{i},\mathcal{S}_{i}^{\text{c}}\right)}{2}D_{\alpha}^{\min}\right)\right]\nonumber \\
 & = & \sum_{l=0}^{N_{\bm{w}}}{N_{\bm{w}} \choose l}p_{\mathrm{obs}}^{l}\left(1-p_{\mathrm{obs}}\right)^{N_{\bm{w}}-l}\exp\left\{ -l\cdot\left(1-\alpha\right)D_{\alpha}^{\min}\right\} \nonumber \\
 & = & \left(1-p_{\mathrm{obs}}\right)^{N_{\bm{w}}}\sum_{l=0}^{N_{\bm{w}}}{N_{\bm{w}} \choose l}\left(\frac{p_{\mathrm{obs}}}{1-p_{\mathrm{obs}}}\right)^{l}\exp\left\{ -l\cdot\left(1-\alpha\right)D_{\alpha}^{\min}\right\} \nonumber \\
 & \overset{(\text{a})}{=} & \left(1-p_{\mathrm{obs}}\right)^{N_{\bm{w}}}\left(1+\frac{p_{\mathrm{obs}}}{1-p_{\mathrm{obs}}}\exp\left\{ -\left(1-\alpha\right)D_{\alpha}^{\min}\right\} \right)^{N_{\bm{w}}}\nonumber \\
 & = & \left(1-p_{\mathrm{obs}}+p_{\mathrm{obs}}\exp\left\{ -\left(1-\alpha\right)D_{\alpha}^{\min}\right\} \right)^{N_{\bm{w}}}\nonumber \\
 & \overset{(\text{b})}{\leq} & \exp\left\{ -N_{\bm{w}}p_{\mathrm{obs}}\left(1-\exp\left(-\left(1-\alpha\right)D_{\alpha}^{\min}\right)\right)\right\} \nonumber \\
 & \overset{(\text{c})}{=} & \exp\left\{ -N_{\bm{w}}p_{\mathrm{obs}}\left(1-\alpha\right)\mathsf{Hel}_{\alpha}^{\min}\right\} \nonumber \\
 & = & \exp\left\{ -p_{\mathrm{obs}}\left(1-\alpha\right)\mathsf{Hel}_{\alpha}^{\min}\frac{1}{2}\left(n^{2}-\sum_{i=0}^{M-1}n_{i}^{2}\right)\right\} ,\label{eq:prob-error}
\end{eqnarray}
where (a) follows from the binomial theorem, (b) relies on the elementary
inequality $1-x\leq e^{-x}$, (c) comes from the definition (\ref{eq:DefnRenyi-min}),
and the last line follows since 
\[
N_{\bm{w}}=\frac{1}{2}\sum_{i=0}^{M-1}\left|\mathcal{S}_{i}\right|\left(n-\left|\mathcal{S}_{i}\right|\right)=\frac{1}{2}\sum_{i=0}^{M-1}n_{i}\left(n-n_{i}\right)=\frac{1}{2}\left(n^{2}-\sum_{i=0}^{M-1}n_{i}^{2}\right).
\]

It remains to control $\sum_{i=0}^{M-1}n_{i}^{2}$. Recognize that
the input is unique only up to global offset, that is, for any $l$,
the inputs $\boldsymbol{w}$ and $\boldsymbol{w}-l\cdot{\bf 1}$ result
in the same pairwise inputs $\left[w_{i}-w_{j}\right]_{1\leq i,j\leq n}$.
Therefore, we assume without loss of generality that\footnote{Otherwise, if $n_{i}=\max\left\{ n_{1},n_{2},\cdots,n_{M-1}\right\} $
instead, we can always enforce a global shift $i$ on $\boldsymbol{w}$
to yield $\boldsymbol{w}-i{\bf 1}$ in order to satisfy this condition
without affecting the output distribution. } 
\begin{equation}
k=n_{0}\geq\max\left\{ n_{1},n_{2},\cdots,n_{M-1}\right\} .\label{eq:defn-k}
\end{equation}
Letting $\rho:=\left\lfloor \frac{n}{k}\right\rfloor $, we claim
that $\sum_{i=0}^{M-1}n_{i}^{2}$ under the constraint $n_{i}\leq k$
is maximized by the configuration 
\[
\begin{cases}
n_{0}=n_{1}=\cdots=n_{\rho-1}=k,\\
n_{\rho}=n-k\rho,\\
n_{\rho+1}=\cdots=n_{M-1}=0,
\end{cases}
\]
which we will prove by contradiction. Without loss of generality,
suppose that the maximizing solution is $n_{0}\geq n_{1}\geq\cdots\geq n_{M-1}$,
and denote by $\tilde{\rho}$ the smallest index such that $n_{\tilde{\rho}}\leq k-1$.
If $\tilde{\rho}\leq\rho-1$, then by replacing $\left(n_{\tilde{\rho}},n_{\tilde{\rho}+1}\right)$
with $\left(n_{\tilde{\rho}}+1,n_{\tilde{\rho}+1}-1\right)$, we obtain
a strictly better feasible solution since
\begin{eqnarray*}
\left(n_{\tilde{\rho}}+1\right)^{2}+\left(n_{\tilde{\rho}+1}-1\right)^{2} & = & n_{\tilde{\rho}}^{2}+n_{\tilde{\rho}+1}^{2}+2\left(n_{\tilde{\rho}}-n_{\tilde{\rho}+1}\right)+2>n_{\tilde{\rho}}^{2}+n_{\tilde{\rho}+1}^{2}.
\end{eqnarray*}
This results in contradiction, and hence $\tilde{\rho}=\rho$. Similarly,
we cannot have $n_{\rho}<n-k\rho$, since replacing $\left(n_{\rho},n_{\rho+1}\right)$
with $\left(n_{\rho}+1,n_{\rho+1}-1\right)$ leads to a strictly better
solution. Consequently, for all $\left\{ n_{i}:0\leq i<M\right\} $
satisfying (\ref{eq:defn-k}), one has 
\begin{equation}
\sum_{i=0}^{M-1}n_{i}^{2}\leq\left\lfloor \frac{n}{k}\right\rfloor \cdot k^{2}+\left(n-k\left\lfloor \frac{n}{k}\right\rfloor \right)^{2},\label{eq:UB-sum-square}
\end{equation}
leaving us two cases below to deal with.

\textbf{Case 1}. Suppose that $k\leq n/2$. The inequality $n-k\left\lfloor \frac{n}{k}\right\rfloor \leq k$
leads to
\[
\sum_{i=0}^{M-1}n_{i}^{2}\leq\left\lfloor \frac{n}{k}\right\rfloor \cdot k^{2}+\left(n-k\left\lfloor \frac{n}{k}\right\rfloor \right)k=nk.
\]
This combined with (\ref{eq:prob-error}) yields
\begin{eqnarray*}
\mathbb{P}_{\boldsymbol{0}}\left\{ \log\frac{\mathrm{d}\mathbb{P}_{\boldsymbol{w}}\left(\boldsymbol{y}\right)}{\mathrm{d}\mathbb{P}_{\boldsymbol{0}}\left(\boldsymbol{y}\right)}>0\right\}  & \leq & \exp\left(-\frac{p_{\mathrm{obs}}\left(n^{2}-nk\right)}{2}\left(1-\alpha\right)\mathsf{Hel}_{\alpha}^{\min}\right).
\end{eqnarray*}
Employing the union bound over $\mathcal{H}_{k}$ we obtain
\begin{eqnarray*}
P_{\mathrm{e},\mathcal{H}_{k}} & \leq & {n \choose k}\left(M-1\right)^{n-k}\exp\left(-\frac{p_{\mathrm{obs}}\left(n^{2}-nk\right)}{2}\left(1-\alpha\right)\mathsf{Hel}_{\alpha}^{\min}\right)\\
 & = & \exp\left(\log{n \choose k}+(n-k)\log\left(M-1\right)-\frac{p_{\mathrm{obs}}\left(n^{2}-nk\right)}{2}\left(1-\alpha\right)\mathsf{Hel}_{\alpha}^{\min}\right).
\end{eqnarray*}
Under the assumption (\ref{eq:Achievability-Hel}), one has 
\begin{eqnarray*}
\left(1-\alpha\right)\mathsf{Hel}_{\alpha}^{\min}\cdot p_{\mathrm{obs}}n & \geq & \log n+2\log\left(M-1\right)
\end{eqnarray*}
for some $0<\alpha<1$, which further gives
\begin{eqnarray*}
\frac{p_{\mathrm{obs}}\left(n^{2}-nk\right)}{2}\left(1-\alpha\right)\mathsf{Hel}_{\alpha}^{\min} & \geq & \frac{\left(n-k\right)\log n}{2}+(n-k)\log\left(M-1\right).
\end{eqnarray*}
Putting the above computation together yields
\begin{eqnarray}
P_{\mathsf{e},\mathcal{H}_{k}} & \leq & \exp\left(\log{n \choose k}-\frac{\left(n-k\right)\log n}{2}\right)\overset{(\text{i})}{\leq}2^{n}\cdot n^{-\frac{1}{2}\left(n-k\right)}\text{ }\overset{(\text{ii})}{\leq}\text{ }2^{n}\cdot n^{-\frac{1}{4}n}\nonumber \\
 & \leq & C_{1}n^{-c_{1}n}\label{eq:Pe-small-k}
\end{eqnarray}
for some universal constants $C_{1},c_{1}>0$, where (i) uses the
fact that ${n \choose k}\leq2^{n}$, (ii) holds since $k\leq n/2$,
and the last inequality follows since $2^{n}\ll n^{-\Theta\left(n\right)}$.
This approaches zero (super)-exponentially fast.

\textbf{Case 2}. We now move on to the case where $k>n/2$. In this
regime one has $\left\lfloor \frac{n}{k}\right\rfloor =1$, and thus
(\ref{eq:UB-sum-square}) gives 
\[
\sum\nolimits{}_{i=0}^{M-1}n_{i}^{2}\leq k^{2}+\left(n-k\right)^{2}=n^{2}-2k\left(n-k\right),
\]
This taken collectively with (\ref{eq:prob-error}) implies that
\begin{eqnarray}
\mathbb{P}_{\boldsymbol{0}}\left\{ \log\frac{\mathrm{d}\mathbb{P}_{\boldsymbol{w}}\left(\boldsymbol{y}\right)}{\mathrm{d}\mathbb{P}_{\boldsymbol{0}}\left(\boldsymbol{y}\right)}>0\right\}  & \leq & \exp\left(-p_{\mathrm{obs}}\left(1-\alpha\right)k\left(n-k\right)\mathsf{Hel}_{\alpha}^{\min}\right).
\end{eqnarray}
Apply the union bound over $\mathcal{H}_{k}$ to deduce that
\begin{eqnarray}
P_{\mathsf{e},\mathcal{H}_{k}} & \leq & {n \choose k}\left(M-1\right)^{n-k}\exp\left(-p_{\mathrm{obs}}\left(1-\alpha\right)k\left(n-k\right)\mathsf{Hel}_{\alpha}^{\min}\right)\nonumber \\
 & = & \exp\left(\log{n \choose k}+(n-k)\log\left(M-1\right)-p_{\mathrm{obs}}\left(1-\alpha\right)k\left(n-k\right)\mathsf{Hel}_{\alpha}^{\min}\right).\label{eq:Perror-large-k}
\end{eqnarray}
For any constant $\delta>0$, if the minimum Hellinger divergence
obeys
\begin{eqnarray*}
\left(1-\alpha\right)\mathsf{Hel}_{\alpha}^{\min}\cdot p_{\mathrm{obs}}n & \geq & \left(1+\delta\right)\log\left(2n\right)+2\log\left(M-1\right)
\end{eqnarray*}
for some $0<\alpha<1$, then in the regime where $k>n/2$ one has
\begin{eqnarray*}
p_{\mathrm{obs}}k\left(n-k\right)\left(1-\alpha\right)\mathsf{Hel}_{\alpha}^{\min} & \geq & \frac{\left(1+\delta\right)k\left(n-k\right)\log\left(2n\right)}{n}+\left(n-k\right)\log\left(M-1\right).
\end{eqnarray*}
Substitution into (\ref{eq:Perror-large-k}) gives
\begin{eqnarray*}
P_{\mathsf{e},\mathcal{H}_{k}} & \leq & \exp\left(\log{n \choose k}-\frac{\left(1+\delta\right)k\left(n-k\right)}{n}\log\left(2n\right)\right),
\end{eqnarray*}
which can be further divided into two cases. 

\begin{enumerate}[(i)] 

\item If $k/n>1-\delta/4$ and $k/n>1/2$, then the error probability
is bounded by
\begin{eqnarray*}
P_{\mathsf{e},\mathcal{H}_{k}} & \leq & \exp\left(\left(n-k\right)\log\left(2n\right)-\frac{\left(1+\delta\right)k}{n}\left(n-k\right)\log(2n)\right)\\
 & = & \exp\left(\left(n-k\right)\log\left(2n\right)\cdot\left(1-\left(1+\delta\right)k/n\right)\right)\\
 & \leq & \exp\left(\left(n-k\right)\log\left(2n\right)\cdot\left(1-\left(1+\delta\right)\max\left\{ 1-\frac{\delta}{4},\text{ }\frac{1}{2}\right\} \right)\right)\\
 & \leq & \exp\left(-\tilde{\delta}\left(n-k\right)\log\left(2n\right)\right),
\end{eqnarray*}
where $\tilde{\delta}:=\max\left\{ \frac{3}{4}\delta-\frac{1}{4}\delta^{2},\text{ }\frac{\delta-1}{2}\right\} $.

\item If $\frac{k}{n}=1-\tau$ for some $\frac{\delta}{4}\leq\tau\leq\frac{1}{2}$,
then 
\begin{eqnarray}
P_{\mathsf{e},\mathcal{H}_{k}} & \leq & \exp\left(nH\left(\tau\right)-\left(1+\delta\right)n\left(1-\tau\right)\tau\log\left(2n\right)\right)\label{eq:case2-i}\\
 & \leq & \exp\left(-n\left(\left(1-\tau\right)\tau\log\left(2n\right)-H\left(\tau\right)\right)\right)\nonumber \\
 & \leq & C_{2}\exp\left(-c_{2}\delta n\log n\right).
\end{eqnarray}
for some universal constants $c_{2},C_{2}>0$, where (\ref{eq:case2-i})
makes use of the fact \cite[Example 11.1.3]{cover2006elements} that
$\frac{1}{n}\log{n \choose k}\leq H(k/n)=H(\tau)$, with $H(\tau)$
denoting the binary entropy function. 

\end{enumerate}

Putting the above inequalities together and applying the union bound
reveal that 
\begin{eqnarray*}
P_{\mathsf{e}} & \leq & \sum_{k=\left\lceil \frac{n}{M}\right\rceil }^{n/2}P_{\mathsf{e},\mathcal{H}_{k}}+\sum_{k=n/2+1}^{\left(1-\frac{\delta}{4}\right)n}P_{\mathsf{e},\mathcal{H}_{k}}+\sum_{k=\left(1-\frac{\delta}{4}\right)n}^{n-1}P_{\mathsf{e},\mathcal{H}_{k}}\\
 & \leq & \frac{n}{2}\cdot C_{1}n^{-c_{1}n}+\frac{n}{2}\cdot C_{2}\exp\left(-c_{2}\delta n\log n\right)+\sum_{k=\left(1-\frac{\delta}{4}\right)n}^{n-1}\exp\left(-\tilde{\delta}\left(n-k\right)\log\left(2n\right)\right)\\
 & \leq & \frac{n}{2}\cdot C_{1}n^{-c_{1}n}+\frac{n}{2}\cdot C_{2}\exp\left(-c_{2}\delta n\log n\right)+\frac{1}{\left(2n\right)^{\tilde{\delta}}}\frac{1}{1-\left(2n\right)^{-\tilde{\delta}}}\\
 & \leq & C_{0}e^{-c_{0}\delta n\log n}+\frac{1}{\left(2n\right)^{\tilde{\delta}}-1}.
\end{eqnarray*}
with $c_{0},C_{0}>0$ denoting some universal constants.

\section{Proof of Theorems \ref{corollary:ConverseHellinger-ERG}(a) and \ref{thm:ConverseM-Degree-KL}\label{sec:Proof-of-Theorem-ConverseM-Degree}}

This section is mainly devoted to proving Theorem \ref{thm:ConverseM-Degree-KL},
which subsumes Theorem \ref{corollary:ConverseHellinger-ERG}(a) as
a special case. Without loss of generality, assume that the minimum
KL divergence can be approached by the following pairs of indices
\begin{eqnarray*}
\mathsf{KL}\left(\mathbb{P}_{1}\hspace{0.2em}\|\hspace{0.2em}\mathbb{P}_{0}\right) & = & \mathsf{KL}^{\min},\\
\mathsf{KL}\left(\mathbb{P}_{l}\hspace{0.2em}\|\hspace{0.2em}\mathbb{P}_{0}\right) & \leq & \left(1+\zeta\right)\mathsf{KL}^{\min},\quad\text{ }\quad2\leq l\leq m^{\mathsf{kl}}\left(\zeta\right),
\end{eqnarray*}
and suppose that both the ground truth and the null hypothesis are
$\boldsymbol{x}=\boldsymbol{x}^{*}=\boldsymbol{0}$. We would like
to ensure that the observation $\boldsymbol{y}$ conditional on $\boldsymbol{x}=\boldsymbol{0}$
is distinguishable from the observation $\boldsymbol{y}$ under any
alternative hypothesis $\boldsymbol{x}\neq\boldsymbol{0}$. 

(1) To begin with, recall the definition 
\[
\mathcal{N}\left(k\cdot\mincut\right):=\left\{ \mathcal{S}\subseteq\mathcal{V}\text{: }e\left(\mathcal{S},\mathcal{S}^{\mathrm{c}}\right)\leq k\cdot\mincut\right\} .
\]
For each vertex set $\mathcal{S}\in\mathcal{N}\left(k\cdot\mincut\right)$,
we generate one representative hypothesis $\boldsymbol{w}$ such that
\[
w_{i}=\begin{cases}
1,\quad & \text{if }i\in\mathcal{S},\\
0,\quad & \text{otherwise}.
\end{cases}
\]
This produces a collection of $|\mathcal{N}\left(k\cdot\mincut\right)|$
distinct alternative hypotheses, denoted by $\mathcal{B}_{k}$. For
each $\boldsymbol{w}\in\mathcal{B}_{k}$, the distributions $\mathbb{P}_{\boldsymbol{w}}$
and $\mathbb{P}_{\boldsymbol{0}}$ disagree only over those locations
residing in the associated cut set, which amounts to at most $k\cdot\mincut$
components. It then follows from the independence assumption of $y_{ij}$
that 
\begin{equation}
\mathsf{KL}\left(\mathbb{P}_{\boldsymbol{w}}\hspace{0.2em}\|\hspace{0.2em}\mathbb{P}_{\boldsymbol{0}}\right)=e\left(\mathcal{S},\mathcal{S}^{\mathrm{c}}\right)\mathsf{KL}^{\min}\leq k\cdot\mincut\cdot\mathsf{KL}^{\min}.\label{eq:KL-ub}
\end{equation}

Suppose that $k_{0}:=\arg\max_{k\geq1}\tau_{k}^{\mathrm{cut}}$ and
fix $0<\epsilon\leq\frac{1}{2}$. Applying the Fano-type inequality
\cite[Equation (2.70)]{tsybakov2009introduction} suggests that if
\begin{eqnarray}
\frac{1}{\left|\mathcal{B}_{k_{0}}\right|}\sum_{\boldsymbol{w}\in\mathcal{B}_{k_{0}}}\mathsf{KL}\left(\mathbb{P}_{\boldsymbol{w}}\hspace{0.2em}\|\hspace{0.2em}\mathbb{P}_{\boldsymbol{0}}\right) & \leq & \left(1-\epsilon\right)\log\Big|\mathcal{N}\left(k_{0}\cdot\mincut\right)\Big|-H\left(\epsilon\right),\label{eq:KL}
\end{eqnarray}
then one necessarily has $\inf_{\psi}P_{\mathrm{e}}\left(\psi\right)\geq\epsilon$.
With (\ref{eq:KL-ub}) and the definition (\ref{eq:DefnExponent})
in mind, we see that (\ref{eq:KL}) would follow from 
\[
\mathsf{KL}^{\min}\cdot k_{0}\mincut\leq\left(1-\epsilon\right)k_{0}\tau^{\mathrm{cut}}-H\left(\epsilon\right),
\]
which can further be ensured if
\[
\mathsf{KL}^{\min}\cdot\mincut\leq\left(1-\epsilon\right)\tau^{\mathrm{cut}}-H\left(\epsilon\right).
\]

(2) Next, suppose that the minimum cut is attained by $\left(\mathcal{S}_{\mathsf{mc}},\mathcal{S}_{\mathsf{mc}}^{\mathrm{c}}\right)$.
Consider another class $\mathcal{C}$ of hypotheses consisting of
$m^{\mathsf{kl}}$ hypotheses. The $l$th candidate $\boldsymbol{w}^{(l)}$
is given by
\[
\forall1\leq l\leq m^{\mathsf{kl}}:\quad w_{i}^{(l)}=\begin{cases}
l,\quad & \text{if }i\in\mathcal{S}_{\mathsf{mc}},\\
0,\quad & \text{otherwise},
\end{cases}
\]
all of which obey 
\begin{equation}
\mathsf{KL}\left(\boldsymbol{w}^{(l)}\|\boldsymbol{0}\right)\leq\left(1+\zeta\right)\mincut\cdot\mathsf{KL}^{\min}.\label{eq:KL-mincut}
\end{equation}
Applying the Fano inequality once again, we get $\inf_{\psi}P_{\mathrm{e}}\left(\psi\right)\geq\epsilon$
as long as 
\begin{equation}
\frac{1}{m^{\mathsf{kl}}}\sum_{l=1}^{m^{\mathsf{kl}}}\mathsf{KL}\left(\boldsymbol{w}^{(l)}\|\boldsymbol{0}\right)\leq\left(1-\epsilon\right)\log m^{\mathsf{kl}}-H\left(\epsilon\right).\label{eq:H-KL}
\end{equation}
Observe from (\ref{eq:KL-mincut}) that (\ref{eq:H-KL}) can be ensured
under the condition
\[
\mathsf{KL}^{\min}\cdot\left(1+\zeta\right)\mincut\leq\left(1-\epsilon\right)\log m^{\mathsf{kl}}-H\left(\epsilon\right).
\]

(3) Finally, consider the set of configurations with \emph{binary}
alphabet having support size $1$, i.e.~the following $M-1$ classes
of hypotheses 
\begin{equation}
\mathcal{H}_{l}:=\left\{ \boldsymbol{x}\mid\left\Vert \boldsymbol{x}\right\Vert _{0}=1,\text{ }\boldsymbol{x}\in\left\{ 0,l\right\} ^{n}\right\} ,\quad1\leq l<M,\label{eq:DefnClass}
\end{equation}
where each class $\mathcal{H}_{l}$ is composed of $n$ distinct alternative
hypotheses. This guarantees that for any $\boldsymbol{w}\in\mathcal{H}_{l}$,
the distribution of $\left\{ y_{ij}\right\} \mid\boldsymbol{x}=\boldsymbol{0}$
differ from that of $\left\{ y_{ij}\right\} \mid\boldsymbol{x}=\boldsymbol{w}$
in at most $d_{\max}$ locations.

For any hypothesis class $\mathcal{H}$ and any $0<\epsilon<\frac{1}{2}$,
the Fano-type inequality \cite[Equation (2.70)]{tsybakov2009introduction}
suggests that $\inf_{\psi}P_{\mathsf{e}}\left(\psi\right)\geq\epsilon$
occurs as long as 
\begin{eqnarray}
\frac{1}{\left|\mathcal{H}\right|}\sum_{\boldsymbol{w}\in\mathcal{H}}\mathsf{KL}\left(\mathbb{P}_{\boldsymbol{w}}\hspace{0.2em}\|\hspace{0.2em}\mathbb{P}_{\boldsymbol{0}}\right) & \leq & \frac{\left|\mathcal{H}\right|+1}{\left|\mathcal{H}\right|}\left\{ \left(1-\epsilon\right)\log\left|\mathcal{H}\right|-H\left(\epsilon\right)\right\} .\label{eq:Fano-KL}
\end{eqnarray}
By picking $\mathcal{H}$ to be $\mathcal{H}=\bigcup_{l=1}^{m^{\mathsf{kl}}}\mathcal{H}_{l}$---which
obeys $\left|\mathcal{H}\right|=m^{\mathsf{kl}}n$---we can see from
definition of $m^{\mathsf{kl}}$ that
\[
\frac{1}{\left|\mathcal{H}\right|}\sum_{\boldsymbol{w}\in\mathcal{H}}\mathsf{KL}\left(\mathbb{P}_{\boldsymbol{w}}\hspace{0.2em}\|\hspace{0.2em}\mathbb{P}_{\boldsymbol{0}}\right)\leq\left(1+\zeta\right)d_{\max}\mathsf{KL}^{\min}
\]
and hence (\ref{eq:Fano-KL}) would hold under the condition
\begin{eqnarray}
\mathsf{KL}^{\min}\cdot\left(1+\zeta\right)d_{\max} & \leq & \left(1-\epsilon\right)\left(\log n+\log m^{\mathsf{kl}}\right)-H\left(\epsilon\right).\label{eq:KL-condition-dmax}
\end{eqnarray}

Putting the above results together establishes Theorem \ref{thm:ConverseM-Degree-KL}. 

We now specialize to Theorem \ref{corollary:ConverseHellinger-ERG}(a),
which follows immediately from (\ref{eq:KL-condition-dmax}). Specifically,
for an Erd\H{o}s\textendash Rényi graph $\mathcal{G}\sim\mathcal{G}_{n,p_{\mathrm{obs}}}$,
the Chernoff-type inequality \cite[Theorems 4.4-4.5]{mitzenmacher2005probability}
indicates that for any $\epsilon>0$, 
\begin{equation}
d_{\max}\leq\left(1+\epsilon\right)np_{\mathrm{obs}}\label{eq:d_max_G}
\end{equation}
holds with probability exceeding $1-n^{-10}$, provided that $p_{\mathrm{obs}}>\frac{c\log n}{n}$
for some sufficiently large constant $c>0$. Substitution into (\ref{eq:KL-condition-dmax})
immediately leads to Theorem \ref{corollary:ConverseHellinger-ERG}(a).

\section{Proof of Theorems \ref{corollary:ConverseHellinger-ERG}(b) and \ref{thm:ConverseHellinger}\label{sec:Proof-of-Theorem-Converse-ERG}}

We start with the proof of Theorem \ref{thm:ConverseHellinger}, which
accounts for a much broader context than Theorem \ref{corollary:ConverseHellinger-ERG}(b).
In similar spirit of Theorem \ref{thm:ConverseM-Degree-KL}, assume
that the minimum Hellinger divergence is achieved by the following
pair of indices
\begin{eqnarray*}
\mathsf{Hel}_{\alpha}\left(\mathbb{P}_{1}\hspace{0.2em}\|\hspace{0.2em}\mathbb{P}_{0}\right) & = & \mathsf{Hel}_{\alpha}^{\min},
\end{eqnarray*}
and let the ground truth and the null hypothesis be $\boldsymbol{x}=\boldsymbol{x}^{*}=\boldsymbol{0}$. 

For any class $\mathcal{H}$ of alternative hypotheses, the minimax
lower bound \cite[Theorem II.1]{guntuboyina2011lower} suggests that
every $f$-divergence $D_{f}\left(\cdot\right)$ obeys
\[
\sum_{\boldsymbol{w}\in\mathcal{H}}D_{f}\left(\mathbb{P}_{\boldsymbol{w}}\hspace{0.2em}\|\hspace{0.2em}\mathbb{P}_{\boldsymbol{0}}\right)\geq f\left(\left|\mathcal{H}\right|\left(1-P_{\mathrm{e}}\right)\right)+\left(\left|\mathcal{H}\right|-1\right)f\left(\frac{\left|\mathcal{H}\right|P_{\mathrm{e}}}{\left|\mathcal{H}\right|-1}\right),
\]
where $\mathbb{P}_{\boldsymbol{w}}$ is the probability measure of
$\left[y_{ij}\right]_{(i,j)\in\mathcal{E}}$ conditional on $\boldsymbol{x}=\boldsymbol{w}$.
When specialized to the Hellinger divergence of order $\alpha$ (which
corresponds to $f(x)=\frac{1}{1-\alpha}\left(1-x^{\alpha}\right)$),
the above inequality leads to
\begin{eqnarray*}
\left(1-\alpha\right)\sum_{\boldsymbol{w}\in\mathcal{H}}\mathsf{Hel}_{\alpha}\left(\mathbb{P}_{\boldsymbol{w}}\hspace{0.2em}\|\hspace{0.2em}\mathbb{P}_{\boldsymbol{0}}\right) & \geq & 1-\left|\mathcal{H}\right|^{\alpha}\left(1-P_{\mathrm{e}}\right)^{\alpha}+\left(\left|\mathcal{H}\right|-1\right)\left\{ 1-\left(\frac{\left|\mathcal{H}\right|P_{\mathrm{e}}}{\left|\mathcal{H}\right|-1}\right)^{\alpha}\right\} \\
 & = & \left|\mathcal{H}\right|-\left|\mathcal{H}\right|^{\alpha}\left(1-P_{\mathrm{e}}\right)^{\alpha}-\left(\left|\mathcal{H}\right|-1\right)^{1-\alpha}\left|\mathcal{H}\right|^{\alpha}P_{\mathrm{e}}^{\alpha}\\
 & \geq & \left|\mathcal{H}\right|-\left|\mathcal{H}\right|^{\alpha}-\left|\mathcal{H}\right|P_{\mathrm{e}}^{\alpha}.
\end{eqnarray*}
Put another way,
\begin{eqnarray}
P_{\mathrm{e}}^{\alpha} & \geq & 1-\frac{\left(1-\alpha\right)\sum_{\boldsymbol{w}\in\mathcal{H}}\mathsf{Hel}_{\alpha}\left(\mathbb{P}_{\boldsymbol{w}}\hspace{0.2em}\|\hspace{0.2em}\mathbb{P}_{\boldsymbol{0}}\right)}{\left|\mathcal{H}\right|}-\frac{1}{\left|\mathcal{H}\right|^{1-\alpha}}.\label{eq:pe_Hel}
\end{eqnarray}

Notably, for any product measures $P^{n}=P\times P\times\cdots\times P$
and $Q^{n}=Q\times Q\times\cdots\times Q$, the Hellinger divergence
satisfies the decoupling equality
\begin{eqnarray}
1-\left(1-\alpha\right)\mathsf{Hel}_{\alpha}\left(P^{n}\hspace{0.2em}\|\hspace{0.2em}Q^{n}\right) & = & {\displaystyle \int}\left(\mathrm{d}P^{n}\right)^{\alpha}\left(\mathrm{d}Q^{n}\right)^{1-\alpha}=\left({\displaystyle \int}\left(\mathrm{d}P\right)^{\alpha}\left(\mathrm{d}Q\right)^{1-\alpha}\right)^{n}\nonumber \\
 & = & \left(1-\left(1-\alpha\right)\mathsf{Hel}_{\alpha}\left(P\hspace{0.2em}\|\hspace{0.2em}Q\right)\right)^{n}.\label{eq:decoupling}
\end{eqnarray}
If all hypotheses $\boldsymbol{w}\in\mathcal{H}$ satisfy $\left\Vert \boldsymbol{w}-\boldsymbol{x}^{*}\right\Vert _{0}\leq k$,
then $\mathbb{P}_{\boldsymbol{w}}$ and $\mathbb{P}_{\boldsymbol{0}}$
are different over at most $kd_{\max}$ locations. Thus, if the divergence
measure at each of these locations is identical and equal to some
given value $h_{\alpha}$, then it follows from the independence assumption
of $y_{ij}$ that
\begin{eqnarray*}
1-\left(1-\alpha\right)\mathsf{Hel}_{\alpha}\left(\mathbb{P}_{\boldsymbol{w}}\hspace{0.2em}\|\hspace{0.2em}\mathbb{P}_{\boldsymbol{0}}\right) & \geq & \Big(1-\left(1-\alpha\right)h_{\alpha}\Big)^{kd_{\max}}.
\end{eqnarray*}
This together with (\ref{eq:pe_Hel}) suggests that: as long as $\left(1-\alpha\right)h_{\alpha}\leq\frac{1}{2}$,
one necessarily has 
\begin{eqnarray*}
P_{\mathrm{e}}^{\alpha} & \geq & \Big(1-\left(1-\alpha\right)h_{\alpha}\Big)^{kd_{\max}}-\left|\mathcal{H}\right|^{-\left(1-\alpha\right)}\\
 & \geq & e{}^{-\left(\left(1-\alpha\right)h_{\alpha}+\left(1-\alpha\right)^{2}h_{\alpha}^{2}\right)kd_{\max}}-\left|\mathcal{H}\right|^{-\left(1-\alpha\right)},
\end{eqnarray*}
which results from the inequality that $\log\left(1-x\right)\geq-x-x^{2}$
for any $0\leq x\leq1/2$.

As a consequence, if the following condition holds
\[
e{}^{-\left(\left(1-\alpha\right)h_{\alpha}+\left(1-\alpha\right)^{2}h_{\alpha}^{2}\right)kd_{\max}}-\left|\mathcal{H}\right|^{-\left(1-\alpha\right)}\text{ }\geq\text{ }\xi^{\alpha}
\]
or, equivalently,
\begin{equation}
\left(1-\alpha\right)h_{\alpha}\left[1+\left(1-\alpha\right)h_{\alpha}\right]\leq-\frac{\log\left(\xi^{\alpha}+\left|\mathcal{H}\right|^{-\left(1-\alpha\right)}\right)}{kd_{\max}},\label{eq:Hel_UB}
\end{equation}
then the minimax probability of error must exceed $\inf_{\psi}P_{\mathrm{e}}\geq\xi$.
Solving the quadratic inequality (\ref{eq:Hel_UB}) and utilizing
the fact $\sqrt{1+4x}-1\geq2x-4x^{2}$ $(x\geq0)$, we see that (\ref{eq:Hel_UB})
would follow as long as
\begin{equation}
\left(1-\alpha\right)h_{\alpha}\leq-\frac{\log\left(\xi^{\alpha}+\left|\mathcal{H}\right|^{-\left(1-\alpha\right)}\right)}{kd_{\max}}-\frac{2\log^{2}\left(\xi^{\alpha}+\left|\mathcal{H}\right|^{-\left(1-\alpha\right)}\right)}{\left(kd_{\max}\right)^{2}}.\label{eq:Hel-condition}
\end{equation}

Finally, setting $\xi=n^{-\epsilon}$ and $\mathcal{H}=\mathcal{H}_{1}$
(cf. Definition (\ref{eq:DefnClass})), one has $\left|\mathcal{H}\right|=n$,
$k=1$ and $h_{\alpha}=\mathsf{Hel}_{\alpha}^{\min}$. In the regime
where 
\[
\epsilon\leq\frac{1-\alpha}{\alpha}\quad\Longleftrightarrow\quad\alpha\leq\frac{1}{1+\epsilon},
\]
we have
\[
\xi^{\alpha}=n^{-\epsilon\alpha}\geq n^{-\left(1-\alpha\right)}=\left|\mathcal{H}\right|^{-\left(1-\alpha\right)}.
\]
The condition (\ref{eq:Hel-condition}) is then guaranteed to hold
if

\begin{equation}
\left(1-\alpha\right)\mathsf{Hel}_{\alpha}^{\min}\leq-\frac{\log\left(2\xi^{\alpha}\right)}{d_{\max}}-\frac{2\log^{2}\left(2\xi^{\alpha}\right)}{d_{\max}^{2}},\label{eq:H-ub}
\end{equation}
which would follow if 
\begin{equation}
\Longleftarrow\quad\left(1-\alpha\right)\mathsf{Hel}_{\alpha}^{\min}\leq\frac{\epsilon\alpha\log n-\log2}{d_{\max}}-\frac{2\left[\epsilon\alpha\log n-\log2\right]^{2}}{d_{\max}^{2}},\label{eq:H-ub-2}
\end{equation}
where we have used $\xi^{\alpha}=n^{-\epsilon\alpha}$. Besides, the
condition $\left(1-\alpha\right)h_{\alpha}\leq\frac{1}{2}$ becomes
$\left(1-\alpha\right)\mathsf{Hel}_{\alpha}^{\min}\leq\frac{1}{2}$,
which can be ensured under (\ref{eq:H-ub-2}) together with the condition
\[
\frac{\epsilon\alpha\log n}{d_{\max}}\leq\frac{1}{2}
\]
as claimed.

Finally, recall that when $\mathcal{G}\sim\mathcal{G}_{n,p_{\mathrm{obs}}}$,
one has $d_{\max}\leq(1+\epsilon)p_{\mathrm{obs}}n$ as long as $np_{\mathrm{obs}}/\log n$
is sufficiently large. Plugging this into the preceding bound completes
the proof of Theorem \ref{corollary:ConverseHellinger-ERG}(b).

\section{Proof of Theorem \ref{thm:AchieveM-Deg}\label{sec:Proof-of-Theorem-AchieveM-Deg}}

Note that $\psi_{\mathrm{ml}}$ distinguishes the null hypothesis
$\boldsymbol{x}=\boldsymbol{x}^{*}=\left\{ x_{i}^{*}\right\} _{1\leq i\leq n}$
from the alternative hypothesis $\boldsymbol{x}=\boldsymbol{w}=\left\{ w_{i}\right\} _{1\leq i\leq n}$
only based on those components ($i,j$) where
\[
x_{i}^{*}-x_{j}^{*}\neq w_{i}-w_{j},
\]
and its recovery capability depends only on the distinction of output
distributions over these locations. For ease of presentation, we will
suppose in the rest of the proof that both the ground truth and the
null hypothesis are $\boldsymbol{x}=\boldsymbol{0}$, but note that
the proof carries over to all other ground truth values. 

Let's divide the set of all alternative hypotheses into several classes
$\mathcal{A}_{k}$ so that for each $k\geq1$,
\begin{equation}
\mathcal{A}_{k}:=\left\{ \boldsymbol{w}\neq{\bf 0}:\text{ }\left|\mathcal{E}\cap\mathrm{supp}\left(\boldsymbol{w}\ominus\boldsymbol{w}\right)\right|<k\cdot\mincut\right\} ,\label{eq:Defn-Ak}
\end{equation}
where we employ the notation 
\[
\mathcal{E}\cap\mathrm{supp}\left(\boldsymbol{w}\ominus\boldsymbol{w}\right):=\left\{ \left(i,j\right)\in\mathcal{E}\mid w_{i}-w_{j}\neq0,\text{ }\text{ }i>j\right\} .
\]
Apparently, any cut set cannot contain more than $n^{2}$ edges, and
hence $\mathcal{A}_{k}=\emptyset$ for any $k\geq n^{2}/\mincut$.
For any $\boldsymbol{w}\in\mathcal{A}_{k}$, if we let $\mathcal{S}_{l}$
represent the set of vertices taking the value $l$, then by definition
of $\mathcal{A}_{k}$ one has 
\begin{equation}
\sum_{l=0}^{M-1}e\left(\mathcal{S}_{l},\mathcal{S}_{l}^{\mathrm{c}}\right)=2\left|\mathcal{E}\cap\mathrm{supp}\left(\boldsymbol{w}\ominus\boldsymbol{w}\right)\right|<2k\cdot\mincut.\label{eq:feasibility}
\end{equation}
On the other hand, consider the case where $k=1$. All $\boldsymbol{w}\in\mathcal{A}_{1}$
are equivalent to ${\bf 0}$ up to some global offset. This is because
for any non-trivial cut $(\mathcal{S}_{l},\mathcal{S}_{l}^{\mathrm{c}})$,
one must have $\left|\mathcal{E}\cap\mathrm{supp}\left(\boldsymbol{w}\ominus\boldsymbol{w}\right)\right|\geq e\left(\mathcal{S}_{l},\mathcal{S}_{l}^{\mathrm{c}}\right)\geq\mincut$,
which violates the feasibility constraint $\left|\mathcal{E}\cap\mathrm{supp}\left(\boldsymbol{w}\ominus\boldsymbol{w}\right)\right|<\mincut$.
In the following lemma, we link the cardinality of each hypothesis
class $\mathcal{A}_{k}$ with the cut-homogeneity exponent $\tau^{\mathrm{cut}}$
defined in (\ref{eq:DefnExponent}).

\begin{lem}\label{lemma:cut-size-alpha-beta}For any $k\leq n^{2}/\mincut$,
the hypothesis class $\mathcal{A}_{k}$ defined in (\ref{eq:Defn-Ak})
satisfies
\begin{eqnarray}
\frac{\log\left|\mathcal{A}_{k}\right|}{k} & < & 2\log M+2\log\left(2k\cdot\mincut\right)+4\tau^{\mathrm{cut}}\nonumber \\
 & \leq & 2\log M+4\log\left(2n\right)+4\tau^{\mathrm{cut}}.\label{eq:Cut-set-size-alpha-beta}
\end{eqnarray}
\end{lem}

\begin{proof}See Appendix \ref{sec:Proof-of-Lemma:cut-size-alpha-beta}.\end{proof}

We are now in position to characterize the recovery ability of $\psi_{\mathrm{ml}}$.
Let $\mathbb{P}_{\boldsymbol{w}}\left(\cdot\right)$ denote the measure
given $\boldsymbol{x}=\boldsymbol{w}$. For any $0<\alpha<1$, it
follows from (\ref{eq:Bound-Hel}) that 
\begin{eqnarray}
\mathbb{P}_{\boldsymbol{0}}\left\{ \log\frac{\mathrm{d}\mathbb{P}_{\boldsymbol{w}}\left(\boldsymbol{y}\right)}{\mathrm{d}\mathbb{P}_{\boldsymbol{0}}\left(\boldsymbol{y}\right)}>0\right\}  & \leq & \exp\left(-\left(1-\alpha\right)\frac{\sum_{i=0}^{M-1}e\left(\mathcal{S}_{i},\mathcal{S}_{i}^{\text{c}}\right)}{2}D_{\alpha}^{\min}\right),\label{eq:Bound-Hel-ERG}
\end{eqnarray}
When restricted to the hypotheses in $\mathcal{A}_{k}\backslash\mathcal{A}_{k-1}$
for any $2\leq k\leq n^{2}/\mincut$, we know from the definition
of $\mathcal{A}_{k}$ that 
\[
\left(k-1\right)\mincut\leq\left|\mathcal{E}\cap\mathrm{supp}\left(\boldsymbol{w}\ominus\boldsymbol{w}\right)\right|=\frac{\sum_{i=0}^{M-1}e\left(\mathcal{S}_{i},\mathcal{S}_{i}^{\text{c}}\right)}{2}<k\mincut.
\]
It then follows from the union bound that
\begin{align}
 & \mathbb{P}_{\boldsymbol{0}}\left\{ \exists\boldsymbol{w}\in\mathcal{A}_{k}\backslash\mathcal{A}_{k-1}:\text{ }\log\frac{\mathrm{d}\mathbb{P}_{\boldsymbol{w}}\left(\boldsymbol{y}\right)}{\mathrm{d}\mathbb{P}_{\boldsymbol{0}}\left(\boldsymbol{y}\right)}>0\right\} \leq\left|\mathcal{A}_{k}\right|\exp\left(-\left(1-\alpha\right)\frac{\sum_{i=0}^{M-1}e\left(\mathcal{S}_{i},\mathcal{S}_{i}^{\text{c}}\right)}{2}D_{\alpha}^{\min}\right)\nonumber \\
 & \quad\leq\exp\left(-\left(k-1\right)\left(\left(1-\alpha\right)D_{\alpha}^{\min}\cdot\mincut-\frac{k}{k-1}\cdot\frac{\log\left|\mathcal{A}_{k}\right|}{k}\right)\right)\nonumber \\
 & \quad\leq\exp\left(-\left(k-1\right)\left(\left(1-\alpha\right)D_{\alpha}^{\min}\cdot\mincut-\frac{2\log\left|\mathcal{A}_{k}\right|}{k}\right)\right)\label{eq:UB_Ak}\\
 & \quad\leq\exp\left\{ -\left(k-1\right)\left(\left(1-\alpha\right)D_{\alpha}^{\min}\cdot\mincut-\left(4\log M+8\log\left(2n\right)+8\tau^{\mathrm{cut}}\right)\right)\right\} ,\label{eq:inequality2}
\end{align}
where (\ref{eq:inequality2}) results from Lemma \ref{lemma:cut-size-alpha-beta}.
This suggests that under the condition
\[
\left(1-\alpha\right)D_{\alpha}^{\min}\cdot\mincut\geq\left(\delta+8\right)\log\left(2n\right)+8\tau^{\mathrm{cut}}+4\log M
\]
for some $0<\alpha<1$, one achieves
\begin{eqnarray*}
P_{\mathrm{e}}\left(\psi_{\mathrm{ml}}\right) & \leq & \sum_{k=2}^{n^{2}/\mincut}\mathbb{P}_{\boldsymbol{0}}\left\{ \exists\boldsymbol{w}\in\mathcal{A}_{k}\backslash\mathcal{A}_{k-1}:\text{ }\log\frac{\mathrm{d}\mathbb{P}_{\boldsymbol{w}}\left(\boldsymbol{y}\right)}{\mathrm{d}\mathbb{P}_{\boldsymbol{0}}\left(\boldsymbol{y}\right)}>0\right\} \\
 & \leq & \sum_{k\geq2}\exp\left\{ -\left(k-1\right)\left(\left(1-\alpha\right)D_{\alpha}^{\min}\cdot\mincut-\left(4\log M+8\log\left(2n\right)+8\tau^{\mathrm{cut}}\right)\right)\right\} \\
 & \leq & \sum_{k\geq1}\exp\left(-k\cdot\delta\log\left(2n\right)\right)\\
 & \leq & \frac{1}{\left(2n\right)^{\delta}}\cdot\frac{1}{1-\left(2n\right)^{-\delta}}=\frac{1}{\left(2n\right)^{\delta}-1}.
\end{eqnarray*}

To finish up, recognizing that $D_{\alpha}^{\min}\geq\mathsf{Hel}_{\alpha}^{\min}$
immediately establishes the recovery condition based on $\mathsf{Hel}_{\alpha}^{\min}$.

\section{Proof of Lemma \ref{lemma:alpha_random_graphs}\label{sec:Proof-of-Lemma:alpha_random_graphs}}

(1) Define the \emph{cut-edge degree }of a vertex $v$ to be the number
of edges in $\mathcal{E}\left(\mathcal{S},\mathcal{S}^{\mathrm{c}}\right)$
that $v$ is incident to. Consider any cut $\left(\mathcal{S},\mathcal{S}^{\mathrm{c}}\right)$
with size 
\begin{equation}
e\left(\mathcal{S},\mathcal{S}^{\mathrm{c}}\right)\leq k\cdot\mincut.\label{eq:cut-size-constraint}
\end{equation}
We shall separate all vertices into two types as follows:
\begin{itemize}
\item \emph{Type-1 vertex}: any vertex whose cut-edge degree is at least
$\frac{1}{2}\kappa\rho\cdot\mincut$;
\item \emph{Type-2 vertex}: any vertex whose cut-edge degree is less than
$\frac{1}{2}\kappa\rho\cdot\mincut$.
\end{itemize}
For ease of presentation, we will color all vertices in $\mathcal{S}$
black and all vertices in $\mathcal{S}^{\mathrm{c}}$ white; each
feasible coloring scheme thus corresponds to one valid cut $\left(\mathcal{S},\mathcal{S}^{\mathrm{c}}\right)$
in $\mathcal{N}\left(k\cdot\mincut\right)$. 

\begin{figure}

\centering

\includegraphics[width=0.38\textwidth]{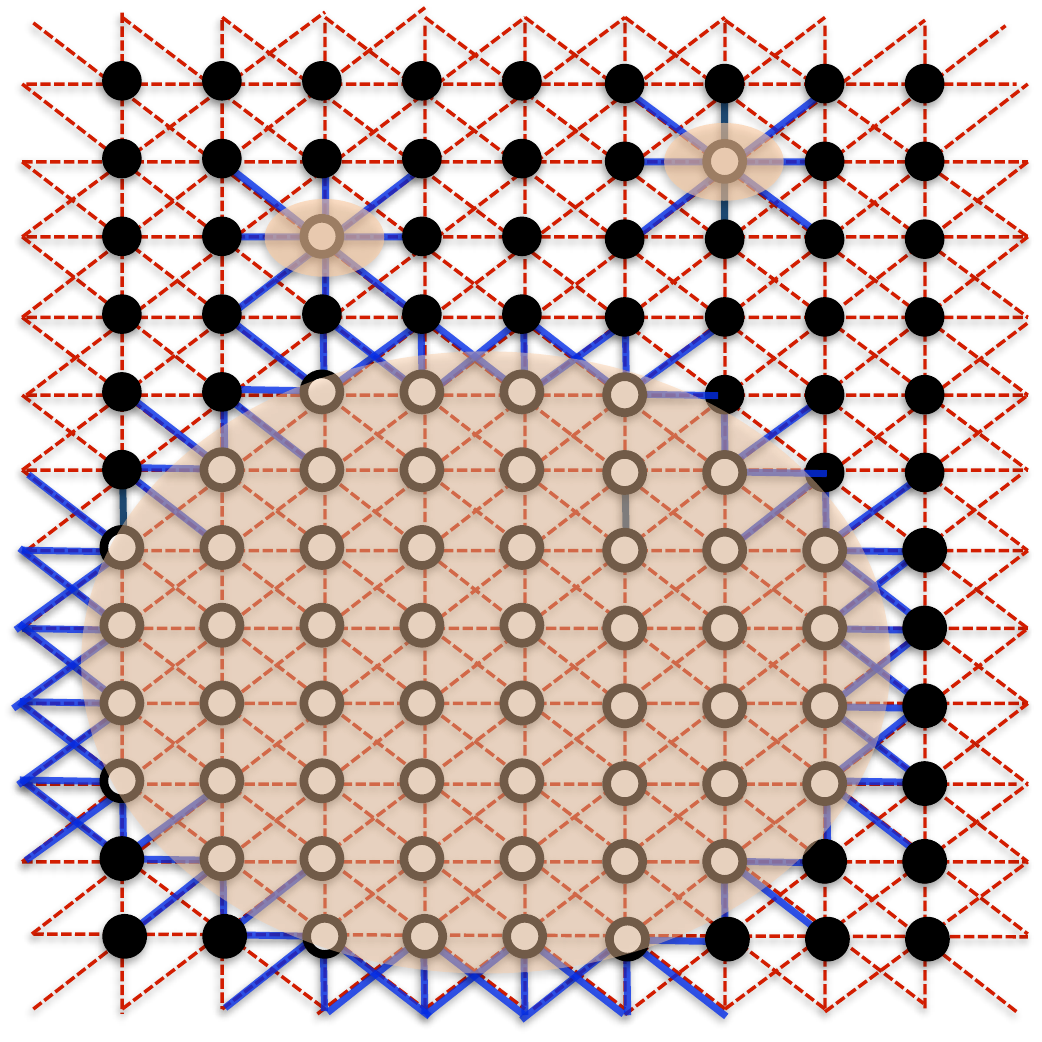}\caption{An example of the cut $(\mathcal{S},\mathcal{S}^{\mathrm{c}})$ in
a geometric graph. Here, $\mathcal{S}$ consists of all black vertices,
while $\mathcal{S}^{\mathrm{c}}$ contains all white vertices. The
blue solid edges represent the cut edges.}
\label{fig:geometric}
\end{figure}

To develop some intuitive understanding of the above notions, we depict
in Fig. \ref{fig:geometric} an example of a cut $\left(\mathcal{S},\mathcal{S}^{\mathrm{c}}\right)$
in a geometric graph, where $\mathcal{S}^{\mathrm{c}}$ consists of
all vertices residing within the shaded area, and the blue solid edges
indicate the cut edges. Typically, type-1 vertices, which are incident
to many cut edges, are lying on or close to the boundary of the cut.
In Fig. \ref{fig:geometric}, these correspond to those vertices lying
around the boundary of the shaded area in addition to those singleton
white vertices. In contrast, type-2 vertices often refer to those
staying away from the cut boundary (e.g. those white nodes in the
center of the shaded area). It may be useful to keep this figure in
mind when reading about the subsequent proof.

To prove Lemma \ref{lemma:alpha_random_graphs}, we start by examining
how many combinations of type-1 vertices are feasible and how many
ways there are to color them. By definition, for any cut obeying (\ref{eq:cut-size-constraint}),
the number $V_{1}$ of type-1 vertices satisfies $V_{1}\leq\frac{2e(\mathcal{S},\mathcal{S}^{\mathrm{c}})}{\frac{1}{2}\kappa\rho\mincut}\leq\frac{4k}{\kappa\rho}$
(note that each edge is incident to two vertices and might be counted
twice). Simple combinatorial arguments thus suggest that there are
at most $n^{4k/\kappa\rho}$ distinct ways to pick these type-1 vertices,
and then no more than $2^{V_{1}}\leq2^{4k/\kappa\rho}$ ways to color
all these type-1 vertices, and finally at most ${2e(\mathcal{S},\mathcal{S}^{\mathrm{c}}) \choose V_{1}}\leq\left(2k\cdot\mincut\right)^{V_{1}}\leq\left(2k\cdot\mincut\right)^{4k/\kappa\rho}$
different combinations of cut-edge degrees among them. Taken together
these counting arguments imply that there exist no more than\footnote{Here, we use the fact that $2k\cdot\mincut\leq n^{2}$, and hence
$\left(2k\cdot\mincut\right)^{4k/\kappa\rho}\leq n^{8k/\kappa\rho}$.} 
\[
n^{4k/\kappa\rho}\left(2k\cdot\mincut\right)^{4k/\kappa\rho}2^{4k/\kappa\rho}<\left(2n\right)^{8k/\kappa\rho}
\]
distinct ways to select the set of type-1 vertices as well as assign
colors and cut-edge degrees for each of them, if one is required to
satisfy the cut size constraint (\ref{eq:cut-size-constraint}).

We claim that for any cut $\left(\mathcal{S},\mathcal{S}^{\text{c}}\right)$
obeying (\ref{eq:cut-size-constraint}), once the following three
pieces of information are gathered:\begin{enumerate}[(i)]

\item which vertices are type-1 vertices,

\item  the cut-edge degrees of these type-1 vertices,

\item  the colors of these type-1 vertices (i.e. whether they belong
to $\mathcal{S}$ or $\mathcal{S}^{\mathrm{c}}$),

\end{enumerate}then the colors of all remaining vertices (and hence
all information about this cut) can be uniquely determined. Following
the preceding pictorial interpretation, the whole point of this claim
is to demonstrate that as long as some appropriate conditions regarding
the cut boundary is known, then one can figure out all remaining cut
information. To establish this claim, we shall consider the following
two cases separately. Without loss of generality, the following discussion
concentrates only on \emph{black} type-1 vertices. 
\begin{itemize}
\item \textbf{Case 1.} Consider any vertex $v$ whose color has been revealed
to be black, and whose cut-edge degree does \emph{not }exceed
\begin{equation}
\left(1-\frac{1}{2}\kappa\right)\rho\cdot\mincut,\label{eq:v-assumption}
\end{equation}
namely, $v$ is connected with no more than $\left(1-\frac{1}{2}\kappa\right)\rho\cdot\mincut$
white vertices. For any of its neighbors $u$ (i.e. $(u,v)\in\mathcal{E}$),
if the color of $u$ has not been revealed, then we claim that it
must be black. To see this, suppose instead that $u$ is white, then
from the above connectivity assumption (\ref{eq:v-assumption}) of
$v$, the number of black vertices that $u$ is linked with is at
least 
\begin{eqnarray*}
 &  & \left|\mathcal{V}\left(u\right)\cap\mathcal{V}\left(v\right)\right|-\left(1-\frac{1}{2}\kappa\right)\rho\cdot\mincut\\
 &  & \quad\geq\text{ }\rho\mincut-\left(1-\frac{1}{2}\kappa\right)\rho\cdot\mincut\text{ }=\text{ }\frac{1}{2}\kappa\rho\mincut,
\end{eqnarray*}
where the inequality follows from Assumption (\ref{eq:GeometricOverlap}).
This means that $u$ must be a type-1 vertex (cf. definition of type-1
vertices) and its color must have been revealed, thus resulting in
contradiction. In summary, all neighbors of $v$ with unknown colors
are necessarily black. 
\end{itemize}
\begin{itemize}[listparindent =2em] \item \textbf{Case 2. }Consider
any vertex $v$ whose color has been revealed to be black, and whose
cut-edge degree is known to be larger than $\left(1-\frac{1}{2}\kappa\right)\rho\cdot\mincut$.
Again, consider any of its neighbors $u$ whose color remains unknown,
which must be incident to fewer than $\frac{1}{2}\kappa\rho\cdot\mincut$
cut edges since by construction it is a type-2 vertex. This already
suggests the following fact: if there are at least $\frac{1}{2}\kappa\rho\cdot\mincut$
vertices falling in $\mathcal{\mathcal{V}}\left(u\right)\cap\mathcal{\mathcal{V}}\left(v\right)$
known to be white (resp. black), then the color of $u$ must be white
(resp. black), since by definition a type-2 vertex cannot be connected
to $\frac{1}{2}\kappa\rho\cdot\mincut$ vertices of opposite color.
As a result, we can uniquely determine the color of $u$ unless 
\begin{itemize}
\item (\textbf{P1}) the colors of fewer than $\kappa\rho\cdot\mincut$ vertices\footnote{Otherwise there are either $\frac{1}{2}\kappa\rho\mincut$ white vertices
or $\frac{1}{2}\kappa\rho\mincut$ black colors in $\mathcal{E}\left(u\right)\cap\mathcal{E}\left(v\right)$
with their colors revealed.} in $\mathcal{\mathcal{V}}\left(u\right)\cap\mathcal{\mathcal{V}}\left(v\right)$
have been revealed. 
\end{itemize}
This remaining situation is the subject of the discussion below. 

Suppose that the true color of $u$ is black. Recall that $u$ is
a type-2 vertex and hence it is connected to fewer than $\frac{1}{2}\kappa\rho$
white vertices. From Assumption (\ref{eq:GeometricOverlap}) and the
condition $\kappa<\frac{1}{2}$, any white neighbor $w$ of $u$ must
be connected with at least 
\[
\left|\mathcal{\mathcal{V}}\left(u\right)\cap\mathcal{\mathcal{V}}\left(w\right)\right|-\frac{1}{2}\kappa\rho\cdot\mincut\geq\left(1-\frac{1}{2}\kappa\right)\rho\mincut\geq\frac{1}{2}\kappa\rho\cdot\mincut
\]
black vertices falling within $\mathcal{\mathcal{V}}\left(u\right)\cap\mathcal{\mathcal{V}}\left(w\right)$,
and hence $w$ must be a type-1 vertex and its color has necessarily
been identified. Similarly, if $u$ is white, then the colors of \emph{all}
black vertices surrounding $u$ must have been revealed.\textbf{ }As
a result, all vertices in $\mathcal{V}(u)$ with unknown colors must
be of the same color as $u$. That being said, as long as one can
identify the color of one extra vertex in $\mathcal{V}\left(u\right)\cap\mathcal{V}\left(v\right)$,
then the color of $u$ and all remaining vertices in $\mathcal{V}\left(u\right)\cap\mathcal{V}\left(v\right)$
can be uniquely determined. 

Now let $w$ be the uncolored vertex in $\mathcal{V}\left(u\right)\cap\mathcal{V}\left(v\right)$
that is the nearest to $v$, which by (P1) must be within the $(\kappa\rho\mincut)$
closest vertices to $v$ in $\mathcal{V}\left(u\right)\cap\mathcal{V}\left(v\right)$.
From Assumption (\ref{eq:Closest}), we see that $w$ must be connected
to all but $\frac{1}{2}\rho\cdot\mincut$ neighbors surrounding $v$
and, as a result, be connected to at least
\begin{eqnarray*}
\text{cut-degree}(v)-\left|\mathcal{V}\left(v\right)\backslash\mathcal{V}\left(w\right)\right| & \geq & \left(1-\frac{1}{2}\kappa\right)\rho\cdot\mincut-\frac{1}{2}\rho\cdot\mincut\\
 & = & \frac{1}{2}\left(1-\kappa\right)\rho\cdot\mincut\text{ }\geq\text{ }\frac{1}{2}\kappa\rho\cdot\mincut.
\end{eqnarray*}
white vertices since $\kappa<\frac{1}{2}$, where $\text{cut-degree}(v)$
represents the cut-edge degree of $v$. Therefore, if $w$ is black,
then it has to be a type-1 vertex, which is contradictory, and we
have determined it to be white. \end{itemize}

Putting the above two cases together indicates that all vertices that
are connected to the set of type-1 vertices can be uniquely colored,
and we shall use $\mathcal{V}_{\text{new}}$ to denote them. If there
still exist uncolored vertices, a nonempty subset of them must be
connected to $\mathcal{V}_{\text{new}}$. Since all vertices in $\mathcal{V}_{\text{new}}$
are type-2 vertices and have cut-degrees not exceeding $\frac{1}{2}\kappa\rho\mincut\leq\left(1-\frac{1}{2}\kappa\right)\rho\mincut$,
repeating the arguments in Case 1 allows us to determine the color
of all vertices surrounding $\mathcal{V}_{\text{new}}$. This step
further shrinks the size of the uncolored set. Repeating this argument
until all vertices are colored, we establish the claim. All in all,
we have thus demonstrated that the number of feasible coloring schemes
is bounded above by $\left(2n\right)^{8k/\kappa\rho}$, which in turn
justifies 
\[
\tau_{k}^{\mathrm{cut}}\leq\frac{8\log\left(2n\right)}{\kappa\rho},\quad\forall k\geq1.
\]

(2) If $\mathcal{G}$ is an expander graph with edge expansion $h_{\mathcal{G}}$,
then for any vertex set $\mathcal{S}$ with $\left|\mathcal{S}\right|\leq\frac{n}{2}$,
one has
\begin{equation}
\left|\mathcal{S}\right|\leq e\left(\mathcal{S},\mathcal{S}^{\mathrm{c}}\right)/h_{\mathcal{G}}\label{eq:ExpanderBoundary}
\end{equation}
from the definition of $h_{\mathcal{G}}$. For any $d>0$, if one
requires that 
\begin{equation}
e\left(\mathcal{S},\mathcal{S}^{\mathrm{c}}\right)\leq kd,\label{eq:S-ub}
\end{equation}
then the above inequality leads to
\[
\left|\mathcal{S}\right|\leq kd/h_{\mathcal{G}},
\]
indicating that there are at most $2{n \choose \left\lfloor \frac{kd}{h_{\mathcal{G}}}\right\rfloor }\leq2n^{kd/h_{\mathcal{G}}}$
feasible cuts $\left(\mathcal{S},\mathcal{S}^{\mathrm{c}}\right)$
satisfying (\ref{eq:S-ub}). Setting $d=\mincut$ immediately leads
to
\[
\left|\mathcal{N}\left(k\cdot\mincut\right)\right|\leq2n^{k\mincut/h_{\mathcal{G}}},
\]
\[
\Rightarrow\quad\tau_{k}^{\mathrm{cut}}=\frac{1}{k}\log\left|\mathcal{N}\left(k\cdot\mincut\right)\right|\leq\frac{\mincut\log n}{h_{\mathcal{G}}}+\frac{\log2}{k},\quad\forall k\geq1
\]
as claimed.

\section{Proof of Lemma \ref{lemma-Outlier}\label{sec:Proof-of-Lemma-Outlier}}

We begin with explicit expressions of the divergence measures. For
any $k\neq l$ and $p\in\left[0,1\right]$, one has
\begin{eqnarray}
 &  & \mathsf{KL}\left(p\delta_{k}+(1-p)\text{Unif}_{M}\hspace{0.3em}\|\hspace{0.3em}p\delta_{l}+(1-p)\text{Unif}_{M}\right)\nonumber \\
 &  & \quad=\left(p+\frac{1-p}{M}\right)\log\left(\frac{p+\frac{1-p}{M}}{\frac{1-p}{M}}\right)+\frac{1-p}{M}\log\left(\frac{\frac{1-p}{M}}{p+\frac{1-p}{M}}\right)\label{eq:KLdivergence_1}\\
 &  & \quad=p\log\left(\frac{\left(M-1\right)p+1}{1-p}\right),\label{eq:KL_single_entry-general}
\end{eqnarray}
where $\delta_{k}$ denotes the Dirac measure on the point $k$, and
(\ref{eq:KLdivergence_1}) follows since the two distributions under
study differ only at two points $x=k$ and $x=l$. Similarly, one
obtains (cf. Definition (\ref{eq:Hellinger})) 
\begin{eqnarray}
 &  & \mathsf{Hel}_{\frac{1}{2}}\left(p\delta_{k}+(1-p)\text{Unif}_{M}\hspace{0.3em}\|\hspace{0.3em}p\delta_{l}+(1-p)\text{Unif}_{M}\right)\nonumber \\
 &  & \quad=2\left(\sqrt{p+\frac{1-p}{M}}-\sqrt{\frac{1-p}{M}}\right)^{2}=\frac{2}{M}\left(\sqrt{\left(M-1\right)p+1}-\sqrt{1-p}\right)^{2}.\label{eq:Hel_single_entry}
\end{eqnarray}
When applied to the outlier model, these suggest
\begin{equation}
\mathsf{KL}^{\min}=p_{\mathrm{true}}\log\left(1+\frac{p_{\mathrm{true}}M}{1-p_{\mathrm{true}}}\right)\leq\frac{p_{\mathrm{true}}^{2}M}{1-p_{\mathrm{true}}},\label{eq:KL-outlier-2}
\end{equation}
\begin{equation}
\text{and}\quad\mathsf{Hel}_{\frac{1}{2}}^{\min}=\frac{2}{M}\left(\sqrt{1-p_{\mathrm{true}}+Mp_{\mathrm{true}}}-\sqrt{1-p_{\mathrm{true}}}\right)^{2}.\label{eq:KL-outlier}
\end{equation}

It remains to control the Hellinger divergence. To this end, the elementary
identity $a-b=\frac{a^{2}-b^{2}}{a+b}$ gives 
\begin{eqnarray*}
 &  & \left(\sqrt{1-p_{\mathrm{true}}+Mp_{\mathrm{true}}}-\sqrt{1-p_{\mathrm{true}}}\right)^{2}=\left(\frac{p_{\mathrm{true}}M}{\sqrt{1-p_{\mathrm{true}}+Mp_{\mathrm{true}}}+\sqrt{1-p_{\mathrm{true}}}}\right)^{2}\\
 &  & \quad\geq\left(\frac{p_{\mathrm{true}}M}{2\sqrt{1-p_{\mathrm{true}}+Mp_{\mathrm{true}}}}\right)^{2}=\frac{p_{\mathrm{true}}^{2}M^{2}}{4\left(1-p_{\mathrm{true}}+Mp_{\mathrm{true}}\right)},
\end{eqnarray*}
indicating that $\mathsf{Hel}_{\frac{1}{2}}^{\min}\geq\frac{p_{\mathrm{true}}^{2}M}{2\left(1-p_{\mathrm{true}}+Mp_{\mathrm{true}}\right)}$
as claimed.

\section{Proof of Lemma \ref{lemma:cut-size-alpha-beta}\label{sec:Proof-of-Lemma:cut-size-alpha-beta}}

Consider any hypothesis $\boldsymbol{x}=\boldsymbol{w}\in\mathcal{A}_{k}$,
which obeys $\left|\mathcal{E}\cap\mathrm{supp}\left(\boldsymbol{w}\ominus\boldsymbol{w}\right)\right|<k\cdot\mincut$.
Denote by $\mathcal{S}_{l}$ the set of vertices that take the value
$l$ $(0\leq l<M)$, and let $\mathcal{I}_{\neg\emptyset}:=\left\{ l\mid\mathcal{S}_{l}\neq\emptyset\right\} $
represent the indices of those non-empty ones. Our proof proceeds
by evaluating the following quantities:
\begin{enumerate}
\item How many distinct choices of $\mathcal{I}_{\neg\emptyset}$ are admissible?
\item For each given $\mathcal{I}_{\neg\emptyset}$, how many combinations
of cut-set sizes $\left\{ e\left(\mathcal{S}_{l},\mathcal{S}_{l}^{\mathrm{c}}\right)\mid l\in\mathcal{I}_{\neg\emptyset}\right\} $
are feasible?
\item For each given cut-set size $N_{l}$, how many cuts $\left(\mathcal{S}_{l},\mathcal{S}_{l}^{\mathrm{c}}\right)$
are compatible with the constraint $e\left(\mathcal{S}_{l},\mathcal{S}_{l}^{\mathrm{c}}\right)\leq N_{l}$?
\end{enumerate}
Clearly, multiplying all these quantities together gives rise to an
upper bound on $\left|\mathcal{A}_{k}\right|$.

We now compute the above quantities separately. 
\begin{itemize}
\item To begin with, our assumption on the min-cut size ensures that 
\begin{equation}
e\left(\mathcal{S}_{l},\mathcal{S}_{l}^{\mathrm{c}}\right)\geq\mincut\label{eq:MinBoundary-Eq}
\end{equation}
for each non-empty $\mathcal{S}_{l}$. This together with the feasibility
constraint 
\begin{equation}
2\left|\mathcal{E}\cap\mathrm{supp}\left(\boldsymbol{w}\ominus\boldsymbol{w}\right)\right|=\sum_{l=0}^{M-1}e\left(\mathcal{S}_{l},\mathcal{S}_{l}^{\mathrm{c}}\right)\leq2k\cdot\mincut\label{eq:cut-size}
\end{equation}
guarantees that the number of non-empty $\mathcal{S}_{l}$'s cannot
exceed $2k$. Consequently, there exist at most $M^{2k}$ possible
combinations of $\mathcal{I}_{\neg\emptyset}$. 
\item Secondly, from (\ref{eq:cut-size}), the total cut-set size is bounded
above by $2k\cdot\mincut$. Therefore, for any given $\mathcal{I}_{\neg\emptyset}$,
there are no more than 
\[
{2k\mincut \choose \left|\mathcal{I}_{\neg\emptyset}\right|}\leq\left(2k\cdot\mincut\right)^{2k}
\]
 feasible ways to assign cut-set sizes $e\left(\mathcal{S}_{l},\mathcal{S}_{l}^{\mathrm{c}}\right)$
for all $l\in\mathcal{I}_{\neg\emptyset}$. 
\item Thirdly, suppose that for each $l\in\mathcal{I}_{\neg\emptyset}$,
\begin{equation}
e\left(\mathcal{S}_{l},\mathcal{S}_{l}^{\mathrm{c}}\right)=c_{l}\cdot\mincut\label{eq:Defn-cl}
\end{equation}
for some numerical values $c_{l}\geq1$. From the definition (\ref{eq:DefnExponent}),
the number of feasible choices of $\left(\mathcal{S}_{l},\mathcal{S}_{l}^{\mathrm{c}}\right)$
compatible with (\ref{eq:Defn-cl}) is bounded above by
\[
\left|\mathcal{N}\left(c_{l}\cdot\mincut\right)\right|\leq\left|\mathcal{N}\left(\left\lceil c_{l}\right\rceil \mincut\right)\right|\leq\exp\left(\left\lceil c_{l}\right\rceil \tau^{\mathrm{cut}}\right)\leq\exp\left(2c_{l}\tau^{\mathrm{cut}}\right).
\]
Recognize that the constraint (\ref{eq:cut-size}) requires
\[
\sum_{l}c_{l}<2k.
\]
As a result, when the cut sizes $e\left(\mathcal{S}_{l},\mathcal{S}_{l}^{\mathrm{c}}\right)$
are given, the total number of valid partitions $\left\{ \mathcal{S}_{l}\mid0\leq l<M\right\} $
cannot exceed
\begin{equation}
\prod_{l=0}^{M-1}\exp\left(2c_{l}\tau^{\mathrm{cut}}\right)<\exp\left(4k\tau^{\mathrm{cut}}\right).\label{eq:DifferentChoiceAlpha}
\end{equation}

\end{itemize}
Putting the above combinatorial bounds together implies that
\begin{eqnarray*}
\left|\mathcal{A}_{k}\right| & < & M^{2k}\left(2k\cdot\mincut\right)^{2k}\exp\left(4k\tau^{\mathrm{cut}}\right).
\end{eqnarray*}
Using the inequality $k\mincut\leq n^{2}$ we conclude the proof.

\section{Proof of Fact \ref{Fact:KL-Hellinger}\label{sec:Proof-of-Fact-KL-Hellinger}}

Recall that KL divergence and Hellinger divergence are both $f$-divergence
associated with the \emph{non-negative} convex functions $f_{1}\left(x\right)=x\log x-x+1$
and $f_{2}\left(x\right)=\left(\sqrt{x}-1\right)^{2}$, respectively.
That said, one can write 
\[
\mathsf{KL}\left(P\hspace{0.2em}\|\hspace{0.2em}Q\right)=\mathbb{E}_{Q}\left[f_{1}\left(\frac{\mathrm{d}P}{\mathrm{d}Q}\right)\right]\quad\text{and}\quad\mathsf{Hel}_{\frac{1}{2}}\left(P\hspace{0.2em}\|\hspace{0.2em}Q\right)=\mathbb{E}_{Q}\left[f_{2}\left(\frac{\mathrm{d}P}{\mathrm{d}Q}\right)\right].
\]
One can verify that the function $f_{1}$ can be uniformly bounded
above using $f_{2}$ in the following way:
\[
\left(2-0.5\left|\log x\right|\right)f_{2}\left(x\right)\leq f_{1}\left(x\right)\leq\left(2+\left|\log x\right|\right)f_{2}\left(x\right),\quad\forall x>0.
\]
This immediately establish that
\[
\mathsf{KL}\left(P\hspace{0.2em}\|\hspace{0.2em}Q\right)=\mathbb{E}_{Q}\left[f_{1}\left(\frac{\mathrm{d}P}{\mathrm{d}Q}\right)\right]\leq\left(2+\log R\right)\mathbb{E}_{Q}\left[f_{2}\left(\frac{\mathrm{d}P}{\mathrm{d}Q}\right)\right]=\left(2+\log R\right)\mathsf{Hel}_{\frac{1}{2}}\left(P\hspace{0.2em}\|\hspace{0.2em}Q\right)
\]
and
\[
\mathsf{KL}\left(P\hspace{0.2em}\|\hspace{0.2em}Q\right)=\mathbb{E}_{Q}\left[f_{1}\left(\frac{\mathrm{d}P}{\mathrm{d}Q}\right)\right]\geq\left(2-0.5\log R\right)\mathbb{E}_{Q}\left[f_{2}\left(\frac{\mathrm{d}P}{\mathrm{d}Q}\right)\right]=\left(2-0.5\log R\right)\mathsf{Hel}_{\frac{1}{2}}\left(P\hspace{0.2em}\|\hspace{0.2em}Q\right).
\]
These together with the well known inequality \cite[Lemma 2.4]{tsybakov2009introduction}
\[
\mathsf{KL}\left(P\hspace{0.2em}\|\hspace{0.2em}Q\right)\geq\mathsf{Hel}_{\frac{1}{2}}\left(P\hspace{0.2em}\|\hspace{0.2em}Q\right)
\]
establish (\ref{eq:KL-Hellinger}).

Similarly, from the inequality
\[
\left(2-0.4\left|\log x\right|\right)f_{2}\left(x\right)\leq f_{1}\left(x\right)\leq\left(2+0.4\left|\log x\right|\right)f_{2}\left(x\right),\quad\forall x\in(0,4.5],
\]
one can show that
\begin{equation}
\max\left\{ 2-0.4\log R,\text{ }1\right\} \cdot\mathsf{Hel}_{\frac{1}{2}}\left(P\hspace{0.2em}\|\hspace{0.2em}Q\right)\leq\mathsf{KL}\left(P\hspace{0.2em}\|\hspace{0.2em}Q\right)\leq\left(2+0.4\log R\right)\cdot\mathsf{Hel}_{\frac{1}{2}}\left(P\hspace{0.2em}\|\hspace{0.2em}Q\right)
\end{equation}
as long as $R\leq4.5$, as claimed.

\bibliographystyle{plain} \bibliographystyle{plain}
\bibliography{bibfileNetDecode}

\end{document}